\documentclass[a4paper,onecolumn,11pt,accepted=2026-03-02]{quantumarticle}
\pdfoutput=1

\usepackage{amsmath}
\usepackage{amssymb}
\usepackage{amsthm}
\usepackage{complexity}
\usepackage{enumitem}
\usepackage{graphicx}
\usepackage{mathtools}
\usepackage{physics}
\usepackage{stmaryrd}
\usepackage{silence}
\usepackage[caption=false]{subfig}
\usepackage{tabularray}
\usepackage{tikz}
\usepackage[dvipsnames]{xcolor}
\usepackage{xspace}
\usepackage[all,cmtip]{xy}

\graphicspath{{figures/}}

\usepackage{csquotes}
\MakeOuterQuote{"}

\SetSymbolFont{stmry}{bold}{U}{stmry}{m}{n}

\usepackage[numbers,sort&compress]{natbib}
\usepackage[colorlinks=true,allcolors=quantumviolet]{hyperref}
\usepackage[capitalize,noabbrev,nameinlink]{cleveref}
\crefname{corollary}{Cor.}{Cors.}
\crefname{definition}{Def.}{Defs.}
\crefname{equation}{Eq.}{Eqs.}
\crefname{fact}{Fact}{Facts}
\crefname{figure}{Fig.}{Figs.}
\crefname{lemma}{Lemma}{Lemmas}
\crefname{proposition}{Prop.}{Props.}
\crefname{section}{Sec.}{Secs.}
\crefname{theorem}{Thm.}{Thms.}
\crefname{enumi}{}{}

\newtheorem{corollary}{Corollary}
\newtheorem{definition}{Definition}
\newtheorem{fact}{Fact}
\newtheorem{lemma}{Lemma}
\newtheorem{proposition}{Proposition}

\newcommand{\field}{\mathbb{F}\xspace}
\newcommand{\identity}{\mathbb{I}\xspace}
\newcommand{\naturals}{\mathbb{N}\xspace}

\newcommand{\integers}{\mathbb{Z}\xspace}

\newcommand{\defsum}[3]{\sum_{#1}^{#2}{#3}}

\newcommand{\rs}[1]{\operatorname{rs}\lp #1 \rp\xspace}
\newcommand{\cs}[1]{\operatorname{cs}\lp #1 \rp\xspace}

\newcommand{\im}[1]{\operatorname{im}#1\xspace}

\newcommand{\wt}[1]{\text{wt\,}{#1}\xspace}
\newcommand{\stabcode}[3]{\llbracket #1, #2, #3 \rrbracket\xspace}
\newcommand{\supp}{{\mathrm{supp}}}
\newcommand{\bas}{{\mathrm{bas}}}
\newcommand{\spn}{{\mathrm{span}}}

\newcommand{\lp}{\left(}
\newcommand{\rp}{\right)}
\newcommand{\lb}{\left[}
\newcommand{\rb}{\right]}
\newcommand{\lc}{\left\{}
\newcommand{\rc}{\right\}}
\newcommand{\set}[1]{\langle #1 \rangle\xspace}
\newcommand{\numth}{\textsuperscript{th}}
\newcommand{\etal}{\emph{et al.}\xspace}
\newcommand{\<}{{\langle}}
\renewcommand{\>}{{\rangle}}

\newcommand{\bigOmega}[1]{\Omega\lp #1 \rp\xspace}
\newcommand{\bigTheta}[1]{\Theta\lp #1 \rp\xspace}

\newcommand{\revA}[1]{{#1}}

\begin{document}
\title{Qudit low-density parity-check codes}

\author{Daniel J. Spencer}
\affiliation{Joint Center for Quantum Information and Computer Science, NIST/University of Maryland, College Park, MD 20742, USA}
\affiliation{Joint Quantum Institute, NIST/University of Maryland, College Park, MD 20742, USA}
\affiliation{Department of Physics, University of Maryland, College Park, MD 20742, USA}
\orcid{0000-0003-1547-2935}
\email{djspence@umd.edu}

\author{Andrew Tanggara}
\affiliation{Centre for Quantum Technologies, National University of Singapore, 3 Science Drive 2, Singapore 117543}
\affiliation{Nanyang Quantum Hub, School of Physical and Mathematical Sciences, Nanyang Technological University, Singapore 639673}
\email{andrew.tanggara@gmail.com}

\author{Tobias Haug}
\affiliation{Quantum Research Center, Technology Innovation Institute, Abu Dhabi, UAE}
\orcid{0000-0003-2707-9962}
\email{tobias.haug@u.nus.edu}

\author{Derek Khu}
\affiliation{Institute for Infocomm Research (I\textsuperscript{2}R), Agency for Science, Technology and Research (A*STAR),
1 Fusionopolis Way, \#21-01, Connexis South Tower, Singapore 138632, Republic of Singapore}
\email{derek\_khu@a-star.edu.sg}

\author{Kishor Bharti}
\affiliation{Quantum Innovation Centre (Q.InC), Agency for Science, Technology and Research (A*STAR), 2 Fusionopolis Way, Innovis
\#08-03, Singapore 138634, Republic of Singapore}
\affiliation{Institute of High Performance Computing (IHPC), Agency for Science, Technology and Research (A*STAR), 1 Fusionopolis
Way, \#16-16 Connexis, Singapore 138632, Republic of Singapore}
\orcid{0000-0002-7776-6608}
\email{kishor.bharti1@gmail.com}

\maketitle

\begin{abstract}
    Qudits offer significant advantages over qubit-based architectures, including more efficient gate compilation, reduced resource requirements, improved error-correction primitives, and enhanced capabilities for quantum communication and cryptography. Yet, one of the most promising families of quantum error correction codes, namely quantum low-density parity-check (LDPC) codes, have so far been mostly restricted to qubits. Here, we generalize recent advancements in LDPC codes from qubits to qudits. We introduce a general framework for finding qudit LDPC codes and apply our formalism to several promising types of LDPC codes. We generalize bivariate bicycle codes, including their coprime variant; hypergraph product codes, including the recently proposed La-cross codes; subsystem hypergraph product (SHYPS) codes; high-dimensional expander codes, which make use of Ramanujan complexes; and fiber bundle codes. Using the qudit generalization formalism, we then numerically search for and decode several novel qudit codes compatible with near-term hardware. Our results highlight the potential of qudit LDPC codes as a versatile and hardware-compatible pathway toward scalable quantum error correction.
\end{abstract}

\newpage
\tableofcontents
\newpage

\section{Introduction}\label{sec:intro}
While most quantum algorithms and devices assume quantum \emph{bits} (qubits) as the fundamental informational unit, many physical systems allow for more than just two states. Qudits, the $q$-dimensional generalization of qubits, expand the Hilbert space available for computation~\cite{wang2017qudit,wang2020qudits,lu2020quantum,neeley2009emulation,blok2021quantum,chi2022programmable,low2020practical,nguyen2019quantum,nagata2020generalization,luo2014geometry,li2013geometry,luo2014universal,reimer2019high} and allow for an expanded model of digital quantum computation. The utilization of qudits in quantum computation has been demonstrated to offer distinct advantages over qubit-based architectures. Notably, qudits enable exponential improvements in the compilation of multi-control gates with respect to the number of control registers~\cite{gokhale2019asymptotic,kiktenko2020scalable}, an advantage directly applicable to the compilation of subroutines for Grover’s search~\cite{grover1996fastquantummechanicalalgorithm} and Shor’s factoring~\cite{shor1999polynomial} algorithms. Further advantages include a reduction in the number of qudits required for qubit-to-qudit circuit compilations~\cite{luo2014universal,luo2014geometry,wang2020qudits,kiktenko2023realizationquantumalgorithmsqudits,Nikolaeva_2024}, an overhead for fault-tolerant protocols comparable to qubit protocols~\cite{keppens2025qudit}, as well as enhanced yield and fidelity in magic-state distillation protocols, which asymptotically approach optimal performance as the qudit dimension increases~\cite{campbell2012magic,campbell2014enhanced}. In addition, the deployment of qudits has been shown to augment both the security and efficiency of quantum cryptographic schemes~\cite{bruss2002optimal,cerf2002security,durt2003security,durt2004security,bradler2016finite,bouchard2017high}, improve the noise resilience and security of quantum communication channels~\cite{cozzolino2019high}, and facilitate the simulation of high-dimensional quantum systems~\cite{sawaya2020resource,tacchino2021proposal,gonzalez2022hardware,meth2023simulating}. These demonstrated advantages have motivated experimental realizations of qudit-based systems across a variety of physical platforms, including superconducting~\cite{neeley2009emulation,blok2021quantum,morvan2021qutrit,roy2023two,goss2022high,fedorov2012implementation} photonic~\cite{lu2020quantum,wang2018multidimensional,chi2022programmable,reimer2019high}, trapped-ion~\cite{low2020practical,hrmo2023native,leupold2018sustained,ringbauer2022universal}, and circuit QED~\cite{brock2024quantumerrorcorrectionqudits} devices.

However, qudit-based quantum computers face some drawbacks compared to their qubit counterparts. For example, it can be more challenging to control qudit systems as more sophisticated techniques are required to transition qudits between more than two states, which can be harder to implement experimentally. One additional problem that qudit systems face that qubit systems also face is, of course, \emph{noise}. In fact, because qudits have more than two states, the error models are often more complex than those of qubits. The field of \emph{quantum error correction}~\cite{shor1995,steane1996,mackay1997,kitaev2003,resch2021} is one of the most promising paths to overcoming the problem of noise to realize large-scale quantum computers, evidenced by the extensive theoretical work that has resulted in promising quantum error correcting codes. Recently, there have even been several exciting experimental results demonstrating small-scale error-corrected quantum computation and early fault-tolerance~\cite{acharya2024,lacroix2024,eickbusch2024,rodriguez2024,bluvstein2024}.

Quantum low-density parity-check (LDPC) codes~\cite{kovalev2013} are a promising class of error correcting codes, with the surface code~\cite{bravyi1998,kitaev2003,fowler2009,fowler2012,zhao2022,google2023} being a notable example. Recently, a number of quantum LDPC codes have been introduced that show potential improvements over comparable surface codes, which offer more densely encoded logical information with a comparable noise-resilience. Examples include bivariate~\cite{bravyi2024} and multivariate~\cite{voss2024} bicycle codes, hypergraph product codes~\cite{tillich2013,pecorari2025}, and subsystem hypergraph product codes~\cite{li2020,malcolm2025}. The intuitive polynomial structure of these codes makes them easy to work with and provides a nice structure for searching for promising codes. In particular for bivariate bicycle codes, hypergraph product codes, and lifted product codes, fault-tolerant quantum computation architectures and protocols suitable for medium-sized devices have been proposed~\cite{yoder2025tourgrossmodularquantum,he2025extractors,cross2024improved,xu2024constant} and demonstrated to have a promising resource overhead scaling, some outperforming that of the surface code. Other quantum LDPC code constructions, such as high-dimensional expander codes~\cite{evra2020} and fiber bundle  codes~\cite{hastings2020}, are based on notions from algebraic topology and also show some promise.

In this work, we generalize these codes to the case where both the physical and the logical systems are of local dimension greater than 2 to give \emph{qudit LDPC codes}. We develop the theory for finding such codes and apply this formalism to the LDPC code constructions mentioned in the previous paragraph. For a subset of these codes, we numerically search for promising qudit LDPC codes and we propose several novel codes along with decoding results. We obtain qudit codes with parameters comparable to those of their qubit counterparts. Our results establish qudit LDPC codes as a promising framework for advancing scalable quantum error correction. We give a summary of our results and the current state of the qudit versions of various LDPC codes in~\cref{tab:qudit-code-summary}.

\begin{table*}[tb]
    \centering
    \begin{tblr}{colspec={|Q[c,m]|Q[c,m]|Q[c,m]|Q[c,m]|Q[c,m]|Q[c,m]|Q[c,m]|Q[c,m]|},row{even}={bg=lightgray}}
        \hline
        \textbf{Code type} & {\textbf{Original}\\\textbf{reference(s)}} & {\textbf{Quditization}\\\textbf{reference}}\\
        \hline[1pt]
        Bivariate bicycle & \cite{bravyi2024} & \cref{sec:qudit-bb-codes}\\
        \hline
        Coprime bivariate bicycle & \cite{wang2025,postema2025} & \cref{subsec:coprime-qudit-bb-codes}\\
        \hline
        Hypergraph product & \cite{tillich2013} & \cref{sec:qudit-hgp-codes}\\
        \hline
        La-cross & \cite{pecorari2025} & \cref{subsec:qudit-la-cross-codes}\\
        \hline
        {Subsystem hypergraph product\\ simplex (SHYPS)} & \cite{li2020,malcolm2025} & \cref{sec:qudit-shyps-codes}\\
        \hline
        Lifted product & \cite{panteleev2021a, panteleev2022} & \cite{panteleev2022}\\
        \hline
        {Two-block group-algebra\\ (2BGA)} & \cite{kovalev2013,wang2023,lin2023,lin2024quantum} & \cite{kovalev2013,wang2023,lin2023,lin2024quantum}\\
        \hline
        {High-dimensional\\ expander (HDX)} & \cite{evra2020} & \cref{sec:qudit-hdx-codes}\\
        \hline
        Fiber bundle & \cite{hastings2020} & \cref{sec:qudit-fiber-bundle-codes}\\
        \hline
    \end{tblr}
    \caption{Summary of previous results and the contributions of this paper in qudit LDPC codes. For prior results, often the codes were introduced in a general way (i.e., over $\field_q$) that does not require explicit quditization.}
    \label{tab:qudit-code-summary}
\end{table*}

The rest of this paper is organized as follows. In~\cref{sec:prelims}, we give some mathematical background needed to understand our results, where we give primers on finite field theory, ring theory, homological algebra, and coding theory. We also specify what we mean by the "quditization" of a qubit LDPC code. Then, the main part of the paper is a series of qudit generalizations of promising quantum LDPC codes. In each section, we give the formalism for generalizing the code from qubits to qudits and in some sections we give numerical results, introducing several novel qudit error correcting codes. We start with qudit bivariate bicycle codes in~\cref{sec:qudit-bb-codes}, including coprime bivariate bicycle codes, a recently proposed subclass. Then, we consider hypergraph product codes in~\cref{sec:qudit-hgp-codes}, including the La-cross codes, which are a specific instance of hypergraph product codes that are easy to work with numerically. In~\cref{sec:qudit-shyps-codes}, we generalize subsystem hypergraph product simplex codes to qudits and give some numerical results. For the last two code constructions, high-dimensional expander codes and fiber bundle codes in~\cref{sec:qudit-hdx-codes,sec:qudit-fiber-bundle-codes}, respectively, we eschew numerical results and only present the formalism for their qudit generalization. Finally, we give some concluding remarks in~\cref{sec:discussion}. In our Appendix, we provide a primer on qudit error correction in~\cref{app-sec:qudit-ec-primer}.

\section{Preliminaries}\label{sec:prelims}
We start by giving a brief overview of several mathematical tools needed to understand our results. In~\cref{subsec:fields-rings-primer}, we introduce some notions from abstract algebra, in particular finite field theory and ring theory. Then, in~\cref{subsec:homological-algebra-primer}, we give a brief overview of some notions from algebraic topology, namely chain complex homology, which is needed in the construction of hypergraph product codes, high-dimensional expander codes, and fiber bundle codes. Finally, in~\cref{subsec:coding-theory-primer}, we give a short introduction to some helpful results from classical coding theory. For a basic introduction to qudit error correction, see~\cref{app-sec:qudit-ec-primer}.

\subsection{Fields and rings}\label{subsec:fields-rings-primer}
We begin with finite field theory and ring theory~\cite{dummit2004}, which play important roles in most of the results in this work. In essence, moving from qubits to qudits means expanding the "alphabet" that we use for our codes. For qubits, we work over the binary field $\mathbb{F}_{2}$ (with elements 0 and 1, where $1 + 1 = 0$). For qudits, we use a larger field $\field_q$. This richer mathematical structure allows for more complex error models and potentially more efficient codes, but it also introduces subtleties in how we ensure that different types of parity checks commute. To understand these subtleties, we give a brief overview of several ideas from algebra here for the reader's convenience. Though many readers are likely familiar with these notions, we include them for completion and consistency of notation. We start by defining a group:

\begin{definition}[Group]\label{def:group}
    Given a set $G$ and binary operation $* : G \times G \to G$, we call $(G,*)$ a \emph{group} if all of the following conditions hold:
    \begin{enumerate}
        \item \textbf{Associativity}: $(a*b)*c = a*(b*c)$ for all $a,b,c \in G$
        \item \textbf{Identity}: There exists a unique \emph{identity} element $e \in G$ such that $e*a = a*e = a$ for all $a \in G$
        \item \textbf{Inverse}: For all $a \in G$, there exists a unique \emph{inverse} element $a^{-1}$ such that $a*a^{-1} = a^{-1}*a = e$.
    \end{enumerate}
\end{definition}

If we have for elements $a,b \in G$ that $a*b = b*a$, then we say that $a$ and $b$ \emph{commute}. If \emph{all} elements of the group $G$ commute under binary operation $*$, then we say that $(G,*)$ is an \emph{Abelian group}; otherwise, it is considered \emph{non-Abelian}. With the notion of a group in mind, we can now define a ring:

\begin{definition}[Ring]\label{def:ring}
    A \emph{ring}\footnote{Unless otherwise stated, all "rings" in this paper are commutative.} is an Abelian group under addition equipped with a second binary operation called multiplication, which is commutative, associative, and distributive over addition.
\end{definition}

In this work, we further assume that a ring has a multiplicative identity as well. One important notion from ring theory is that of an \emph{ideal} of a ring:

\begin{definition}[Ring ideal]\label{def:ring-ideal}
    An \emph{ideal} of a (commutative) ring is an Abelian group under addition that is closed under multiplication by ring elements on the left or right.
\end{definition}

We will find the notion of a \emph{principal} ideal particularly useful later on, so we define that here:

\begin{definition}[Principal ideal]\label{def:principal-ideal}
    A \emph{principal ideal} is an ideal of a ring that is generated by a single element of the ring (through its multiplication by every element of the ring).
\end{definition}

If $\iota \in \mathcal{R}$ is a principal ideal of a ring $\mathcal{R}$, we denote it using the notation $(\iota)$. With these notions defined from ring theory in hand, we can now define a field:

\begin{definition}[Field]\label{def:field}
    A \emph{field} is a ring containing the identity element in which every nonzero element has a multiplicative inverse.
\end{definition}

Finally, we define a finite (or Galois) field:

\begin{definition}[Finite field]\label{def:galois-field}
    A \emph{finite} or \emph{Galois field} is a field with a finite number of elements.
\end{definition}

We denote a finite field by $\field_q$, which has \emph{order} $q = p^s$ for some prime $p$, which we call the \emph{characteristic} of the field, and some $s \in \naturals$. Explicitly, $\field_q$ is a field that contains the $0$ element and $q-1$ other nonzero elements. If $q$ is prime (i.e., $s = 1$), then $\field_q$ is isomorphic to the integers modulo $q$: $\field_q \cong \integers/q\integers$. For example, for qubits, the field $\field_2$ contains the elements $\{0,1\}$ and is isomorphic to $\integers/2\integers$. When $s \geq 2$, this is not true and $\field_q$ is instead a (nontrivial) finite field \emph{extension} of $\field_p$, where we have $\field_q \cong \field_p[x] / \text{irr}(x)$ for an irreducible polynomial in $x$ of degree $s$. It is important to note that when $s \geq 2$, the field $\field_q$ is not simply isomorphic to the integers modulo $q$, $\mathbb{Z}/q\mathbb{Z}$. Instead, it is constructed using polynomials. This construction is necessary because $\mathbb{Z}/q\mathbb{Z}$ does not form a field when $q$ is not prime (e.g., in $\mathbb{Z}/4\mathbb{Z}$, 2 has no multiplicative inverse). The field extension construction ensures that every nonzero element remains invertible. If we denote the nonzero elements of $\field_q$ by $\field_q^\times$, then a primitive element $\omega \in \field_q$ is an element of $\field_q$ such that each element of $\field_q^\times$ can be written as $\omega^i$ for some $i \in \naturals$. If $n$ is a positive integer coprime to $q$ and $m$ is the smallest positive integer such that $q^m = 1 \pmod{n}$, the field $\field_{q^m}$ is the smallest extension field of $\field_q$ containing a primitive $n\numth$ root of unity. Fixing a primitive element $\omega$ of $\field_{q^m}$, a primitive $n\numth$ root of unity in $\field_{q^m}$ is given as
\begin{align}
    \beta = \omega^{\frac{q^m - 1}{n}}\,.
\end{align}
For several code constructions, we will be working with polynomials over finite fields, which we denote by $\field_q[x]$. In particular, we will find the notion of a \emph{primitive polynomial} useful. To understand primitive polynomials, though, we need to define the notion of an \emph{irreducible polynomial}:

\begin{definition}[Irreducible polynomial]\label{def:irreducible-polynomial}
    A polynomial $f(x) \in \field_q[x]$ is said to be \emph{irreducible} if it cannot be factored into two polynomials $g(x), h(x) \in \field_q[x]$ such that the degrees of both $g(x)$ and $h(x)$ are less than the degree of $f(x)$.
\end{definition}

Now, we define primitive polynomials:

\begin{definition}[Primitive polynomial]\label{def:primitive-polynomial}
    Let $\field_q$ be a finite field, where $q = p^s$ for prime $p$ and $s \in \naturals$. A degree-$s$ polynomial $f(x)$ over $\field_p$ is called \emph{primitive} if $f(x)$ is an irreducible polynomial in which the smallest integer $n$ for which $f(x)$ divides $x^n-1$ is $n=p^s-1$.
\end{definition}

\subsection{Homological algebra and chain complexes}\label{subsec:homological-algebra-primer}
We now give a basic introduction to homological algebra, a concept from algebraic topology; for a more detailed exposition, we refer the reader to~\cite{hatcher2001,nakahara2003}. Homological algebra provides a powerful, abstract language for describing quantum error-correcting codes, especially Calderbank--Shor--Steane (CSS) codes. A chain complex organizes the components of the system (the physical qudits and the stabilizer checks) into a sequence of vector spaces connected by ``boundary operators'' $\partial_j$. We start with a \emph{chain complex}:
\begin{align}
    \cdots \xrightarrow{\partial_{j+1}} \mathcal{A}_j \xrightarrow{\partial_j} \mathcal{A}_{j-1} \xrightarrow{\partial_{j-1}} \cdots \xrightarrow{\partial_1} \mathcal{A}_0 \xrightarrow{\partial_0 = 0} 0
\end{align}
for vector spaces $\mathcal{A}_j$ over $\field_q$ and linear maps $\partial_j \colon \mathcal{A}_j \to \mathcal{A}_{j-1}$, called \emph{boundary operators}, with the property $\partial_j\partial_{j+1} = 0$ (or equivalently, $\im \partial_{j+1} \subseteq \ker {\partial_j}$). The crucial property $\partial_{j}\partial_{j+1} = 0$ captures a fundamental topological idea: the boundary of a boundary is empty. For example, the boundary of a solid disk is a circle, but that circle itself has no boundary. In the context of quantum error correction, this condition is directly related to the requirement that $X$ and $Z$ stabilizers must commute. A \emph{$k$-complex} is a chain complex with $\mathcal{A}_j = 0$ for all $j > k$ and $\mathcal{A}_k \neq 0$. The basis elements of $\mathcal{A}_j$ are called \emph{$j$-cells} and vectors in $\mathcal{A}_j$ are called \emph{$j$-chains}. The Hamming weight of a chain is the number of cells in the chain with nonzero coefficients. The \emph{homology} of a chain complex $\mathcal{A}$ is a sequence of vector spaces
\begin{align}
    H_j(\mathcal{A}) = \ker {\partial_j} / \im \partial_{j+1}\,.
\end{align}
The $j\numth$ \emph{Betti number} is defined as $b_j(\mathcal{A}) = \dim_{\field_q}{H_j(\mathcal{A})}$. An element of $\ker{\partial_j}$ is called a \emph{$j$-cycle}. The homology groups $H_{j}(\mathcal{A})$ measure the presence of ``cycles'' that are not boundaries. In quantum codes, these groups correspond to the logical operators, and the dimension of the homology group (the Betti number $b_j$) typically tells us the number of encoded logical qudits $k$. The \emph{cohomology} $H^j(\mathcal{A})$ is the homology of the dual chain:
\begin{align}
    &\lp \cdots \xleftarrow{\partial^*_{j+1}} \mathcal{A}_j^* \xleftarrow{\partial_j^*} \mathcal{A}_{j-1}^* \xleftarrow{\partial_{j-1}^*} \cdots \xleftarrow{\partial_1^*} \mathcal{A}_0^* \xleftarrow{0} 0 \rp\\
    \cong &\lp \cdots \xleftarrow{\partial_{j+1}^\intercal} \mathcal{A}_j \xleftarrow{\partial_j^\intercal} \mathcal{A}_{j-1} \xleftarrow{\partial_{j-1}^\intercal} \cdots \xleftarrow{\partial_1^\intercal} \mathcal{A}_0 \xleftarrow{0} 0 \rp\,,
\end{align}
where $\mathcal{A}_j^*$ is the vector space of linear functionals on $\mathcal{A}_j$.

\subsubsection{Tensor product of chain complexes}\label{subsubsec:tensor-product-chain}
A standard construction in algebraic topology, which is also central to the construction of many quantum error-correcting codes (such as hypergraph product codes), is the tensor product of chain complexes. This allows us to build larger, more complex structures from simpler components. Suppose we have two chain complexes, $(A, \partial^A)$ and $(B, \partial^B)$. We define their tensor product $C = A \otimes B$. The vector spaces of the resulting complex $C$ are formed by combining the spaces of $A$ and $B$ such that their dimensions add up. Specifically, the space of $k$-chains, $C_k$, is defined as the direct sum of the tensor products of the constituent spaces:
\begin{align}
    C_k = \bigoplus_{i+j=k} A_i \otimes_{\field_q} B_j\,, \label{eq:tensor_product_spaces}
\end{align}
where the tensor product is taken over the field $\field_q$ underlying the vector spaces. (We will omit the subscript $\field_q$ in subsequent notation when the underlying field is obvious.) An element in $C_k$ is a linear combination of elements of the form $a \otimes b$, where $a \in A_i$ and $b \in B_j$ with $i + j = k$. The crucial step is defining the boundary operator $\partial^C_k : C_k \to C_{k-1}$ for the product complex. We want this operator to utilize the existing boundary operators $\partial^A$ and $\partial^B$, and it must satisfy the fundamental property of a chain complex: $(\partial^C)^2 = 0$. One might intuitively try to define the boundary operator analogously to the product rule in calculus:
\begin{equation}
    \partial_{\text{naive}}(a \otimes b) \stackrel{?}{=} (\partial^A a) \otimes b + a \otimes (\partial^B b).
\end{equation}
Let us check if this definition satisfies the required property. Applying the boundary operator twice:
\begin{align}
    \partial_{\text{naive}}^2(a \otimes b) &= \partial_{\text{naive}}\left((\partial^A a) \otimes b + a \otimes (\partial^B b)\right)\\
    &= \left((\partial^A)^2 a\right) \otimes b + (\partial^A a) \otimes (\partial^B b) + (\partial^A a) \otimes (\partial^B b) + a \otimes \left((\partial^B)^2 b\right)\,.
\end{align}
Since $A$ and $B$ are chain complexes, $(\partial^A)^2=0$ and $(\partial^B)^2=0$. The expression simplifies to:
\begin{equation}
    \partial_{\text{naive}}^2(a \otimes b) = 2 \left((\partial^A a) \otimes (\partial^B b)\right)\,.
\end{equation}
This expression is not necessarily zero unless the underlying field $\mathbb{F}_q$ has characteristic 2 (i.e., where $1+1=0$). For general fields, this naive definition fails to produce a valid chain complex. To ensure that the cross-terms cancel and $(\partial^C)^2=0$ holds regardless of the field characteristic, we must introduce a sign adjustment known as the \emph{Koszul sign rule}. The rule introduces a sign based on the dimension of the element that the second operator effectively ``jumps over.'' The correct boundary operator $\partial^C$ acting on an element $a \otimes b$, where $a \in A_i$, is thus defined as:
\begin{equation}
    \partial^C(a \otimes b) = (\partial^A a) \otimes b + (-1)^i a \otimes (\partial^B b)\,.\label{eq:koszul_boundary}
\end{equation}
We can now verify that this definition ensures $(\partial^C)^2=0$. Note that $\partial^A a \in A_{i-1}$.
\begin{align}
    (\partial^C)^2(a \otimes b) &= \partial^C\left((\partial^A a) \otimes b + (-1)^i a \otimes (\partial^B b)\right)\\
    \begin{split}
        &= (\partial^A)^2 a \otimes b + (-1)^{i-1} (\partial^A a) \otimes (\partial^B b)\\
        &\quad + (-1)^i (\partial^A a) \otimes (\partial^B b) + (-1)^i a \otimes (\partial^B)^2 b
    \end{split}\\
    &= 0 + \left((-1)^{i-1} + (-1)^i\right) (\partial^A a) \otimes (\partial^B b) + 0\\
    &= 0\,,
\end{align}
as desired.

\subsubsection{Tensor product of co-chain complexes}\label{subsubsec:tensor-product-co-chain}
The construction is entirely analogous for the tensor product of \emph{co-chain complexes}. Recall that a co-chain complex utilizes \emph{co-boundary operators}, which we denote by $\delta$. Unlike boundary operators, these maps increase the dimension. That is, $\delta_j: A^*_j \to A^*_{j+1}$, but they must still satisfy the fundamental property $\delta^2 = 0$. When working over a field, $\delta$ corresponds to the transpose of the boundary operator, $\partial^\intercal$. Suppose we have two co-chain complexes $(A^*, \delta^A)$ and $(B^*, \delta^B)$. We define their tensor product $C^* = A^* \otimes B^*$. The underlying vector spaces are defined similarly to the chain complex case in \cref{eq:tensor_product_spaces}:
\begin{align}
    C^*_k = \bigoplus_{i+j=k} A^*_i \otimes B^*_j\,.
\end{align}
We now seek a co-boundary operator $\delta^C_k: C^*_k \to C^*_{k+1}$. Just as before, we might try a naive definition analogous to the product rule:
\begin{align}
    \delta_{\text{naive}}(a \otimes b) \stackrel{?}{=} (\delta^A a) \otimes b + a \otimes (\delta^B b).
\end{align}
However, applying this operator twice yields the same problem encountered previously:
\begin{align}
    \delta_{\text{naive}}^2(a \otimes b) &= \delta_{\text{naive}}\left((\delta^A a) \otimes b + a \otimes (\delta^B b)\right)\\
    &= \lp (\delta^A)^2 a \rp \otimes b + (\delta^A a) \otimes (\delta^B b) + (\delta^A a) \otimes (\delta^B b) + a \otimes \lp (\delta^B)^2 b \rp\\
    &= 2 \lp (\delta^A a) \otimes (\delta^B b) \rp\,.
\end{align}
This is not generally zero unless the field has characteristic 2. The \emph{Koszul sign rule} is therefore also essential for the tensor product of co-chain complexes to ensure the resulting structure is a valid complex. For $a \in A^*_i$ and $b \in B^*_j$, the co-boundary operator $\delta^C$ is defined as:
\begin{align}
    \delta^C(a \otimes b) = (\delta^A a) \otimes b + (-1)^i a \otimes (\delta^B b)\,.\label{eq:koszul_co-boundary}
\end{align}
We can verify that this definition ensures that $(\delta^C)^2=0$. It is crucial to note how the degree changes: if $a \in A^*_i$, then $\delta^A a \in A^*_{i+1}$. When we apply $\delta^C$ a second time, the Koszul sign rule must account for this increased degree
\begin{align}
    (\delta^C)^2(a \otimes b) &= \delta^C\lp (\delta^A a) \otimes b + (-1)^i a \otimes (\delta^B b) \rp\\
    &= \delta^C\lp (\delta^A a) \otimes b \rp + (-1)^i \delta^C\lp a \otimes (\delta^B b) \rp\\
    \begin{split}
        &= \lp (\delta^A)^2 a \otimes b + (-1)^{i+1} (\delta^A a) \otimes (\delta^B b) \rp\\
        &\quad + (-1)^i \lp (\delta^A a) \otimes (\delta^B b) + (-1)^i a \otimes (\delta^B)^2 b \rp
    \end{split}\\
    &= 0 + \left((-1)^{i+1} + (-1)^i\right) (\delta^A a) \otimes (\delta^B b) + 0\\
    &= 0\,.
\end{align}
The introduction of the sign $(-1)^i$, adjusted appropriately when applying the operator sequentially (becoming $(-1)^{i+1}$ because the degree is increased), ensures the cancellation of the cross-terms, guaranteeing that the tensor product of two co-chain complexes is itself a valid co-chain complex.

\subsection{Coding theory}\label{subsec:coding-theory-primer}
Finally, we give a brief overview of coding theory. We start by defining a \emph{linear code}:

\begin{definition}[Linear code]\label{def:linear-code}
    A \emph{linear code} $C \subseteq \field_q^n$ is a linear subspace of the $n$-dimensional vector space $\field_q^n$. The elements $c$ of the code $C$ are called the \emph{codewords}.
\end{definition}

The cardinality of the code $C$ is given as $\abs{C} = q^k$, where $k$ is called the code \emph{dimension}. To generate the set of codewords, we need the \emph{generator matrix} $G^{k \times n}$, where $uG$ for all $u \in \field_q^k$ gives us the codewords. Its \emph{parity-check matrix} $H$ annihilates all codewords, namely $GH^\intercal = 0$. Given two codewords $c, c' \in \field_q^n$, the \emph{Hamming distance} is given by
\begin{align}
    d(c, c') = \abs{\{1 \leq i \leq n | c_i \neq c'_i\}}\,.
\end{align}
The \emph{Hamming weight} $w_H(c)$ of an element $c \in C$ is the number of nonzero elements of $c$. The minimum distance $d$ of $C$ is then the minimum Hamming distance between distinct codewords, which for linear codes is given by
\begin{align}
    d = \min\{w_H(c) | c \in C \backslash \{\boldsymbol{0}\}\}\,,
\end{align}
that is, the minimum Hamming weight of any nonzero codeword in the code $C$. The code $C$ can then correct $d-1$ erasures and $\lfloor \frac{d-1}{2} \rfloor$ errors.

The above formalism applies to both classical and quantum codes. We now focus our attention on just quantum codes. Given two classical codes $C_Z^\perp$ and $C_X$ such that $C_Z^\perp \subseteq C_X \subseteq \field_q^n$, we can use the Calderbank--Shor--Steane (CSS) construction~\cite{calderbank1996good,steane1996} to define the corresponding CSS code as
\begin{align}
    \text{CSS}(C_X, C_Z^\perp) &= \mathcal{Q}_X \oplus \mathcal{Q}_Z\\
    &= C_Z / C_X^\perp \oplus C_X / C_Z^\perp \,.
\end{align}
We denote this code using the familiar $\stabcode{n}{k}{d}_q$ notation, where $n$ is the number of physical qudits (the code length), $k$ is the number of logical qudits (the code dimension), and $d$ is the code distance. We denote the local qudit dimension with the subscript $q$. When $q=2$, we have a qubit code and we usually omit the subscript. The code dimension $k$ can be found as
\begin{align}
    k = n - \rank_{\field_q}{H_X} - \rank_{\field_q}{H_Z}
\end{align}
for $X$- and $Z$-parity check matrices $H_X$ and $H_Z$, respectively, where we explicitly denote that all operations are done over the field $\field_q$. The minimum distance is given as
\begin{align}
    d &= \min_{\substack{c \in C_Z/C_X^\perp\\ \text{ or } c \in C_X/C_Z^\perp}}w_H(c)\\
    &= \min\{d_X, d_Z\}\,.
\end{align}
A CSS code is considered a \emph{low-density parity-check} (LDPC) code if the row and column weights of its parity check matrices are bounded by constants $\Delta_{\text{col}}$ and $\Delta_{\text{row}}$. In this work, we mostly consider LDPC codes of weights 4, 5 and 6 as these are amenable to near-term hardware.

Finally, to reference the mathematical background given in the previous section, we can describe an $n$-qudit CSS code $\text{CSS}(C_X,C_Z^\perp)$ by a length-3 chain complex $C$ over the $\field_q$ vector space with boundary map $\partial$, which we write as
\begin{align}
    C = 0\xrightarrow{} C_2 \xrightarrow{\partial_2} C_1 \xrightarrow{\partial_1} C_0 \xrightarrow{} 0
\end{align}
where $C_2 = \field_q^{m_Z}$, $C_1 = \field_q^n$, and $C_0 = \field_q^{m_X}$, with boundary maps given by $\partial_2 = H_Z^\intercal$ and $\partial_1 = H_X$ which satisfy $\partial_1 \circ \partial_2 = 0$. This chain complex representation elegantly captures the structure of a CSS code. The physical qudits reside in the middle space $C_{1}$. The $X$-checks are defined by $\partial_{1} = H_{X}$, and the $Z$-checks are defined by $\partial_{2} = H_{Z}^{\intercal}$. Crucially, the homological requirement $\partial_{1}\circ\partial_{2}=0$ is precisely the CSS compatibility condition $H_{X}H_{Z}^{T}=0$. This will be useful when we generalize some of the codes that are derived from chain complexes, such as high-dimensional expander codes.

\subsection{Quditization}
Finally, before presenting our first qudit LDPC code, we specify what we mean by generalizing the error correction codes in this paper from qubits to qudits, a process we call \emph{quditization}. Intuitively (or physically), this is simple: instead of considering a two-level system, a qubit, we now consider a multilevel system, a qudit. Mathematically, this introduces some subtleties, which we deal with throughout the rest of the paper.

The starting point for quditization is the promotion of the field that we work over from $\field_2$ to $\field_q$, where $q$ is a prime power. This expands the space over which the mathematical objects that we work with are defined (i.e., vector spaces, polynomials, coefficients, chain complexes, etc), and it is this expanded space that we take advantage of to find new error correcting codes. One implication of working with $\field_q$ instead of $\field_2$ is a modified CSS condition, where, recall, the CSS condition is that, given parity-check matrices $H_X$ and $H_Z$, the following must hold:
\begin{align}
    H_X H_Z^\intercal = 0\,.\label{eqn:css-condition}
\end{align}
Since we mostly consider CSS codes in this work, this condition must still be true for the qudit codes we present, but the way in which this is achieved for some codes is slightly different. The easiest way to ensure that~\cref{eqn:css-condition} holds is to introduce coefficients at appropriate places, such as at the point of defining the parity-check matrices, as we will see in several codes later on.

Once the subtleties that arise when promoting codes from $\field_2$ to $\field_q$ are handled appropriately, the primary thing left to do is to find the code parameters, that is, find $\stabcode{n}{k}{d}_q$. After this is done, in this work, we consider the code "quditized," though of course for a fully-fledged qudit LDPC code, we must also provide a decoding algorithm, describe the logical operators and gates, etc. We leave these questions to future work.

\section{Qudit bivariate bicycle codes}\label{sec:qudit-bb-codes}
The first LDPC code that we quditize is the \emph{bivariate bicycle code}, which was introduced in the work by Bravyi \etal~\cite{bravyi2024}. We define a qudit bivariate bicycle code of local qudit dimension $q = p^s$ for prime $p$ and $s \geq 1$ by the polynomials
\begin{equation}
    \begin{aligned}
        A &= \alpha_1 A_1 + \alpha_2 A_2 + \alpha_3 A_3\,,\\
        B &= \beta_1 B_1 + \beta_2 B_2 + \beta_3 B_3\,,
    \end{aligned}
    \label{eqn:qudit-bb-AB-defs}
\end{equation}
where $\alpha_i, \beta_i \in \field_q$ and $A_i, B_i$ are powers of the commuting matrices $x = S_\ell \otimes I_m$ and $y = I_\ell \otimes S_m$. Here, $\ell, m \in \naturals$, $I_m$ is the $m \times m$ identity matrix and $S_\ell$ is the \emph{cyclic shift matrix} of size $\ell \times \ell$, which is defined such that the $i\numth$ row has a single nonzero element equal to 1 in column $(i+1) \mod{\ell}$. For example, the first two cyclic shift matrices are
\begin{align}
    S_2 =
    \begin{pmatrix}
        0 & 1\\
        1 & 0
    \end{pmatrix}
    \text{ and }
    S_3 =
    \begin{pmatrix}
        0 & 1 & 0\\
        0 & 0 & 1\\
        1 & 0 & 0
    \end{pmatrix}\,.
\end{align}
$A$ and $B$ are defined such that each $A_i$ is unique and each $B_i$ is unique; that is, $A_i \neq A_j$ and $B_i \neq B_j$ for $i \neq j$. The $t$-variate generalization introduced in~\cite{voss2024} is defined similarly, except each $A_i$ and $B_i$ is a power of one of $\{x_1, \ldots, x_t\}$; for example, for $t = 3$ (the trivariate case), we have $x \equiv x_1 = S_\ell \otimes I_m$, $x_2 \equiv y = I_\ell \otimes S_m$, and $x_3 \equiv z = S_\ell \otimes S_m$ and each $A_i$ and $B_i$ is a power of $x$, $y$, or $z$. Like the toric code, the physical qudits of a bivariate or multivariate bicycle code can be arranged on a two-dimensional grid with periodic boundary conditions, where the check operators can be obtained from a single pair of $X$- and $Z$-checks by applying horizontal and vertical shifts on the grid. Unlike the toric code, though, these check operators are not geometrically local.

We define $A$ and $B$ as we do in~\cref{eqn:qudit-bb-AB-defs} over the polynomial ring $\field_q[x,y]$. This is effectively generalizing the cyclic shift matrices used to define the qubit bivariate bicycle codes to $\field_q$. Alternatively, we could have quditized the qubit bivariate bicycle codes using the $\ell \times \ell$ \emph{circulant matrix} over $\field_q$, which is defined as
\begin{align}
    C_\ell =
    \begin{pmatrix}
        c_0 & c_1 & c_2 & \dots & c_{\ell-1} \\
        c_{\ell-1} & c_0 & c_1 & \dots & c_{\ell-2} \\
        c_{\ell-2} & c_{\ell-1} & c_0 & \dots & c_{\ell-3} \\
        \vdots & \vdots & \vdots & \ddots & \vdots \\
        c_1 & c_2 & c_3 & \dots & c_0
    \end{pmatrix}\,,
    \label{eqn:fq-circulant-matrix}
\end{align}
where $c_0, \dots, c_{\ell-1} \in \mathbb{F}_q$. Defining the polynomial ring $\field_q[S_\ell]$, where $S_\ell$ is the $\ell \times \ell$ cyclic shift matrix, $C_\ell$ is then just an element of this ring:
\begin{align}
    C_\ell = c_0 S_\ell^0 + c_1 S_\ell^1 + \dots + c_{\ell-1} S_\ell^{\ell-1}\,.
\end{align}
Circulant matrices have the nice property that any two circulant matrices commute, that is, $[C_\ell, C'_\ell] = 0$. These circulant matrices then define the variables $x$ and $y$ used to construct the polynomials $A$ and $B$. We note that the construction via the cyclic shift matrices and the construction via circulant matrices are equivalent, but for the rest of this analysis, we adopt the former.

We note a few properties of the objects that we have introduced. First, $x$ and $y$ must commute:
\begin{align}
    xy &= (S_\ell \otimes I_m)(I_\ell \otimes S_m) = S_\ell \otimes S_m\\
    yx &= (I_\ell \otimes S_m)(S_\ell \otimes I_m) = S_\ell \otimes S_m\\
    \implies xy &= yx\,.
\end{align}
We also have $x^\ell = y^m = I_{\ell m}$:
\begin{align}
    x^\ell &= (S_\ell \otimes I_m)^\ell\\
    &= S_\ell^\ell \otimes I_m^\ell\\
    &= I_\ell \otimes I_m\\
    &\equiv I_{\ell m}\,,
\end{align}
and likewise for $y^m$. Using the commutation relation of $x$ and $y$, we can easily check that $A$ and $B$, which comprise sums of powers of $x$ and $y$, also commute. Note that $A, B \in \field_q^{\ell m}$, so $A$ and $B$ define a qudit bivariate bicycle code with length $n = 2\ell m$. With all of this in hand, we define the parity-check matrices:
\begin{align}
    H_X = [\gamma_1 A | \gamma_2 B] \text{ and } H_Z = [\delta_1 B^\intercal | \delta_2 A^\intercal]\,,
\end{align}
where we have introduced a set of nonzero \emph{block coefficients} $\gamma_1, \gamma_2, \delta_1, \delta_2 \in \field_q$; for qubits, these coefficients are simply all 1. The CSS condition,
\begin{align}
    H_XH_Z^\intercal = 0\,,
\end{align}
must hold, which gives us a constraint on $\{\gamma_1, \gamma_2, \delta_1, \delta_2\}$:
\begin{align}
    H_XH_Z^\intercal &= [\gamma_1 A | \gamma_2 B]
        \lb
        \begin{array}{c}
            \delta_1 B\\
            \hline
            \delta_2 A
        \end{array}
        \rb\\
    &= \gamma_1\delta_1 AB + \gamma_2\delta_2 BA\\
    &= \gamma_1\delta_1 AB + \gamma_2\delta_2 AB\\
    &= (\gamma_1\delta_1 + \gamma_2\delta_2)AB\,,
\end{align}
where we used the fact that $A$ and $B$ commute. For this to be equal to 0 (if $AB \ne 0$), we must have $\gamma_1\delta_1 + \gamma_2\delta_2 = 0$, where arithmetic is done over the field $\field_q$. For any $q$, this is true if one of the coefficients is equal to $-1$ and the others are equal to 1 (or three are equal to $-1$ and one is equal to 1), such as $\gamma_1 = \gamma_2 = \delta_1 = 1$ and $\delta_2 = -1$, giving the following parity-check matrices for qutrit bivariate bicycle codes:
\begin{align}
    H_X = [A|B] \text{ and } H_Z = [B^\intercal|-A^\intercal]\,. \label{eqn:simplified_HX_HZ}
\end{align}
In general, any set of $\{\gamma_1, \gamma_2, \delta_1, \delta_2\}$ that satisfy $\gamma_1\delta_1 + \gamma_2\delta_2 = 0$ is a valid set of block coefficients. However, when searching for codes for a fixed $\{A_i, B_i\}_{i=1}^3$ and varying $\{\alpha_i, \beta_i\}_{i=1}^3$ and $\{\gamma_i, \delta_i\}_{i=1}^2$, we may without loss of generality restrict to the case above with $\gamma_1 = \gamma_2 = \delta_1 = 1$ and $\delta_2 = -1$. To see this, observe that scaling the parity check matrices does not change the code spaces, so we may replace $H_X = [\gamma_1 A | \gamma_2 B]$ and $H_Z = [\delta_1 B^\intercal | \delta_2 A^\intercal]$ with $H_X = [\gamma_1 \gamma_2^{-1} A | B]$ and $H_Z = [ B^\intercal | \delta_2 \delta_1^{-1}  A^\intercal]$. The condition $\gamma_1\delta_1 + \gamma_2\delta_2 = 0$ then implies that $\delta_2 \delta_1^{-1}  A^\intercal = -\gamma_1 \gamma_2^{-1} A^\intercal$. Now, replacing each $\alpha_i$ by $\gamma_1^{-1} \gamma_2 \alpha_i$ instead, we see that $A$ is replaced by $\gamma_1^{-1} \gamma_2 A$, which in turn means we have $H_X = [ A | B]$ and $H_Z = [ B^\intercal | -A^\intercal]$, as desired. In fact, swapping $A$ and $B$ only permutes the elements of the codewords and yields an equivalent code, so we may without loss of generality let $A_1$ correspond to the lowest-degree monomial (e.g., in lexicographic order) among the $A_i, B_j$ with nonzero coefficient. Since scaling $H_X$ and $H_Z$ does not change the code space, it is also possible to fix $\alpha_1 = 1$, leaving just the five variables $\alpha_2, \alpha_3, \beta_1, \beta_2, \beta_3$ and the choice of monomials $\{A_i, B_i\}_{i=1}^3$ with $A_1$ corresponding to the lowest-degree monomial.

We now formally present the code parameters of our qudit bivariate bicycle codes in the following theorem, which is the main result of this section. Here, we denote by $\ker(\cdot)$ the kernel and by $\rs{\cdot}$ the row space:

\begin{proposition}[Qudit multivariate bicycle code parameters]\label{prop:qudit-bb-code-parameters}
    The matrices $A, B \in \field_q^{\ell m \times \ell m}$, defined above, define a qudit bivariate bicycle code with parameters $\stabcode{n}{k}{d}_q$, where
    \begin{align}
        n &= 2\ell m\,,\\
        k &= 2 \cdot \dim\lp \ker(A) \cap \ker(B) \rp\,,\\
        d &= \min\{d_X, d_Z\}\,,
    \end{align}
    where $d_X$ and $d_Z$ are the code distances for $X$-type and $Z$-type errors defined as
    \begin{align}
        d_X &= \min\{\abs{v} : v \in \ker(H_Z) \backslash \rs{H_X}\}\,,\\
        d_Z &= \min\{\abs{v} : v \in \ker(H_X) \backslash \rs{H_Z}\}\,,
    \end{align}
    where $\abs{v}$ is the Hamming weight of a vector $v \in \field_q^n$, that is, the number of nonzero elements in $v$. Furthermore, $d_X = d_Z$, so the distance is simply $d = d_X = d_Z$.
\end{proposition}

\begin{proof}
    We can see that $n = 2\ell m$ by looking at the number of columns in $H_X$ and $H_Z$ (which both have the same dimension). Since $H_X = \lb \gamma_1 A | \gamma_2 B \rb$ and $A$ and $B$ are $\ell m\times \ell m$ matrices, $H_X$ has dimension $(\ell m) \times (2\ell m)$, so $n = 2\ell m$. Thus, $H_X$ and $H_Z$ are of dimension $\frac{n}{2} \times n$.

    To find $k$, we follow the proof for qubits in~\cite[Lemma 1]{bravyi2024}. We start by using a known result for CSS codes, which we presented in~\cref{subsec:coding-theory-primer}:
    \begin{align}
        k = n - \rank{H_X} - \rank{H_Z}\,,
    \end{align}
    We start by finding $\rank{H_Z}$:
    \begin{align}
        \rank{H_Z} &= \frac{n}{2} - \dim{\ker{H_Z^\intercal}}\\
        &= \frac{n}{2} - \dim{\ker{\lb
        \begin{array}{c}
            \delta_1 B\\
            \hline
            \delta_2 A
        \end{array}
        \rb}}\\
        &= \frac{n}{2} - \dim(\ker(A) \cap \ker(B))\,,
    \end{align}    
    where the first line follows from a combination of the relation $\rank{H} = \rank{H^\intercal}$ and the rank-nullity theorem. It turns out that $\rank{H_X} = \rank{H_Z}$, which we can see by the following transformation. Take $\Pi_\ell$ to be a self-inverse permutation matrix such that the $i\numth$ column of $\Pi_\ell$ has one nonzero entry equal to 1 (the multiplicative identity of $\field_q$) at row $j = -i \bmod{\ell}$. Define $\Pi_m$ in the same way and take $\Pi = \Pi_\ell \otimes \Pi_m$. Using these matrices, we have $\Pi_\ell S_\ell \Pi_\ell = S_\ell^\intercal$ and $\Pi_m S_m \Pi_m = S_m^\intercal$, which yields
    \begin{align}
        A^\intercal = \Pi A \Pi,\quad B^\intercal = \Pi B \Pi\,.\label{eqn:AB-transpose-permutation}
    \end{align}
    Thus, we can write
    \begin{align}
        H_Z &= [\delta_1 B^\intercal | \delta_2 A^\intercal] \\
        &= \Pi [\delta_2(\gamma_1\gamma_1^{-1})A | \delta_1(\gamma_2\gamma_2^{-1})B]
        \begin{bmatrix}
            0 & \Pi\\
            \Pi & 0
        \end{bmatrix} \\
        &= \Pi H_X
        \begin{bmatrix}
            0 & \delta_2\gamma_1^{-1}\Pi\\
            \delta_1\gamma_2^{-1}\Pi & 0
        \end{bmatrix}\,.
    \end{align}
    The matrix on the left hand side of $H_X$ has an inverse since the permutation matrix $\Pi$ is invertible, that is, $\Pi\Pi^{-1} = I_{\ell m}$, and both $\delta_2\gamma_1^{-1},\delta_1\gamma_2^{-1} \in \mathbb{F}_q$ have a multiplicative inverse in $\mathbb{F}_q$\footnote{This is because $\gamma_1,\delta_1,\gamma_2,\delta_2$ are nonzero by assumption, hence $\gamma_1,\gamma_2$ have multiplicative inverses $\gamma_1^{-1},\gamma_2^{-1}$ which are nonzero. Thus $\delta_2\gamma_1^{-1},\delta_1\gamma_2^{-1}$ are also nonzero and have nonzero multiplicative inverses.}.
    Denoting the multiplicative inverses of $\delta_2\gamma_1^{-1}$ and $\delta_1\gamma_2^{-1}$ as $\xi_2$ and $\xi_1$, respectively, we have
    \begin{align}
        \begin{bmatrix}
            0 & \delta_2\gamma_1^{-1}\Pi\\
            \delta_1\gamma_2^{-1}\Pi & 0
        \end{bmatrix}
        \begin{bmatrix}
            0 & \xi_1\Pi^{-1}\\
            \xi_2\Pi^{-1} & 0
        \end{bmatrix}
        &=
        \begin{bmatrix}
            \xi_2\delta_2\gamma_1^{-1} I_{\ell m} & 0\\
            0 & \xi_1\delta_1\gamma_2^{-1} I_{\ell m}
        \end{bmatrix}\\
        &= I_{2\ell m}\,.
    \end{align}
    Since we can obtain $H_X$ by multiplying $H_Z$ on the left and right by an invertible matrix, we conclude that $\rank{H_X} = \rank{H_Z}$. Using this, we have
    \begin{align}
        k &= n - \rank{H_X} - \rank{H_Z}\\
        &= n - 2\lp \frac{n}{2} - \dim\lp \ker(A) \cap \ker(B) \rp \rp\\
        &= 2 \cdot \dim\lp \ker(A) \cap \ker(B) \rp\,,
    \end{align}
    as stated in the theorem.

    We now prove the statements for the distance $d$. Recall, the Hamming weight $w_H$ of $v$ is the number of nonzero elements in $v$: $w_H(v) = \abs{\{i : v_i \neq 0\ \forall i \in [n]\}}$. Thus, to find $d$, we find the set of vectors $v$ that are in the space $\ker{H_Z} \backslash \rs{H_X}$ and $\ker{H_X} \backslash \rs{H_Z}$, take the vector of smallest Hamming weight for $d_X$ and $d_Z$, and take the minimum value of $d_X$ and $d_Z$.
    
    Finally, we prove that finding both $d_X$ and $d_Z$ and taking the minimum of the two is unnecessary as $d_X = d_Z$. Let $f \in \ker{H_Z} \backslash \rs{H_X}$ such that $|f| = d_X$ so that $X(f) = \prod_{j=1}^n X_j^{f_j}$ is a logical operator with minimum weight $d_X$. Then, there exists $Z(g) = \prod_{j=1}^n Z_j^{g_j}$ such that $f^\intercal g = a\neq 0$ and $g \in \ker{H_X} \backslash \rs{H_Z}$. 
    Here, $f$ and $g$ are length $n$ vectors in $\field_q^n$ which we take to be $f = (f_1, f_2)$ and $g = (g_1, g_2)$, where $f_i, g_i$ are length $n/2$ vectors for $i = 1,2$. The notation $(f_1, f_2)$ and $(g_1, g_2)$ follows from the two-block structure of bivariate bicycle codes. To have $f \in \ker{H_Z} \backslash \rs{H_X}$ and $g \in \ker{H_X} \backslash \rs{H_Z}$, we must have $H_Z f = 0$ and $H_X g = 0$, which is true if and only if the following holds:
    \begin{equation}\label{eqn:equiv-conditions-logicals}
        \begin{gathered}
            H_Z f = 0 \iff \delta_1 B^\intercal f_1 = -\delta_2 A^\intercal f_2 \\
            H_X g = 0 \iff \gamma_1 A g_1 = -\gamma_2 B g_2\,.
        \end{gathered}
    \end{equation}
    Now, for the $\Pi$ matrices defined above, we define
    \begin{align}
        h &= 
        \begin{bmatrix}
            \gamma_1^{-1}\gamma_2 \Pi f_2\\
            \delta_1\delta_2^{-1} \Pi f_1
        \end{bmatrix}\\
        \intertext{and}
        e &= 
        \begin{bmatrix}
            \Pi g_2\\
            \delta_1\delta_2^{-1}\gamma_1\gamma_2^{-1} \Pi g_1
        \end{bmatrix}\,.
    \end{align}
    We now show that $H_X h = 0$:
    \begin{align}
        H_X h &= [\gamma_1 A| \gamma_2 B] 
        \begin{bmatrix}
            \gamma_1^{-1}\gamma_2 \Pi f_2\\
            \delta_1\delta_2^{-1} \Pi f_1
        \end{bmatrix}\\
        &= \Pi (\gamma_1^{-1}\gamma_2\gamma_1 A^\intercal f_2 + \delta_1\delta_2^{-1}\gamma_2 B^\intercal f_1) \\
        &= \Pi (\gamma_1^{-1}\gamma_2\gamma_1 A^\intercal f_2 - \delta_1\delta_2^{-1}\gamma_2\delta_2\delta_1^{-1} A^\intercal f_2) \\
        &= (\gamma_1^{-1}\gamma_2\gamma_1 - \delta_1\delta_2^{-1}\gamma_2\delta_2\delta_1^{-1}) \Pi A^\intercal f_2 \\
        &= 0\,.
    \end{align}
    Likewise, $H_Z e = 0$:
    \begin{align}
        H_Z e &= [\delta_1 B^\intercal | \delta_2 A^\intercal]
        \begin{bmatrix}
            \Pi g_2\\
            \delta_1\delta_2^{-1}\gamma_1\gamma_2^{-1} \Pi g_1
        \end{bmatrix}\\
        &= \delta_1 B^\intercal\Pi g_2 + \delta_1\delta_2^{-1}\gamma_1\gamma_2^{-1} \delta_2 A^\intercal\Pi g_1\\
        &= \Pi\lp \delta_1 Bg_2 + \delta_1\delta_2^{-1}\gamma_1\gamma_2^{-1} \delta_2 Ag_1 \rp\,.
    \end{align}
    Using $\gamma_1 Ag_1 = -\gamma_2 Bg_2 \implies Ag_1 = -\gamma_1^{-1}\gamma_2 Bg_2$ from~\cref{eqn:equiv-conditions-logicals}, we have
    \begin{align}
        H_Z e &= \Pi\lp \delta_1 Bg_2 - \delta_1\delta_2^{-1}\gamma_1\gamma_2^{-1} \delta_2\gamma_1^{-1}\gamma_2 Bg_2 \rp\\
        &= \lp \delta_1 - \delta_1\delta_2^{-1}\gamma_1\gamma_2^{-1} \delta_2\gamma_1^{-1}\gamma_2 \rp\Pi Bg_2 \\
        &= 0\,,
    \end{align}
    Thus, we conclude that $h \in \ker{H_X}$ and $e \in \ker{H_Z}$. Furthermore, since $f^\intercal g = a$ for some $a \in \field_q$ and $a \neq 0$, we have
    \begin{align}
        h^\intercal e &= (\gamma_1\gamma_2^{-1})f_2^\intercal \Pi\Pi g_2 + (\delta_1\delta_2^{-1}\delta_1\delta_2^{-1}\gamma_1\gamma_2^{-1})f_1^\intercal \Pi\Pi g_1\\
        &= \gamma_1\gamma_2^{-1}(f_2^\intercal g_2 + f_1^\intercal g_1)\\
        &= \gamma_1\gamma_2^{-1}(f^\intercal g)\\
        &= a\gamma_1\gamma_2^{-1}\,,
    \end{align}
    which implies that $h \in \field_q^n \backslash \rs{H_Z}$ and $e \in \field_q^n \backslash \rs{H_X}$ because $a\gamma_1\gamma_2^{-1}$ is nonzero. Then, $X(e) = \prod_{j=1}^n X_j^{e_j}$ and $Z(h)=\prod_{j=1}^n Z_j^{h_j}$ correspond to logical operators with weights $|e|$ and $|h|$, respectively, since $h \in \ker{H_X} \backslash \rs{H_Z}$ and $e \in \ker{H_Z} \backslash \rs{H_X}$. Hence, by using the relation $d_Z \leq \abs{h}$, we conclude that
    \begin{align}
        d_Z &\leq \abs{h}\\
        &= \abs{\gamma_1^{-1}\gamma_2 \Pi f_2} + \abs{\delta_1\delta_2^{-1} \Pi f_1}\\
        &= \abs{f_2} + \abs{f_1}\\
        &= \abs{f}\\
        &= d_X\,,
    \end{align}
    Starting with $Z(g)=\prod_{j=1}^n Z_j^{g_j}$ such that $|g| = d_Z$, we can follow a similar argument to obtain a logical operator $X(e) = \prod_{j=1}^n X_j^{e_j}$ such that $|e| = |g| = d_Z$, which implies that $d_X \leq d_Z$. Thus we have $d_Z = d_X$, which we simply call $d$, concluding the proof.
\end{proof}

We conclude our discussion of qudit bivariate bicycle codes with a result that is helpful in looking for new codes in higher local qudit dimensions:

\begin{lemma}[Qudit bivariate bicycle codes via extension of scalars]\label{lemma:qudit-code-scalar-extension}
    Fix a prime power $q$ and an integer $t \ge 2$. Given an $\stabcode{n}{k}{d}_q$ qudit bivariate bicycle code over $\field_q$ with $H_X = [\gamma_1 A | \gamma_2 B]$ and $H_Z = [\delta_1 B^\intercal | \delta_2 A^\intercal]$ for some integers $\ell, m \geq 1$, matrices $A, B \in \field_q^{\ell m \times \ell m}$ and $\gamma_1, \gamma_2, \delta_1, \delta_2 \in \field_q^\times$, the qudit bivariate bicycle code over $\field_{q^t}$ defined by extension of scalars along $\field_q \to \field_{q^t}$, i.e., with the same $H_X$ and the same $H_Z$, has parameters $\stabcode{n}{k}{d}_{q^t}$.
\end{lemma}

\begin{proof}
    First, it is clear that the same $H_X$ and $H_Z$ gives rise to a valid bivariate bicycle code because the condition $\gamma_1 \delta_1 + \gamma_2 \delta_2 = 0$ continues to hold. Let the parameters of the new code over $\field_{q^t}$ be $\stabcode{n'}{k'}{d'}_{q^t}$. From~\cref{prop:qudit-bb-code-parameters}, the number of physical qubits does not change since $n' = 2\ell m = n$. Similarly, $k = 2 \dim (\ker(A) \cap \ker(B))$ where $\ker(A) \cap \ker(B)$ is the kernel of the matrix $\lb \begin{array}{c} A \\ \hline B \end{array}\rb$ when viewed as a matrix over $\field_q$, while $k'$ is defined similarly except the matrix is defined over $\field_{q^t}$. These quantities are equal since the dimension of the kernel of a matrix does not change under field extensions. One way to see this to note this is true for ranks, for example, by row reduction, and then apply the rank-nullity theorem\footnote{More abstractly, this invariance is a consequence of the fact that tensoring with a field extension is an exact functor, as field extensions are flat; thus, the kernel of the linear map represented by the matrix is preserved under base change.}.

    It remains to check that $d = \min\{\abs{v} : v \in \ker(H_Z) \backslash \rs{H_X}\}$ does not change when we consider $H_Z$ and $H_X$ as matrices over $\field_{q^t}$ (and their kernels and row spaces as subspaces of $\field_{q^t}^n$ instead of $\field_q^n$). To see this, choose $v \in \ker_{\field_q^n}(H_Z) \backslash \operatorname{rs}_{\field_q^n}(H_X)$ and $v' \in \ker_{\field_{q^t}^n}(H_Z) \backslash \operatorname{rs}_{\field_{q^t}^n}(H_X)$ achieving the minimum weights $d$ and $d'$ respectively. We can fix a basis $\{v_1, \cdots, v_b\}$ for $\ker_{\field_q^n}(H_Z)$ such that $\{v_1, \cdots, v_c\}$ is a basis for $\operatorname{rs}_{\field_q^n}(H_X)$ (for $c < b$). Note that these are also bases for $\ker_{\field_{q^t}^n}(H_Z)$ and $\operatorname{rs}_{\field_{q^t}^n}(H_X)$ respectively\footnote{This is because $\ker_{\field_{q^t}^n}(H_Z) = \ker_{\field_{q}^n}(H_Z) \otimes_{\field_q} \field_{q^t}$ and $\operatorname{rs}_{\field_{q^t}^n}(H_X) = \operatorname{rs}_{\field_q^n}(H_X) \otimes_{\field_q} \field_{q^t}$, which follow from the exactness of tensoring with a field extension.}. Since $v \in \ker_{\field_q^n}(H_Z) \backslash \operatorname{rs}_{\field_q^n}(H_X)$, $v$ is an $\field_q$-linear combination of $v_1, \ldots, v_b$, with a nonzero coefficient for some $v_i$ with $i > c$. Then, $v \in \ker_{\field_{q^t}^n}(H_Z) \backslash \operatorname{rs}_{\field_{q^t}^n}(H_X)$ as well, so its weight $d$ must be greater than or equal to $d'$, by definition of $d'$. Similarly, $v'$ is an $\field_{q^t}$-linear combination of $v_1, \ldots, v_b$, with a nonzero coefficient for some $v_j$ with $j > c$. We may assume without loss of generality (by rescaling $v'$ by a nonzero constant, which does not change $d'$) that the coefficient of $v_j$ has a nonzero field trace\footnote{The field trace of $y \in \field_{q^t}$ is the trace of the $\field_q$-linear transformation from $\field_{q^t}$ to $\field_{q^t}$ defined by $x \mapsto xy$. It is a fact that this is not the zero map, since $\field_{q^t}/\field_q$ is a separable field extension.}. Now, we can consider the vector $v'' \in \field_q^n$ formed by applying the field trace map to each coefficient of $v_1, \ldots, v_b$ in the expansion of $v'$. Since its coefficient of $v_j$ ($j > c$) is nonzero by assumption, $v'' \in \ker_{\field_q^n}(H_Z) \backslash \operatorname{rs}_{\field_q^n}(H_X)$. Thus, the weight of $v''$ is at least $d$. But the weight of $v''$ is also at most that of $v'$ (since the trace of the zero element is zero), so $d' \geq d$. Therefore, it must be the case that $d' = d$.
\end{proof}

There are a few useful implications of this result. First, it limits the search space for qudit bivariate bicycle codes in higher dimensions. That is, given a good code found over $\field_q$, this automatically provides a code for all powers of $q$. Another implication is that, for the same set of $n$, $k$, and $d$, a power $t$ of a code over $\field_q$ has more codewords, namely $q^{tk}$ codewords instead of $q^k$.

\subsection{Coprime qudit bivariate bicycle codes}\label{subsec:coprime-qudit-bb-codes}
Bivariate bicycle codes were further developed in~\cite{wang2025,postema2025}, where their algebraic properties were explored in more detail, resulting in \emph{coprime} bivariate bicycle codes. We extend this formalism to \emph{qudit coprime bivariate bicycle} codes and give analytic results for calculating the code dimension, avoiding some of the costly numerical search required for the more general qudit bivariate bicycle codes presented above.

Instead of performing the numerical search for $k$ as is done for general bivariate bicycle codes by calculating the rank of the parity-check matrices, we can get some analytic insight on $k$ by using the algebraic properties of these codes. Let $\ell$, $m$, and $q$ be mutually coprime and define $z = xy$, where $x$ and $y$ are defined as they were above. Then, the bivariate polynomials $A(x,y)$ and $B(x,y)$ can be expressed as univariate polynomials $A(z)$ and $B(z)$:

\begin{fact}[Univariate generator~{\cite[Thm. 2.2]{postema2025}}]\label{fact:qubit-bb-code-univariate-generator}
    Let $q = p^s$ for prime $p$ and $s \in \naturals$ and $\ell$, $m$, and $q$ all be mutually coprime. Then, bivariate bicycle codes allow for a univariate representation. Set $z = xy$. Since $\set{x}$ and $\set{y}$ are cyclic groups of order $\ell$ and $m$, respectively, $\set{xy}$ is a cyclic group of order $\ell m$. More concretely, there exists a bijective ring homomorphism $\psi$ such that
    \begin{align}
        \psi : \frac{\field_q[x,y]}{(x^\ell - 1, y^m - 1)} \to \frac{\field_q[z]}{(z^{\ell m} - 1)}\,.
    \end{align}
    This mapping is carried out by
    \begin{align}
        \psi : x \mapsto z^{m^{-1_\ell} m}, y \mapsto z^{\ell^{-1_m} \ell}\,,
    \end{align}
    where $m^{-1_\ell}$ is the multiplicative inverse of $m$ in $\integers / \ell\integers$ and $\ell^{-1_m}$ is defined similarly.
\end{fact}

For proof of this statement, we refer the reader to~\cite[Theorem 2.2]{postema2025}. Then, we have the following result:

\begin{proposition}[Qudit coprime bivariate bicycle code dimension]\label{prop:qudit-coprime-bb-dimension}
    Given $q = p^s$ for prime $p$ and $s \geq 1$, coprime integers $\ell$ and $m$ that are also coprime to $q$, and $z = xy$ for $x$ and $y$ defined above, let $h(z) = \gcd(A(z), B(z), z^{\ell m} - 1)$. Then, the bivariate bicycle code defined by $A(z)$ and $B(z)$ of length $n = 2\ell m$ has dimension
    \begin{align}
        k = 2\deg(h(z))\,.
    \end{align}
\end{proposition}

\begin{proof}
    The proof is essentially the same as the proof for the equivalent statement in~\cite{postema2025}. Let $x, y \in \field_q^{\ell m \times \ell m}$, and note that $\set{x}$ and $\set{y}$ are cyclic groups of orders $\ell$ and $m$, respectively. Define $z = xy$. Then, $\set{z}$ is a cyclic group of order $\ell m$, so any monomial in the set $\{x^i y^j | 0 \leq i \leq \ell, 0 \leq j \leq m\}$ can be expressed as a power of $z$. Then, any polynomial in $\field_q[x,y] / (x^\ell - 1, y^m - 1)$ can be expressed as one in $\field_q[z] / (z^{\ell m} - 1)$ under the isomorphism $\psi$ in~\cref{fact:qubit-bb-code-univariate-generator}. The column space of the parity-check operator $H_X = [\gamma_1 A | \gamma_2 B]$ is
    \begin{align}
        \cs{H_X} &= \{H_X x | x \in \field_q^{2\ell m}\}\\
        &= \{\gamma_1 Au + \gamma_2 Bv | u, v \in \field_q^{\ell m}\}\,,
    \end{align}
    which differs from the qubit case only by the introduction of the block coefficients $\gamma_1, \gamma_2 \in \field_q$, which we can drop by absorbing these into $u, v$ respectively. This can be written as
    \begin{align}
        \cs{H_X} = \{A(z)u(z) + B(z)v(z) | u(z), v(z) \in \field_q[z]/(z^{\ell m} - 1)\}\,.
    \end{align}
    Since $\field_q[z]/(z^{\ell m} - 1)$ is the quotient of a univariate polynomial ring, all its ideals are principal (see~\cref{def:ring-ideal}). Thus, $\cs{H_X}$, which is the ideal generated by $A(z)$ and $B(z)$, is a principal ideal and is generated by $h(z) = \gcd(A(z), B(z), z^{\ell m} - 1)$ and $\rank{H_X} = \dim(\cs{H_X}) = \ell m - \deg(h(z))$. The code dimension is then
    \begin{align}
        k &= 2\ell m - 2\rank{H_X}\\
        &= 2\ell m - 2(\ell m - \deg(h(z)))\\
        &= 2\deg(h(z))\,.
    \end{align}
\end{proof}

\subsection{Numerical code search}\label{subsec:qudit-bb-code-numerics}
We now detail our numerical search for novel qudit bivariate bicycle codes and present several novel codes in \cref{tab:new-qudit-bb-codes}. For the standard (non-coprime) qudit bivariate bicycle codes, we run a thorough numerical search for codes by choosing random values for the various input parameters: $\ell$, $m$, $A$ (both the coefficients $\alpha_i$ and the monomials $A_i$), and $B$ (both the coefficients $\beta_i$ and the monomials $B_i$). We choose the block coefficients $(\gamma_1, \gamma_2, \delta_1, \delta_2)$ as $(1,1,1,-1)$. Then, we construct the associated $X$- and $Z$-type parity-check matrices and numerically calculate the number of logical qudits and the distance, keeping only those codes that have a high rate and good distance ($d \geq 4$). To compute the distance, we  use an algorithm based on mixed-integer programming where the code is available on Github in Ref.~\cite{haug2025qudit}. Here, we find the logical operator with lowest weight, where the weight corresponds to the code distance. We present some candidate qudit bivariate bicycle codes of various weights in~\cref{tab:new-qudit-bb-codes}.

\begin{table*}[tb]
    \centering
    \begin{tblr}{colspec={|Q[c,m]|Q[c,m]|Q[c,m]|Q[c,m]|Q[c,m]|Q[c,m]|Q[c,m]|Q[c,m]|},row{even}={bg=lightgray}}
        \hline
        $\stabcode{n}{k}{d}_q^w$ & {Net encoding\\ rate $= k/2n$} & $(\ell,m)$ & Coprime? & $A$ & $B$\\
        \hline[1pt]
        $\stabcode{24}{4}{4}_3^5$ & $1/12$ & $(4,3)$ & No & $x + x^2$ & $x^3 + 2y + 2y^2$\\
        \hline
        $\stabcode{30}{4}{5}_3^5$ & $1/15$ & $(5,3)$ & No & $I + y + x^2y^2$ & $2I + x$\\
        \hline
        $\stabcode{48}{4}{7}_3^5$ & $1/24$ & $(8,3)$ & No & $I + 2x$ & $I + y + x^3y^2$\\
        \hline
        $\stabcode{88}{8}{5}_3^5$ & $1/22$ & $(4,11)$ & Yes & $2y^3 + y^8$ & $y + I + y^9$\\
        \hline
        $\stabcode{84}{6}{5}_5^4$ & $1/28$ & $(7,6)$ & Yes & $4y^3 + 4x$ & $3x^2 + 2x^6$\\
        \hline
        $\stabcode{30}{4}{5}_5^5$ & $1/15$ & $(5,3)$ & No & $I + y + x^2y^2$ & $2I + 3x$\\
        \hline
        $\stabcode{48}{4}{7}_5^5$ & $1/24$ & $(8,3)$ & Yes & $2I + 4x$ & $3I + 3y + x^3y^2$\\
        \hline
        $\stabcode{54}{6}{6}_5^6$ & $1/18$ & $(3,9)$ & No & $y^2 + 2y^7 + 2y^4$ & $3y^6 + 3y^3 + 4y^8$\\
        \hline
        $\stabcode{64}{8}{5}_5^6$ & $1/16$ & $(8,4)$ & No & $x^3 + 3x^5 + 4I$ & $3x + 3x^6 + 2x^7$\\
        \hline
        $\stabcode{28}{4}{5}_5^6$ & $1/14$ & $(7,2)$ & Yes & $3x^6 + 3x^2 + 4x^3$ & $3x^5 + x + I$\\
        \hline
        $\stabcode{30}{4}{5}_7^6$ & $1/15$ & $(3,5)$ & Yes & $4x + 2y + x^2$ & $y^4 + 2y^3 + 4y^2$\\
        \hline
    \end{tblr}
    \caption{New qudit bivariate bicycle codes of varying weights (superscript $w$) for $q = 3, 5$ (subscript). $I$ denotes the identity matrix (i.e., one of $x$ or $y$ raised to the power 0). We also delineate between regular and coprime codes (i.e., the mutual coprimeness of $\ell$, $m$, and $q$).}
    \label{tab:new-qudit-bb-codes}
\end{table*}

When considering the coprime qudit bivariate bicycle codes discussed in~\cref{subsec:coprime-qudit-bb-codes}, we can greatly reduce the search space. We restrict the values of $\ell$ and $m$ to those that are coprime with each other and pairwise coprime with $q$. Then, we check whether the choice of input will yield a high-rate code by algebraically calculating the code dimension $k$. Only in the case where the encoding rate is sufficiently high, we construct the actual parity-check matrix and compute the code distance.

\begin{figure*}[tb]
    \centering
    \subfloat[]{%
        \includegraphics[width=0.32\textwidth]{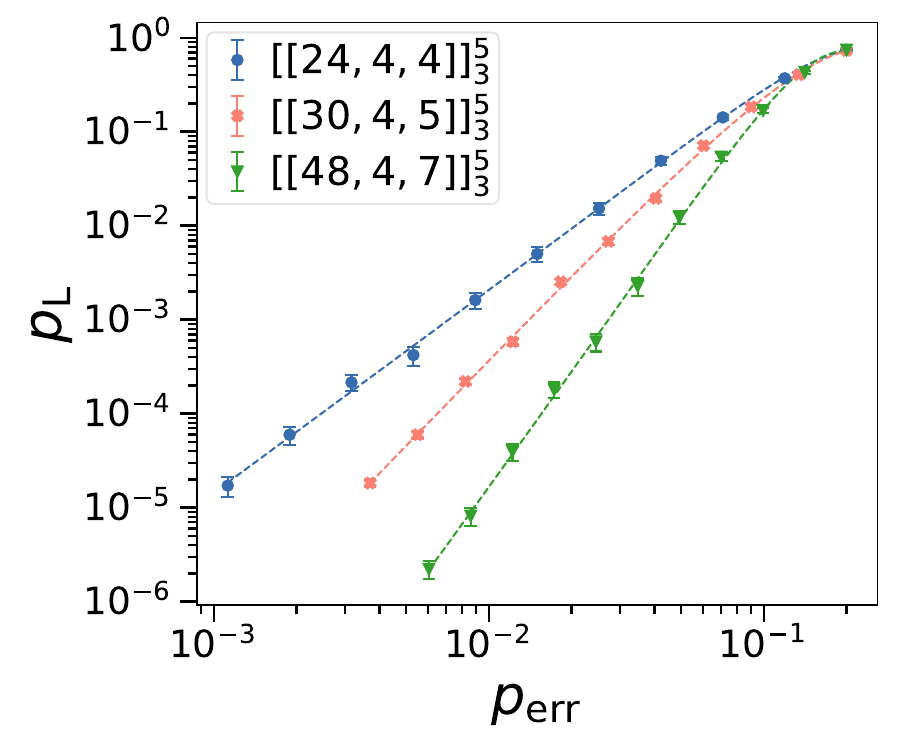}
    }
    \subfloat[]{%
        \includegraphics[width=0.32\textwidth]{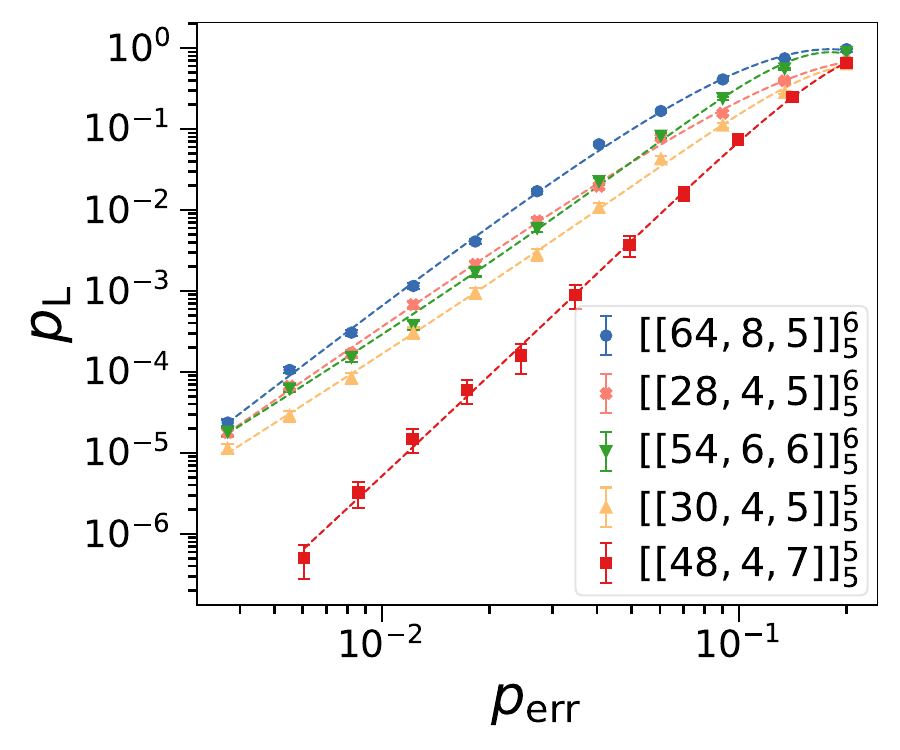}
    }
    \subfloat[]{%
        \includegraphics[width=0.32\textwidth]{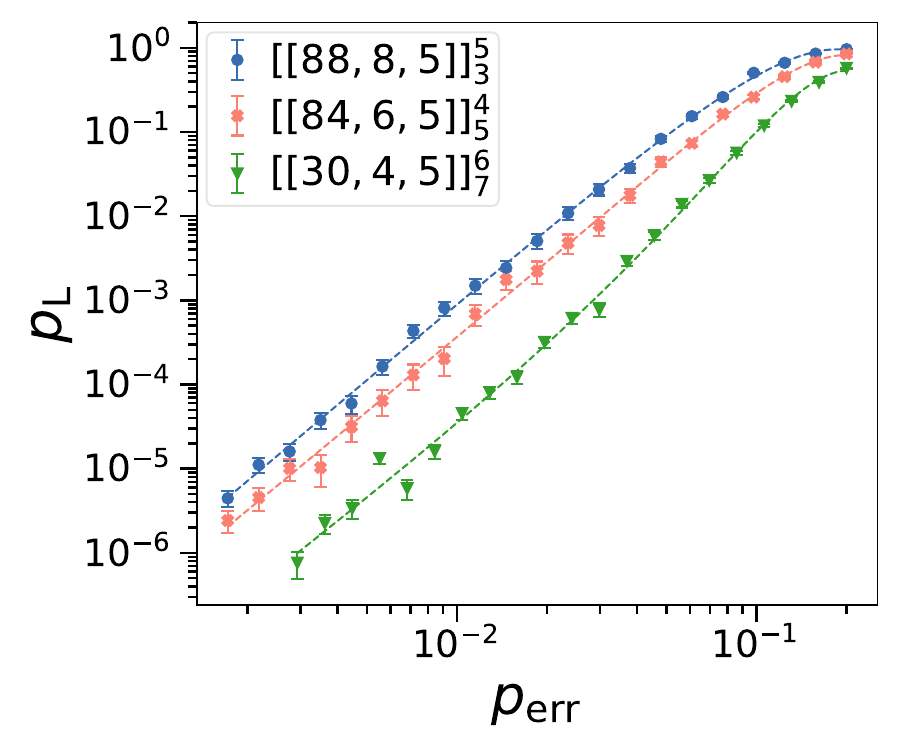}
    }
    \caption{Code capacity for the best qudit bivariate bicycle codes we have found. We plot logical error $p_\text{L}$ against physical error $p$. We show (a) regular, $q = 3$ with $d_\text{fit}=\{3.3, 5.1, 7.0\}$; (b) regular, $q = 5$ with $d_\text{fit}=\{5.8, 5.1, 4.4, 4.6, 7.1\}$; (c) coprime, $q = 3,5,7$ with $d_\text{fit}=\{5.0,4.8,4.5\}$.
    The subscript on each code in the legend corresponds to the local qudit dimension $q$ and the superscript denotes the weight of the parity check operators of the code. See~\cite[Fig. 3]{bravyi2024} and~\cite[Fig. 6]{wang2025} for comparable qubit bivariate bicycle codes.}
    \label{fig:bb-code-decoding}
\end{figure*}

Next, we numerically simulate the capability of our code to preserve quantum information. We simulate the code capacity, where we assume that the syndrome measurements are noise-free. As we consider CSS codes, we can restrict ourselves to $X$ errors, while $Z$ errors follow analogously. We assume that each qudit is subject to $X$ errors with probability $p_{\text{err}}$, where for each qudit there are $i=1, \ldots, q-1$ possible errors $X^i$, each appearing with probability $p_{\text{err}}/(q-1)$. After the qudits are subjected to errors, we measure the parity-checks to get their corresponding syndrome.

After this, we decode the syndrome to get the most likely error. There has been some work on decoding algorithms for qudits~\cite{bimberg2010,beermann2011,anwar2014fast,watson2014,geller2016}, but compared to the qubit case, there is a lack of open-source software for qudit decoding. As such, we wrote our own decoders based on mixed-integer programming, which searches for the optimal solution of the (non-degenerate) maximum-likelihood error problem~\cite{iyer2015}. After the decoding, we apply the correction operator and compute the probability $p_\text{L}$ that a logical error has occurred, that is, at least one of the logical $Z$ operator has flipped. We fit the logical error $p_\text{L}$ with the heuristic formula~\cite{bravyi2013,bravyi2024}
\begin{align}
    p_\text{L} = p_{\text{err}}^{(d_\text{fit}+1)/2 }\exp(c_0 + c_1p_{\text{err}} + c_2p_{\text{err}}^2)\label{eq:heuristic_fit}
\end{align}
where for qubit surface codes and bicycle codes it has been observed that the fitted distance corresponds to the actual distance~\cite{bravyi2013}. In~\cref{fig:bb-code-decoding}, we present the decoding results for the best qudit bivariate bicycle codes that we found, for both the regular and coprime cases. We find that the fitted error suppression $d_\text{fit}$ matches the distance $d$ of the codes up to its parity, i.e. $(d_\text{fit}+1)/2\approx \lceil d/2 \rceil$.

\section{Qudit hypergraph product codes}\label{sec:qudit-hgp-codes}
We now move to \emph{qudit hypergraph product} (HGP) codes~\cite{tillich2013,pecorari2025}. The HGP construction provides a powerful method to generate large quantum codes from two smaller, classical codes. Conceptually, it generalizes the product structure of the toric code. By taking the "product" of the parity-check matrices of the classical codes, we can obtain quantum LDPC codes with significantly better parameters (i.e., the encoding rate and the distance) than the surface code. There are two frameworks with which we can work: an algebraic approach similar to what we presented above for bivariate bicycle codes or a homological approach. For completion, we quditize HGP codes in both frameworks. We start with the algebraic framework, then quditize the HGP codes in the homological framework, and finally we conclude by studying in detail a particular type of HGP codes, the La-cross codes.

\subsection{Algebraic framework}\label{subsec:algebraic-framework}
We start with two classical linear LDPC codes $C_A$ and $C_B$ over $\field_q$, where, as usual, $q = p^s$ for prime $p$ and $s \in \naturals$. The codes have parameters $[ n_i, k_i, d_i]_q$ for $i = A,B$, where $n_i$, $k_i$, and $d_i$ are the code size, dimension, and distance, respectively. These codes are represented by their parity-check matrix $H_i \in \field_q^{n_i^\intercal \times n_i}$, where $n_i^\intercal := n_i - k_i$ is the number of checks. In terms of the parity-check matrix, the code is simply given as $C_i = \ker{H_i}$. For each of these classical codes, we have a corresponding \emph{transposed code}, denoted $C_i^\intercal$, with parameters $[ n_i^\intercal, k_i^\intercal, d_i^\intercal]_q$ with parity-check matrix $H_i^\intercal$. Thus, the transposed code is $C_i^\intercal = \ker{H_i^\intercal}$. The HGP construction then combines the classical linear codes along with their transposes to give an $\stabcode{n}{k}{d}_q$ quantum stabilizer code with quantum parity-check matrices
\begin{align}
    H_X &= [\gamma_1 H_A \otimes I_{n_B} | \gamma_2 I_{n_A^\intercal} \otimes H_B^\intercal]\,,\\
    H_Z &= [\delta_1 I_{n_A} \otimes H_B | \delta_2 H_A^\intercal \otimes I_{n_B^\intercal}]\,,
\end{align}
where we again introduce the block coefficients $\gamma_1, \gamma_2, \delta_1, \delta_2 \in \field_q$ to satisfy the CSS condition:
\begin{align}
    H_X H_Z^\intercal
    &= [\gamma_1 H_A \otimes I_{n_B} | \gamma_2 I_{n_A^\intercal} \otimes H_B^{\intercal}]
    \lb
    \begin{array}{c}
        \delta_1 I_{n_A} \otimes H_B^\intercal\\
        \hline
        \delta_2 H_A \otimes I_{n_B^\intercal}
    \end{array}
    \rb\\
    &= \gamma_1 \delta_1 H_A \otimes H_B^\intercal + \gamma_2\delta_2 H_A \otimes H_B^\intercal\\
    &= (\gamma_1 \delta_1 + \gamma_2\delta_2)H_A \otimes H_B^\intercal\,,
\end{align}
which must be equal to 0. This is trivially true in $\field_2$ for $\gamma_1 = \gamma_2 = \delta_1 = \delta_2 = 1$, but does not hold over $\field_q$ for odd $q$. Thus, to satisfy the CSS condition, we require
\begin{align}
    \gamma_1\delta_1 + \gamma_2\delta_2 = 0\,,
\end{align}
where arithmetic is done over $\field_q$. The easiest choice (which we made for the qudit bivariate bicycle codes) is to take $\gamma_1 = \gamma_2 = \delta_1 = 1$ and $\delta_2 = -1$, but any choice of block coefficients that satisfy $\gamma_1\delta_1 + \gamma_2\delta_2 = 0$ suffices\footnote{By a similar reasoning as the one presented in the paragraph after~\cref{eqn:simplified_HX_HZ}, it is in fact sufficient (for the purpose of searching for new codes) to restrict the search space to the special case $\gamma_1, \gamma_2, \delta_1 = 1$ and $\delta_2 = -1$.}. The number of stabilizers is $n_1 n_2^\intercal + n_2 n_1^\intercal$, which is the number of rows in $H$. The quantum code parameters are the same as they are for the qubit case~\cite{pecorari2025}, which are given as $n = n_An_B^\intercal+n_A^\intercal n_B$, $k = k_Ak_B^\intercal + k_A^\intercal k_B$, $d = \min\{d_A, d_B, d_A^\intercal, d_B^\intercal\}$, all of which we get from our choice of the classical, linear LDPC seed codes.

\subsection{Homology framework}\label{subsec:qudit-hgp-homology}
Alternatively, since an HGP code is a CSS code, it can be described using the language of chain complexes and homology~\cite{zeng2019,zeng2020}; see~\cref{subsec:coding-theory-primer} for a brief introduction to homology. In particular, we take the tensor product of two length-2 chain complexes describing a pair of classical linear codes. We again consider two classical codes $C_A$ and $C_B$ over $\field_q$ with parameters $\stabcode{n_A}{k_A}{d_A}_q$ and $\stabcode{n_B}{k_B}{d_B}_q$, respectively. We describe these codes in terms of chain complexes:
\begin{align}\label{eqn:classical-code-chain-complexes}
    A &: 0 \xrightarrow{\partial_2^A} A_1 \xrightarrow{\partial_1^A} A_0 \xrightarrow{\partial_0^A} 0\\
    B &: 0 \xrightarrow{\partial_2^B} B_1 \xrightarrow{\partial_1^B} B_0 \xrightarrow{\partial_0^B} 0\,.
\end{align}
where $A_1 \cong \field_q^{n_A}, A_0 \cong \field_q^{n_A^\intercal}$ and $B_1 \cong \field_q^{n_B}, B_0 \cong \field_q^{n_B^\intercal}$. Here, the boundary operators $\partial_1^A$ and $\partial_1^B$ can be taken as the parity-check matrices of $C_A = \ker{\partial_1^A}$ and $C_B = \ker{\partial_1^B}$, respectively. We also consider their transposed codes, defined by
\begin{align}\label{eqn:transposed-codes}
    C_A^\intercal &= \ker(\partial_1^A)^\intercal = \{c\in\field_q^{n_A^\intercal} : (\partial_1^A)^\intercal c = 0\}\,,\\
    C_B^\intercal &= \ker(\partial_1^B)^\intercal = \{c\in\field_q^{n_B^\intercal} : (\partial_1^B)^\intercal c = 0\}\,,
\end{align}
with parameters $\stabcode{n_A^\intercal}{k_A^\intercal}{d_A^\intercal}_q$ and $\stabcode{n_B^\intercal}{k_B^\intercal}{d_B^\intercal}_q$, respectively. The tensor product between chain complexes $A$ and $B$ is denoted by $A \otimes B$, where its $k\numth$ vector space is given by
\begin{align}
    (A \otimes B)_k = \bigoplus_{i,j:i+j=k} A_i \otimes B_j\,.
\end{align}
For $u \in A_i$ and $v \in B_j$, the boundary operator is given by
\begin{align}
    \partial_{i+j}^{A \otimes B}(u \otimes v) &= (\partial_i^A \otimes I_j^B)(u\otimes v) + (-1)^i (I_i^A \otimes \partial_j^B)(u \otimes v)\\
    &= \partial_i^A(u) \otimes v + (-1)^i u \otimes \partial_j^B(v)\,,
\end{align}
where $I_i^A$ and $I_j^B$ are identity operators over $A_i$ and $B_j$, respectively. Taking the chain complexes $A$ and $B$ in~\cref{eqn:classical-code-chain-complexes}, we obtain the tensor product chain complex $A \otimes B$ given by
\begin{align}
    0 &\xrightarrow{\partial_3^{A \otimes B}} A_1 \otimes B_1 \\
    &\xrightarrow{\partial_2^{A \otimes B}} (A_0 \otimes B_1) \oplus (A_1 \otimes B_0) \\
    &\xrightarrow{\partial_1^{A \otimes B}} A_0 \otimes B_0 \\
    &\xrightarrow{\partial_0^{A \otimes B}} 0\,,\label{eqn:hypergraph-code-chain-complex}
\end{align}
which gives us a hypergraph product code with $X$ and $Z$ parity check matrices
\begin{align}
    H_X &= \partial_1^{A \otimes B} = \lb I_0^A \otimes \partial_1^B | \partial_1^A \otimes I_0^B \rb\\
    \intertext{and}
    H_Z &= \lp \partial_2^{A \otimes B} \rp^\intercal = \lb \lp \partial_1^A \rp^\intercal \otimes I_1^B | -I_1^A \otimes \lp \partial_1^B \rp^\intercal \rb\,,
\end{align}
satisfying the CSS code condition
\begin{align} 
    H_X H_Z^\intercal &= (I_0^A \otimes \partial_1^B)(\partial_1^A \otimes I_1^B) - (\partial_1^A \otimes I_0^B)(I_1^A\otimes \partial_1^B)\\
    &= \partial_1^A\otimes \partial_1^B - \partial_1^A\otimes \partial_1^B\\
    &= 0\,.
\end{align}
The number of logical qudits $k$ encoded is given by the dimension of the first homology group of $A\otimes B$, that is, $H_1(A\otimes B) = \ker{\partial_1^{A\otimes B}} / \im{\partial_2^{A\otimes B}}$.
Note that this construction is the HGP construction in Section~\ref{subsec:algebraic-framework} with block coefficients given by $\gamma_1 = \gamma_2 = \delta_1 = 1$ and $\delta_2 = -1$.

\begin{proposition}[Qudit hypergraph product code parameters]\label{prop:qudit-hgp-dimension}
    The $\llbracket n,k,d\rrbracket_q$ parameters of a qudit hypergraph product code defined by parity check matrices $H_X,H_Z$ in~\cref{eqn:hypergraph-code-chain-complex} such that $\partial_1^A,\partial_1^B,(\partial_1^A)^\intercal,(\partial_1^B)^\intercal$ are not full-rank is given by $n = n_An_B^\intercal + n_A^\intercal n_B$, $k = k_Ak_B^\intercal + k_A^\intercal k_B$, and $d = \min\{d_A,d_A^\intercal,d_B,d_B^\intercal\}$.
\end{proposition}

\begin{proof}
    The number of data qudits of a CSS code is the dimension of the first vector space of the chain complex that defines it.
    For the qudit HGP code the first vector space is $(A_0 \otimes B_1) \oplus (A_1 \otimes B_0)$, which is a vector space over $\field_q$ of dimension $n_An_B^\intercal+n_A^\intercal n_B$.

    Now we derive the number of logical qudits. By the K\"unneth formula~\cite{maclane1995}, the $k\numth$ homology group of $A \otimes B$ is given by
    \begin{align}
        H_k(A \otimes B) = \bigoplus_{i,j:i+j=k} H_i(A) \otimes H_j(B)\,.
    \end{align}
    The number of logical qudits $k$ in the HGP code defined by $H_X, H_Z$ in~\cref{eqn:hypergraph-code-chain-complex} is given by the first homology group of $A \otimes B$ in~\cref{eqn:hypergraph-code-chain-complex}:
    \begin{align}\label{eqn:H1-tensor-product-complex}
        k &= \dim{H_1\lp A \otimes B \rp}\\
        &= \dim{H_1(A)}\dim{H_0(B)} + \dim{H_0(A)}\dim{H_1(B)}\\
        &= \dim{\ker{\partial_1^A} / \im{\partial_2^A}}\dim{\ker{\partial_0^B} / \im{\partial_1^B}} + \dim{\ker{\partial_0^A} / \im{\partial_1^A}}\dim{\ker{\partial_1^B} / \im{\partial_2^B}}\,,
    \end{align}
    where the second equality comes from the K\"unneth formula and the third equality comes from the definition of the first and zeroth homology groups of $A$ and $B$. Note that $\im{\partial_2^A} = \im{\partial_2^A} = \{0\}$, since the second vector space of $A$ and $B$ are both trivial. Thus, we have
    \begin{align}
        \dim{\ker{\partial_1^A} / \im{\partial_2^A}} &= \dim{\ker{\partial_1^A}}\\
        &= \dim{C_A}\\
        &= k_A \label{eqn:h1-dimensions-A} \\
        \intertext{and}
        \dim{\ker{\partial_1^B} / \im{\partial_2^B}} &= \dim{\ker{\partial_1^B}}\\
        &= \dim{C_B}\\
        &= k_B\,. \label{eqn:h1-dimensions-B}
    \end{align}
    Also, note that $\ker{\partial_0^A} = A_0$ and $\ker{\partial_0^B} = B_0$, since the $-1\textsuperscript{st}$ vector spaces of $A$ and $B$ are trivial. On the other hand, $\im{\partial_1^A} = \rs{\partial_1^A}^\intercal$ and $\im{\partial_1^B} = \rs{\partial_1^B}^\intercal$. At the same time, the transposed codes $C_A^\intercal$ and $C_B^\intercal$ are given by $\ker{\partial_1^A}^\intercal$ and $\ker{\partial_1^B}^\intercal$, respectively. Therefore, we have
    \begin{align}
        \dim{\ker{\partial_0^A} / \im{\partial_1^A}} &= \dim{A_0 / \rs{\partial_1^A}^\intercal}\\
        &= \dim A_0 - \dim \rs{\partial_1^A}^\intercal \\
        &= \dim A_0 - \dim \cs{\partial_1^A}^\intercal \\
        &= \dim \ker {\partial_1^A}^\intercal \\
        &= \dim{C_A^\intercal}\\
        &= k_A^\intercal \label{eqn:h0-dimensions-A}  \\ 
        \intertext{and}
        \dim{\ker{\partial_0^B} / \im{\partial_1^B}} &= \dim{B_0 / \rs{\partial_1^B}^\intercal}\\
        &= \dim{C_B^\intercal}\\
        &= k_B^\intercal\,. \label{eqn:h0-dimensions-B} 
    \end{align}
    Substituting~\cref{eqn:h1-dimensions-A,eqn:h1-dimensions-B,eqn:h0-dimensions-A,eqn:h0-dimensions-B} into the expression for $k$ in~\cref{eqn:H1-tensor-product-complex} gives $k = k_Ak_B^\intercal + k_A^\intercal k_B$.

    Finally, we derive the distance. Recall that the distance of a CSS code with parity check matrices $H_X$ and $H_Z$ with corresponding classical codes $C_X$ and $C_Z$ with parameters $[n_X,k_X,d_X]$ and $[n_Z,k_Z,d_Z]$, respectively, is given by $d = \min\{d_X,d_Z\}$ where 
    \begin{equation}\label{eqn:HGP_XZ_distances}
    \begin{aligned}
        d_X &= \min\{\abs{v} : v\in\ker \partial_1^{A\otimes B}\backslash\im \partial_2^{A\otimes B} \}\\
        d_Z &= \min\{\abs{v} : v\in\ker (\partial_2^{A\otimes B})^\intercal\backslash\im (\partial_1^{A\otimes B})^\intercal \} \,,
    \end{aligned}
    \end{equation}
    where $\partial_1^{A\otimes B}=H_X$ and $(\partial_2^{A\otimes B})^\intercal=H_Z$.
    The distances of the seed codes $C_A$ and $C_B$ and their transpose codes $C_A^\intercal$ and $C_B^\intercal$ are given by
    \begin{align}
        d_A &= \min\{\abs{v} : v\in\ker\partial_1^A \setminus\{0\} \}\\
        d_B &= \min\{\abs{v} : v\in\ker\partial_1^B\setminus\{0\} \}\\
        d_A^\intercal &= \min\{|v| : v\in\ker(\partial_1^A)^\intercal\setminus\{0\} \} \\
        d_B^\intercal &= \min\{|v| : v\in\ker(\partial_1^B)^\intercal\setminus\{0\} \} \,,
    \end{align}
    respectively. 

    Let $\bas(X)$ be the presumed basis of vector space $X$ (for vector spaces $A_j,B_j$ in the chain complex $A,B$, their bases $\bas(A_j),\bas(B_j)$ are determined by the parity check matrices $\partial_1^A,\partial_1^B$) and $\supp(v)$ be the support of vector $v$.
    Then consider
    \begin{equation}
    \begin{gathered}
        \mathcal{A}_1(p) = \{ a' \in \bas(A_1) \mid \exists b \in \bas(B_0) \text{ such that } a' \otimes b \in \supp(p) \} \,, \\
        \mathcal{B}_1(q) = \{ b' \in \bas(B_1) \mid \exists a \in \bas(A_0) \text{ such that } a \otimes b' \in \supp(q) \}
    \end{gathered}
    \end{equation}
    for $p\in A_1\otimes B_0$ and $q\in A_0\otimes B_1$.
    Now let $\partial_1^{A_1(p)}$ being the restriction of $\partial_1^A$ to $A_1(p):= \spn(\mathcal{A}_1(p))$ and $\partial_1^{B(q)}$ the restriction of $\partial_1^B$ to $B_1(q) := \spn(\mathcal{B}(q))$ so that for all $u\in A_1(p)$ and $v\in B_1(q)$ it holds that
    \begin{equation}\label{eqn:HGP_restriction_boundary_map}
        \partial_1^A(u) = \partial_1^{A(p)}(u) 
        \quad\text{and}\quad
        \partial_1^B(v) = \partial_1^{B(q)}(v) \,,
    \end{equation}
    which are the boundary maps of chain complexes
    \begin{equation}
        A(p): A_1(p) \xrightarrow{\partial_1^{A(p)}} A_0 
        \quad\text{and}\quad
        B(q): B_1(q) \xrightarrow{\partial_1^{B(q)}} B_0 \,,
    \end{equation}
    respectively. The product chain complex $C(p,q) := A(p)\otimes B(q)$ therefore is given by
    \begin{equation}
    \begin{gathered}
        C(p,q) : A_1(p)\otimes B_1(q) \rightarrow A_0 \otimes B_1(q) \oplus A_1(p)\otimes B_0 \rightarrow A_0\otimes B_0 \,,
    \end{gathered}
    \end{equation}
    so~\cref{eqn:HGP_restriction_boundary_map} implies that
    \begin{equation}
        \partial_1^A\otimes I_0^B(p) = \partial_1^{A(p)}\otimes I_0^B(p) 
        \quad\text{and}\quad
        I_0^A\otimes\partial_1^B(q) = I_0^A\otimes\partial_1^{B(q)}(q) \,.
    \end{equation}

    Now consider a $w\in (A\otimes B)_1 = (A_0\otimes B_1)\oplus(A_1\otimes B_0)$ and its restriction to $A_0\otimes B_1$ (resp. $A_1\otimes B_0$), denoted by $p$ (resp. $q$) so that $w=p\oplus q$.
    We show that if $w\in\ker\partial_1^{A\otimes B}$ and $|w|<\min\{d_A,d_B\}$, then $w\in\im\partial_2^{A\otimes B}$.
    Note that all $u\in A_1(p)$ has weight $|u|<d_A$ since $\dim A_1(p) = |\mathcal{A}_1(p)| \leq |w| < d_A$.
    Thus, because $|u'|\geq d_A$ for all $u'\in\ker\partial_1^A\setminus\{0\}$ it must hold that 
    \begin{equation}
        u\in A_1(p) \Rightarrow u\notin (\ker\partial_1^A\setminus\{0\}) \,.
    \end{equation}
    By a similar argument, all $v\in B_1(q)$ has weight $|v|<d_B$ since $\dim B_1(q) = |\mathcal{B}_1(q)| \leq |w| < d_B$, so that
    \begin{equation}
        v\in B_1(q) \Rightarrow v\notin(\ker\partial_1^B\setminus\{0\}) \,.
    \end{equation}
    Therefore the first homology group of $A(p)$ and $B(q)$ are trivial, i.e.
    \begin{equation}
        H_1(A(p)) = \{0\}
        \quad\text{and}\quad
        H_1(B(q)) = \{0\} \,,
    \end{equation}
    which implies that the first homology group of $C(p,q)$ is also trivial:
    \begin{equation}
    \begin{aligned}
        H_1(C(p,q)) = H_0(A(p)) \otimes H_1(B(q)) \oplus H_1(A(p)) \otimes H_0(B(q)) = \{0\} \,,
    \end{aligned}
    \end{equation}
    where the first equality is by the K\"unneth formula.
    
    Now, by the assumption that $w=p\oplus q\in\ker\partial_1^{A\otimes B}$ we have $w\in\ker\partial_1^{C(p,q)}$ because $w=p\oplus q \in C(p,q)_1 = A_0 \otimes B_1(q) \oplus A_1(p)\otimes B_0 \subseteq (A\otimes B)_1$.
    But because $H_1(C(p,q))=\{0\}$ and because $\im\partial_2^{C(p,q)}\subseteq\ker\partial_1^{C(p,q)}$ (due to $\partial_1^{C(p,q)}\partial_2^{C(p,q)}=0$) we have $\im\partial_2^{C(p,q)}=\ker\partial_1^{C(p,q)}$.
    So there exists some $w'\in C(p,q)_2 = A_1(p)\otimes B_1(Q)$ such that $\partial_2^{C(p,q)}(w') = w$.
    Since $A_1(p)\otimes B_1(Q) \subseteq A_1\otimes B_1 = (A\otimes B)_2$ therefore $w'\in (A\otimes B)_2$ such that $\partial_2^{A\otimes B}(w') = \partial_2^{C(p,q)}(w') = w$.
    Thus $w\in\im\partial_2^{A\otimes B}$.

    By the same line of argument, we can also show that if $w\in\ker(\partial_2^{A\otimes B})^\intercal$ and $|w|<\min\{d_A^\intercal,d_B^\intercal\}$, then $w\in\im(\partial_1^{A\otimes B})^\intercal$. Thus, combining both bounds we have
    \begin{equation}
        d = \min\{d_X,d_Z\} \geq \min\{d_A,d_A^\intercal,d_B,b_B^\intercal\} \,,
    \end{equation}
    where $d_X,d_Z$ are given in~\cref{eqn:HGP_XZ_distances}.

    Now we show that $d \leq \min\{d_A,d_A^\intercal,d_B,b_B^\intercal\}$. 
    We adopt the proof of Eq. (19) in~\cite{zeng2019}.
    For $j\in\{0,1\}$, let $\mathcal{A}_j$ and $\mathcal{B}_j$ be sets of basis for $H_j(A)$ and $H_j(B)$, respectively.
    Then by the K\"unneth formula
    \begin{align}
        H_1(A\otimes B) &= H_0(A)\otimes H_1(B) \oplus H_1(A)\otimes H_0(B) \\
        &= (A_0/\im\partial_1^A)\otimes(\ker\partial_1^B/\{0\}) \oplus (\ker\partial_1^A/\{0\})\otimes(B_0/\im\partial_1^B) \,.
    \end{align}
    $\mathcal{A}_0\otimes\mathcal{B}_1 \cup \mathcal{A}_1\otimes\mathcal{B}_0$ is a basis of $H_1(A\otimes B)$. Suppose that $\mathcal{A}_1$ (resp. $\mathcal{B}_1$) contains a class $[u]\in H_1(A)$ where $u\in \ker\partial_1^A$ is a minimum weight codeword i.e. $|u|=d_A$ (resp. contains class $[v]\in H_1(B)$ such that $v\in \ker\partial_1^B$ and $|v|=d_B$). Note also that for a basis element $e\in A_0\setminus\im\partial_1^A$ (which always exist since $\partial_1^A$ is not full-rank), the zeroth homology group $H_0(A)=A_0/\im\partial_1^A$ contains a class $[e]=\{e+f : f\in\im\partial_1^A\}$. Since $[e]$ contains $e$, which is a basis element (of weight 1), therefore $\Tilde{d}_A := \min\{|t| : [t]\in H_0(A)\} = 1$. By a similar argument, we also have $\Tilde{d}_B := \min\{|t| : [t]\in H_0(B)\} = 1$. Combining these facts with $\mathcal{A}_0\otimes\mathcal{B}_1 \cup \mathcal{A}_1\otimes\mathcal{B}_0$ being a basis of $H_1(A\otimes B)$, we have
    \begin{align}
        d_X \leq \min \{\Tilde{d}_A d_B , d_A \Tilde{d}_B\} = \min\{d_B , d_A\}\,.
    \end{align}
    By considering cohomology group 
    \begin{align}
        H^1(A\otimes B) &= H^0(A)\otimes H^1(B) \oplus H^1(A)\otimes H^0(B)\\
        &= (\ker(\partial_1^A)^\intercal/\{0\}) \otimes (B_1/\im(\partial_1^B)^\intercal) \oplus (A_1/\im(\partial_1^A)^\intercal) \otimes (\ker(\partial_1^B)^\intercal/\{0\})\,,
    \end{align}
    the same line of argument give us
    \begin{align}
        d_Z \leq \min\{\Tilde{d}_A d_B^\intercal, d_A^\intercal \Tilde{d}_B\} = \min\{d_B, d_A\} \,.
    \end{align}
    Therefore, the distance of the HGP code is upper bounded as
    \begin{align}
        d = \min\{d_X,d_Z\} \leq \min\{d_A,d_B,d_A^\intercal,d_B^\intercal\} \,,
    \end{align}
    which completes the proof.
\end{proof}

\subsection{Qudit La-cross codes}\label{subsec:qudit-la-cross-codes}
We now specify the details of a specific type of qudit HGP code, namely the quditized \emph{La-cross codes}~\cite{pecorari2025}. The La-cross code is obtained by taking both seed codes of the HGP code to be a length-$n_c$ cyclic code described by parity-check matrix $H$, which gives a CSS code with check matrices
\begin{align}
    H_X &= [I_n\otimes H | H^\intercal\otimes I_{n^\intercal}] \\
    \intertext{and}
    H_Z &= [H \otimes I_n | -I_{n^\intercal}\otimes H^\intercal] \,.
\end{align}
Again, we take the block coefficients to be $\gamma_1 = \gamma_2 = \delta_1 = 1$ and $\delta_2 = -1$. Before moving on with the qudit La-cross code construction, we formally define cyclic codes over $\field_q$ and discuss their properties relevant to La-cross codes:

\begin{definition}[$\field_q$ cyclic codes]\label{def:fq-cyclic-codes}
    A length-$n$ linear code $C \subseteq \field_q^n$ is a \emph{cyclic code} if, for any $c = (c_0, \ldots, c_{n-1}) \in C$, it holds that $(c_{n-1}, c_0, \ldots, c_{n-2})\in C$.
\end{definition}

A cyclic code $C$ can be described in terms of polynomials in a quotient polynomial ring $\mathcal{R} := \field_q[x]/(x^n-1)$, where a codeword $c \in C$ corresponds to a polynomial $\defsum{i=0}{n-1}{c_i x^i}$ (for a more extensive discussion, see, for example,~\cite[Chapter 4]{huffman2003}). In its polynomial representation, $C$ can be described as an ideal $I(C)$ of $\mathcal{R}$ (see~\cref{def:ring-ideal} and discussion in~\cref{subsec:fields-rings-primer}). If we take $I(C)$ to be the set of polynomials $I(C) = \{\defsum{i=0}{n-1}{c_ix^i} : c \in C\}$, it holds that:
\begin{enumerate}
    \item The zero polynomial $0$ is in $I(C)$ since $0^n\in C$.
    \item For polynomials $\defsum{i=0}{n-1}{c_ix^i}$ and $\defsum{i=0}{n-1}{c_i'x^i}$ in $I(C)$, it holds that $\defsum{i=0}{n-1}{(c_i+c_i')x^i}$, since for $c,c' \in C$ we have $c + c' \in C$.
    \item For any $f' = \defsum{i=0}{n-1}{f_i'x^i} \in \mathcal{R}$ and $\defsum{i=0}{n-1}{c_ix^i} \in I(C)$, it holds that the polynomial $f = f' \defsum{i=0}{n-1}{c_ix^i}$ is in $I(C)$ since
    \begin{align}
        f = \defsum{j=0}{n-1}{f_j'\defsum{i=0}{n-1}{c_ix^{i+j}}}
    \end{align}
    is nothing but an $\field_q$-linear combination of $\{ \text{shift}_j(c)\}_j$, where $\text{shift}_j(c)$ is a shift of elements of $c$ to the right $j$ times, that is, $c \mapsto (c_{n-1-j}, \ldots, c_{n-1}, c_0, \ldots, c_{n-j-2})$. Namely, $\defsum{j}{}{f_j'\,\text{shift}_j(c)} \in C$.
\end{enumerate}
Thus, we can always describe $C$ as an ideal $I(C) \subseteq \mathcal{R}$. On the other hand, since $\mathcal{R}$ is a principal ideal ring (see~\cref{def:principal-ideal}), there always exists a \emph{generator polynomial} $g \in \mathcal{R}$ such that $C$ is described by an ideal
\begin{align}
    I = \< g \>_{\mathcal{R}} = \{fg : f\in \mathcal{R}\}\,.
\end{align}
We may always replace $g$ by the greatest common divisor of $g$ and $x^n-1$ in $\field_q[x]$, so we can assume $g$ divides $x^n-1$ when viewed as polynomials in $\field_q[x]$. For a polynomial $f = \defsum{i=0}{n-1}{f_i x^i} \in \mathcal{R}$, the elements of $I$ are of the form
\begin{align}
    fg = \defsum{i,j=0}{n-1}{f_ig_j x^{i+j}}\,,
\end{align}
where the polynomial $f = \defsum{i=0}{n-1}{f_i x^i} \in I$ corresponds to codeword $(f_0, \ldots, f_{n-1})$. For any given generator polynomial $g$, it can be checked that for an ideal $I = \<g\>_{\mathcal{R}}$, the set $C(I) = \{(f_0, \ldots, f_{n-1}) : f \in \<g\>_{\mathcal{R}}\} \subseteq \field_q^n$ is a cyclic code:
\begin{enumerate}
    \item $0^n\in C(I)$ because $0 \in I$.
    \item If $(c_0, \cdots, c_{n-1}), (c_0', \cdots, c_{n-1}') \in C(I)$, then $(c_0 - c_0', \cdots, c_{n-1} - c_{n-1}') \in C(I)$ because for $\defsum{i=0}{n-1}{c_ix^i}$ and $\defsum{i=0}{n-1}{c_i'x^i}$ in $I$, it must be that $\defsum{i=0}{n-1}{(c_i-c_i')x^i}$ is also in $I$.
    \item If $(c_0, \cdots, c_{n-1}) \in C(I)$, then $(c_{n-1}, c_0, \cdots, c_{n-2}) \in C(I)$ because $\defsum{i=0}{n-1}{c_ix^i}$ in $I$ implies that $x\defsum{i=0}{n-1}{c_ix^i} = \defsum{i=0}{n-1}{c_ix^{i+1}} = c_{n-1} + \defsum{i=0}{n-1}{c_{i-1}x^{i}}$ is in $I$ (using that $x^n = 1$ in $\mathcal{R}$).
\end{enumerate}
Thus, we can see that for any $c \in C(I)$, it holds that $(c_{n-1}, c_0, \ldots, c_{n-2}) \in C(I)$, which implies that $C(I)$ is a cyclic code.

A parity-check matrix $H$ of a code $C$ (i.e., a matrix $H$ such that $Hc = 0$ for all $c \in C$) can also be described by a polynomial $h \in \field_q[x]$, called the \emph{check polynomial}, which we encountered for the coprime qudit bivariate bicycle codes in~\cref{subsec:coprime-qudit-bb-codes}. For a generator polynomial $g$ (dividing $x^n-1$) of the ideal $I$, the check polynomial $h$ is given by
\begin{align}\label{eqn:cyclic_code_check_polynomial}
    h = \frac{x^n-1}{g}
\end{align}
so that for a polynomial $c = fg = \defsum{i,j=0}{n-1}{f_ig_j x^{i+j}} \in I = \< g \>_{\mathcal{R}}$ for some $f \in \mathcal{R}$, it holds that
\begin{align}
    hc = hfg = (x^n-1)f = 0\,,
\end{align}
since $x^n-1 = 0$ in $\mathcal{R} = \field_q[x]/(x^n-1)$.

The dimension $k$ of an $[n,k]$ cyclic code $C$ with generator polynomial $g$ and check polynomial $h$ can be determined from either the degree of $g$ or the degree of $h$. In particular, the dimension $k$ of $C$ is given by
\begin{align}
    k = n - \text{deg}(g)\,,
\end{align}
as can be seen in, for example,~\cite[Theorem 4.2.1]{huffman2003}). Since $hg = x^n-1$ by the definition of the check polynomial of $h$ in~\cref{eqn:cyclic_code_check_polynomial}, the degree of $h$ is $\text{deg}(h) = n - \text{deg}(g)$. Therefore, the dimension of $C$ is given by the degree of its check polynomial:
\begin{align}\label{eqn:dimension_cyclic_code_check_polynomial}
    k = \text{deg}(h)\,.
\end{align}
The La-cross code is then defined by taking a cyclic code $C$ with parity check matrix $H$ to be the seed codes of an HGP code:

\begin{definition}[Qudit La-cross code]\label{def:qudit-la-cross-code}
    A \emph{qudit La-cross code} is a CSS code with parity check matrices given by
    \begin{align}
        H_X &= [I_n\otimes H | H^\intercal\otimes I_{n^\intercal}]\\
        \intertext{and}
        H_Z &= [H \otimes I_n | -I_{n^\intercal}\otimes H^\intercal]\,,
    \end{align}
    where $H$ is the parity-check matrix of a cyclic code $C \leq \field_q^n$ with check polynomial $h(x) = h_0 + h_1x + h_kx^k \in \field_q[x]$ for some $k \leq n$ and nonzero $h_0, h_1, h_k \in \field_q$. For cyclic code parity-check matrix $H$ with $\rank{H} = n - k$, if $H$ is a matrix in $\field_q^{n\times n}$, then the La-cross code defined by $H$ is a La-cross code with \emph{periodic boundaries}. On the other hand, if $H$ is a matrix in $\field_q^{(n-k)\times n}$, then the La-cross code defined by $H$ is a La-cross code with \emph{open boundaries}. 
\end{definition}

With this definition in hand, we can express the parameters $\stabcode{n}{k}{d}_q$ of a qudit La-cross code with periodic boundaries and open boundaries in terms of the cyclic code which are defined by:

\begin{proposition}[Periodic boundary qudit La-cross code parameters]\label{prop:periodic-boundary-qudit-la-cross}
    For a cyclic code $C \leq \field_q^n$ with parity-check matrix $H \in \field_q^{n\times n}$ with check polynomial $h(x) = h_0 + h_1x + h_kx^k \in \field_q[x]$ for $k\leq n$, a qudit La-cross code defined by $H$ with periodic boundaries has code parameters $\stabcode{2n^2}{2k^2}{\min\{d,d^\intercal\}}_q$. Here, $d$ and $d^\intercal$ are the distance of the code $C$ and $C^\intercal$, respectively, where $C^\intercal$ is the code corresponding to parity-check matrix $H^\intercal$, which is also a cyclic code.
\end{proposition}

\begin{proof}
    By~\cref{prop:qudit-hgp-dimension}, an HGP code with seed codes $C_A$ and $C_B$ with $C_A = C_B \equiv C$ has parameters $\stabcode{2n^2}{k^2 +(k^\intercal)^2}{D}$, where $D = \min\{d, d^\intercal\}$. For the qudit La-cross code, the seed code $C$ is a cyclic code with parity-check matrix $H$ and parameters $[n,k,d]$ (where $k = \deg(h)$ as noted in~\cref{eqn:dimension_cyclic_code_check_polynomial}). On the other hand, the transpose code $C^\intercal$ is also a cyclic code with parity check matrix $H^\intercal$ and parameters $[n,k,d^\intercal]_q$.
\end{proof}

Likewise, we have

\begin{proposition}[Open boundary qudit La-cross code parameters]\label{prop:open-boundary-qudit-la-cross}
    For a cyclic code $C \leq \field_q^n$ with parity-check matrix $H \in \field_q^{n\times n}$ with check polynomial $h(x) = h_0 + h_1x + h_kx^k \in \field_q[x]$ for $k\leq n$, a qudit La-cross code defined by $H$ with open boundaries has code parameters $\stabcode{n^2+(n-k)^2}{k^2}{d}_q$. Here, $d$ is the distance of code $C$.
\end{proposition}

\begin{proof}
    Again we use~\cref{prop:qudit-hgp-dimension}. The code $C'$ defined by parity-check matrix $H \in \field_q^{(n-k)\times k}$ and check polynomial $h$ has parameters $[n,k,d]_q$, where $k = \deg(h)$. However, in this case the transpose code $(C')^\intercal$ has parity-check matrix $H^\intercal \in \field_q^{n\times(n-k)}$ with rank $n-k$. By the rank-nullity theorem,  the transpose code then has dimension 0. Thus the open boundary La-cross code has $n^2+(n-k)^2$ data qudits, dimension $k^2+(k^\intercal)^2 = k^2$ and distance $\min\{d,d^\intercal\}$.
\end{proof}

In~\cite[Fig. 6]{pecorari2025}, it is observed that La-cross code with open boundaries, defined by polynomials of the form $1 + x + x^k$ for $k \in \{2,3,4\}$ incurs smaller logical errors compared to surface codes with the same number of data qubits and the same number of encoded logical qubits. As such, we restrict our numerical code search to polynomials of this form, except we now allow for coefficients in $\field_q$: $\alpha_0 + \alpha_1 x + \alpha_2 x^k$, where $\alpha_i \in \field_q$ for $i = 0,1,2$. We find several codes, where the most promising codes for each value of $q \in \{3,5,7\}$ that we consider are the $\stabcode{89}{9}{5}_3$, $\stabcode{52}{4}{5}_5$, and $\stabcode{34}{4}{4}_7$ codes, respectively; see \cref{tab:qudit-la-cross-codes}.

\begin{table}[b]
    \centering
    \begin{tblr}{colspec={|Q[c,m]|Q[c,m]|Q[c,m]|},row{even}={bg=lightgray}}
        \hline
        $\stabcode{n}{k}{d}_q$ & $[n_c,k_c]$ & $\alpha_0 + \alpha_1 x + \alpha_2 x^k$\\
        \hline[1pt]
        $\stabcode{89}{9}{5}_3$ & $[8,3]$ & $2 + x + x^3$\\
        \hline
        $\stabcode{52}{4}{5}_5$ & $[6,2]$ & $4 + 4x + 3x^2$\\
        \hline
        $\stabcode{34}{4}{4}_7$ & $[5,2]$ & $6 + 5x + x^2$\\
        \hline
    \end{tblr}
    \caption{New qudit La-cross codes of local qudit dimension $q$ with qudit code parameters $n = n_c^2 + (n_c-k_c)^2$, $k = k_c^2$, and $d$ and classical seed code parameters $n_c$ and $k_c$.}
    \label{tab:qudit-la-cross-codes}
\end{table}

We perform a numerical search for the qudit La-cross codes in a similar way that we did for the qudit bivariate bicycle codes. However, the search for qudit La-cross codes is simpler as we only consider a single polynomial of the form $\alpha_0 + \alpha_1 x + \alpha_2 x^k$, where $\alpha_0, \alpha_1, \alpha_2 \in \field_q$ and $k \in \{2,3,4\}$. In our search, we take $q \in \{3,5,7\}$. We only consider open boundary conditions, meaning the classical seed code parity-check matrix is rectangular and lies in $\field_q^{(n-k) \times n}$, where we take the code length $n \in \{3, \ldots, 7\}$. Rather than iterating through the entire search space of parameters $\{q, n, k, \alpha_0, \alpha_1, \alpha_2\}$, which has a lot of redundancy (i.e., the same resulting from code from different input parameters), we choose random combinations of input parameters and search for $\sim 10^3$ codes. We deem a code ``promising'' if it has distance at least 4 and encoding rate $r = k/2n$ at least $1/50$. We present these promising codes in~\cref{tab:qudit-la-cross-codes} and the decoding results in~\cref{fig:qudit-la-cross-decoding}. We use the same mixed integer programming-based distance finding algorithm and decoder that we used for the qudit bivariate bicycle codes, as discussed in~\cref{subsec:qudit-bb-code-numerics}.

\begin{figure}[tb]
    \centering
    \includegraphics[scale=0.4]{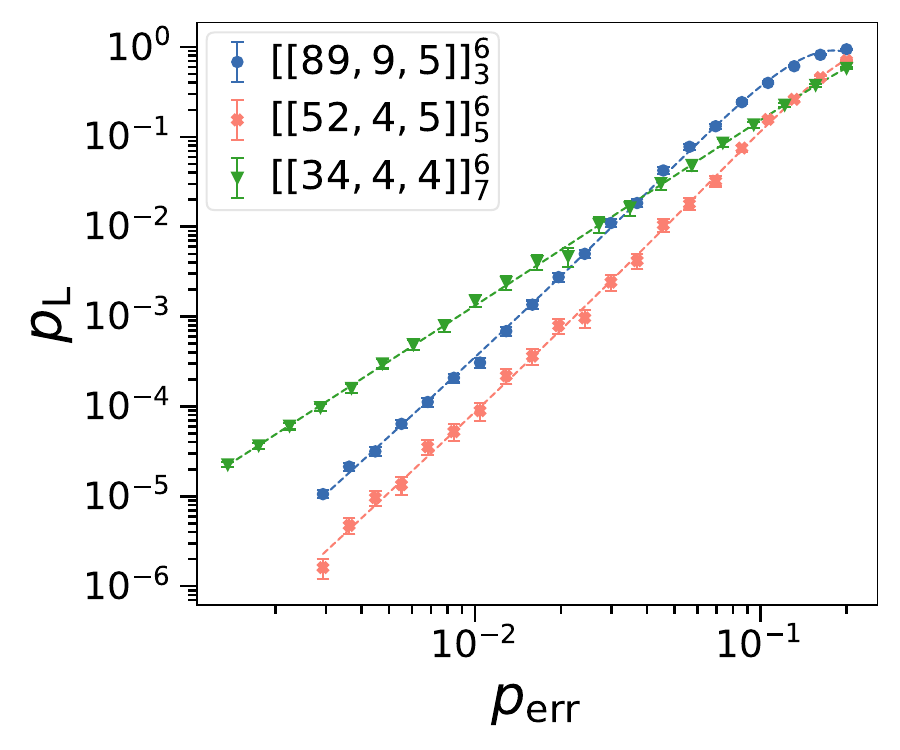}
    \caption{Code capacity of the best found qudit La-cross codes. We plot logical error $p_\text{L}$ against physical error $p_{\text{err}}$. The subscript on each code in the legend corresponds to the local qudit dimension $q$. We fit with~\cref{eq:heuristic_fit}  where we find $d_\text{fit}=\{4.7,4.8,3.1\}$. See~\cite{pecorari2025} for comparable qubit La-cross codes.}
    \label{fig:qudit-la-cross-decoding}
\end{figure}

\section{Qudit subsystem hypergraph product simplex codes}\label{sec:qudit-shyps-codes}
We now quditize subsystem hypergraph product (SHP) codes~\cite{li2020}, specifically the subsystem hypergraph product simplex (SHYPS) codes~\cite{malcolm2025}. SHYPS codes are a combination of SHP codes and classical simplex codes and allow for efficient implementation of Clifford operations. As their name suggests, SHYPS codes are a type of of \emph{quantum subsystem code}, which can be thought of as a generalization of stabilizer codes~\cite{liu2024}. We first give an overview of classical simplex codes over $\field_q$, then an overview of SHP codes for qubits, and finally combine these into qudit SHYPS codes.

\subsection{Classical simplex codes}\label{subsec:classical-simplex-codes}
We start by reviewing classical binary simplex codes and then generalize to $q$-ary simplex codes. We take $r$ to be an integer greater than or equal to 3, $r \geq 3$, and define $n_r := 2^r - 1$ and $d_r := 2^{r-1}$. Then, the classical simplex codes $C(r)$ are a family of $[n_r, r, d_r]$-linear codes that are dual to the Hamming codes. For $3 \leq r < 500$, it is known that there is a polynomial $h(x) = 1 + x^a + x^b \in \field_2[x]/(x^{n_r}-1)$ such that $\gcd(h(x), x^{n_r}-1)$ is a primitive polynomial of degree $r$. The $n_r \times n_r$ matrix
\begin{align}
    H = 
    \begin{pmatrix}
        h(x)\\
        xh(x)\\
        \vdots\\
        x^{n_r - 1}h(x)
    \end{pmatrix}
\end{align}
then defines a parity-check matrix for $C(r)$, we use the polynomial notation for cyclic matrices (similar to that used for the bivariate bicycle codes in~\cref{sec:qudit-bb-codes}):
\begin{align}
    \defsum{i=0}{n_r - 1}{a_ix^i \pmod{x^{n_r} - 1}} \mapsto (a_0, \ldots, a_{n_r - 1})\,,
\end{align}
where $(a_0, \ldots, a_{n_r - 1}) \in \field_2^{n_r}$. These codes have parameters $[2^r-1, r, 2^{r-1}]$.

Over $\field_q$, we have the $q$-ary simplex codes, denoted $S(q,r)$, with code parameters $[n_r,r,q^{r-1}]_q$, where, as usual, $q$ is a prime power and $n_r = (q^r - 1)/(q - 1)$. We denote by $G$ the generator matrix of this code, whose columns form a complete set of representatives for the elements of the projective space $PG(r-1,q)$. To see this, note that to construct a $q$-ary simplex code, each column of $G$ is taken to be a representative nonzero vector $\vb{v} \in \field_q^r$ from a 1-dimensional subspace, where by "1-dimensional subspace" we mean a nonzero vector and all of its scalar multiples (including zero). Likewise, the projective space $PG(r-1,q)$ is defined as
\begin{equation}
    PG(r-1,q) := \{\langle \vb{v} \rangle : \vb{v} \in \field_q^r \setminus \{\vb{0}\}\}\,,
\end{equation}
where $\langle \vb{v} \rangle = \{\alpha\vb{v} : \alpha \in \field_q\}$. That is, $PG(r-1,q)$ is exactly the set of nonzero vectors in $\field_q^r$ up to scalar multiples, which is, of course, the set of 1-dimensional subspaces of $\field_q^r$. Because $G$ contains exactly one representative from each of these subspaces, there is a one-to-one correspondence between the columns of $G$ and the points of $PG(r-1,q)$.

If we write $G$ in standard form as $G = [I_k | P]$, we can use the constraint that $GH^\intercal = 0$, where $H$ is the parity-check matrix, to find that $H = [-P^\intercal | I_{n-k}]$. Alternatively, given $G$, we can find $H$ by using the relation $GH^\intercal = 0$, as mentioned in~\cref{subsec:coding-theory-primer}.

\subsection{Subsystem codes}\label{subsec:shp-codes}
Quantum subsystem codes are a generalization of stabilizer codes in that they are defined with respect to a subgroup $\mathcal{G} \subseteq \mathcal{P}_n$ called the \emph{gauge group}, where, recall, $\mathcal{P}_n$ is the $n$-qubit Pauli group. The effectively adds a new set of qubits to the code that we call "gauge qubits." We do not need to protect the information stored in these gauge qubits. Instead of measuring high-weight stabilizers directly, the error correction procedure involves measuring low-weight "gauge operators" (the generators of $\mathcal{G}$). The actual stabilizers are then products of these gauge operators.

We can define subsystem codes formally according to the discussion in Ref.~\cite[Sec. VIII.B]{malcolm2025}. Given $\mathcal{G}$, the corresponding stabilizer group $\mathcal{S}$ is the center of $\mathcal{G}$ modulo phases:
\begin{align}
    \langle S, iI \rangle = Z(\mathcal{G}) := C_{\mathcal{P}_n}(\mathcal{G}) \cap \mathcal{G}\,.
\end{align}
The fixed point space is nontrivial and decomposes into $C_{\mathcal{L}} \otimes C_{\mathcal{G}}$, where elements of $\mathcal{G} \backslash \mathcal{S}$ fix $C_{\mathcal{L}}$, where $C_{\mathcal{L}}$ denotes the set of logical qubits and $C_{\mathcal{G}}$ denotes the set of gauge qubits. These codes have two types of logical operators: the operators given by $C_{\mathcal{P}_n}(\mathcal{G}) \backslash \mathcal{G}$ that act nontrivially on $C_{\mathcal{L}}$, which we call \emph{bare logicals}, and the operators given by $C_{\mathcal{P}_n}(\mathcal{S}) \backslash \mathcal{G}$ that act on $C_{\mathcal{L}}$ and $C_{\mathcal{G}}$, which we call \emph{dressed logicals}. The minimum distance is given by
\begin{align}
    d = \min\{\abs{P} : P \in C_{\mathcal{P}_n}(\mathcal{S}) \backslash \mathcal{G}\}\,,
\end{align}
that is, the minimum weight Pauli that acts nontrivially on the logical qubits (i.e., the minimum weight of the dressed logicals).

\subsection{Qudit subsystem hypergraph product simplex codes}\label{subsec:qudit-shyps-codes}
We now put everything together to give our qudit subsystem hypergraph product simplex (SHYPS) codes. We start by defining the SHYPS codes for qubits~\cite{malcolm2025}:

\begin{definition}[Subsystem hypergraph product simplex code~{\cite[Def. VIII.8]{malcolm2025}}]\label{def:shyps-code}
    Let $r \geq 3$, $n_r = 2^r - 1$, $d_r = 2^{r-1}$, and $H_r$ be the parity-check matrix for the $[n_r, r, d_r]$-classical simplex code. Then, the subsystem hypergraph product of two copies of $H_r$, which we denote by SHYPS($r$), is the subsystem CSS code with gauge generators
    \begin{align}
        G_X = (H_r \otimes I_{n_r})\,, \quad G_Z = (I_{n_r} \otimes H_r)\,.
    \end{align}
    SHYPS($r$) is then called a \emph{subsystem hypergraph product simplex code}.
\end{definition}

SHYPS($r$) codes are a type of quantum LDPC code with gauge generators of weight 3. The code parameters are given in the following theorem:

\begin{fact}[SHYPS code parameters~{\cite[Thm. VIII.9]{malcolm2025}}]\label{fact:shyps-code-parameters}
    Let $r \geq 3$ and $H$ be the parity-check matrix for the $(2^r - 1, r, 2^{r-1})$-classical simplex code. The subsystem hypergraph product simplex code SHYPS($r$) is an $\stabcode{n}{k}{d}$-quantum subsystem code with gauge group generated by 3-qubit operators and
    \begin{align}
        n &= (2^r - 1)^2\,,\\
        k &= r^2\,,\\
        d &= 2^{r-1}\,.
    \end{align}
\end{fact}

\begin{proof}
    We reproduce the proof from \cite{malcolm2025}. We start by making an $n_r \times n_r$ lattice such that the number of physical qubits needed is $n = n_r^2$. To find the number of logical qubits $k$, we note that the Pauli operators that commute with all gauge generators are generated by
    \begin{align}
        \mathcal{L}_X = (I_{n_r} \otimes G)\,, \quad \mathcal{L}_Z = (G \otimes I_{n_r})\,,
    \end{align}
    where $G$ is a rank-$r$ generator matrix for the classical simplex code such that $HG^\intercal = 0$. The stabilizer generators are given by
    \begin{align}
        \mathcal{S}_X = (H \otimes G)\,, \quad \mathcal{S}_Z = (G \otimes H)\,.
    \end{align}
    The number of logical qubits $k$ is then the difference between the ranks:
    \begin{align}
        k &= \rank{\mathcal{L}_X} - \rank{\mathcal{S}_X}\\
        &= n_r \cdot r - (n_r - r) \cdot r\\
        &= r^2\,.
    \end{align}
    Finally, the distance is the same as for the classical simplex code: $d = d_r = 2^{r - 1}$.
\end{proof}

While for all of the other codes that we have analyzed in this paper, it has been sufficient to replace $\field_2$ by $\field_q$, in this case it is not so simple; we cannot simply promote the polynomials $h(x)$ to be defined over $\field_q[x]/(x^{n_r}-1)$ instead of $\field_2[x]/(x^{n_r}-1)$. The reason for this is that, while three-term polynomials $h(x)$ that meet the restrictive criteria above exist over the quotient ring $\field_2[x]/(x^{n_r} - 1)$, they do not exist over $\field_q[x]/(x^{n_r} - 1)$ for $q > 2$ (even if we drop the three-term criteria). Formally, we have

\begin{proposition}[Existence of higher order polynomials]\label{prop:polynomial-existence}
    For integers $q > 2$, $r \geq 3$, and $n_r = (q^r-1)/(q-1)$, there does not exist a polynomial $h(x) \in \field_q[x]/(x^{n_r}-1)$ such that $\gcd(h(x), x^{n_r}-1)$ is a primitive polynomial of degree $r$.
\end{proposition}

\begin{proof}
    We define $p(x) = \gcd(h(x), x^{n_r} - 1)$, where $n_r = (q^r-1)/(q-1)$ for a $q$-ary simplex code~\cite{eczoo2025}. If $p(x)$ is a primitive polynomial over $\field_q[x]/(x^{n_r}-1)$, the smallest positive integer $k$ such that $p(x)$ divides $x^k-1$ is $k = q^r-1$. But $p(x) = \gcd(h(x), x^{n_r}-1)$ implies that $p(x)$ divides $x^{n_r}-1$, so $k \leq n_r$ by the minimality of $k$. Since $n_r = (q^r-1)/(q-1) = k/(q-1)$, we must have $q-1 \leq 1$, and thus $q = 2$ is the only value for $q$ that works.
\end{proof}

Instead, we take as the classical simplex seed code the $q$-ary simplex codes defined in~\cref{subsec:classical-simplex-codes}. Recall from above, a $q$-ary simplex code $S(q,r)$ has parameters $[n_r,r,q^{r-1}]_q$, where $n_r = (q^r-1)/(q-1)$. Here, the $q$-ary simplex code $S(q,r)$ is a projective code whose generator matrix is constructed by taking the columns as chosen representatives of each element in the projective space $PG(r-1,q)$. Then, the qudit SHYPS code parameters are given by the following:

\begin{proposition}[Qudit SHYPS code parameters]\label{prop:qudit-shyps-code-parameters}
    Let $r \geq 3$ and $H$ be the parity-check matrix for the $(n, r, q^{r-1})_q$-classical $q$-ary simplex code, where $n = \frac{q^r - 1}{q - 1}$. The qudit subsystem hypergraph product simplex code $\text{SHYPS}(q, r)$ is an $[n,k,d]_q$-quantum subsystem code, where
    \begin{align}
        n &= n_r^2 = \lp\frac{q^r - 1}{q-1}\rp^2\,,\\
        k &= r^2\,,\\
        d &= q^{r-1}\,.
    \end{align}
\end{proposition}

\begin{proof}
    The proof is just an extension of the proof for qubits in~\cref{fact:shyps-code-parameters}. The number of physical qudits is derived from the $n_r \times n_r$ lattice, so we have $n = n_r^2$, where $n_r = (q^r-1)/(q-1)$ is inherited from the classical seed code. For $k$, we construct the gauge generators in the same way:
    \begin{align}
        \mathcal{L}_X = (I_{n_r} \otimes G)\,, \quad \mathcal{L}_Z = (G \otimes I_{n_r})\,,
    \end{align}
    where the only difference is in how we construct $G$, though $\rank{G} = r$ still holds and we still have $GH^\intercal = 0$, where $H$ is the parity-check matrix of the classical $q$-ary simplex code. Likewise, the stabilizer generators are still defined as
    \begin{align}
        \mathcal{S}_X = (H \otimes G)\,, \quad \mathcal{S}_Z = (G \otimes H)\,,
    \end{align}
    where $\rank{H} = n_r-r$ holds. Thus, the number of logical qudits is
    \begin{align}
        k &= \rank{\mathcal{L}_X} - \rank{\mathcal{S}_X}\\
        &= r^2\,,
    \end{align}
    the same as it was for qubits. Finally, the distance is inherited from the distance of the classical $q$-ary simplex code, which yields $d = d_r = q^{r-1}$.
\end{proof}

Since we do not construct the classical gauge generator matrix $G$ out of a three-term polynomial and instead from elements of the subspaces of the projective group $PG(r-1, q)$, we lose the constant weight of 3 for all values of $r$ that the qubit SHYPS codes benefit from. Instead, as can be seen in~\cref{fig:qudit-shyps-min-max-weights}, while the minimum number of nonzero elements in each row of the parity-check matrix for different values of $r$ remains constant at 3, the maximum number scales as $r+1$. Thus, while qudit SHYPS codes benefit from a favorable distance scaling as $q^{r-1}$, it comes at the cost of an increasing weight. As such, these codes lose their LDPC character as we fix $q$ and increase $r$ to get a larger distance. Thus, it becomes a question of balancing the weight of the parity-checks and the desired distance of the resulting code. We present some sample codes in~\cref{tab:sample-shyps-codes}.

\begin{figure}[t]
    \centering
    \includegraphics[scale=0.5]{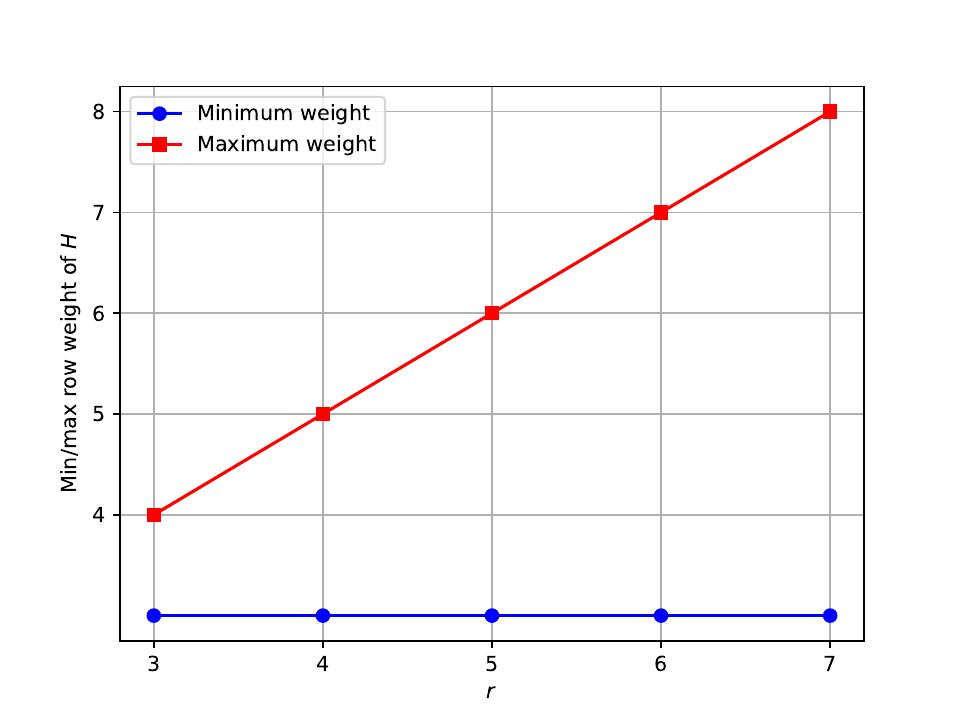}
    \caption{Minimum and maximum weights of classical $q$-ary simplex code parity-check matrix as a function of $r$.}
    \label{fig:qudit-shyps-min-max-weights}
\end{figure}

\begin{table*}[b]
    \centering
    \begin{tblr}{colspec={|Q[c,m]Q[c,m]Q[c,m]},hlines,vlines,row{even}={bg=lightgray}}
        $q$ & $r$ & $\stabcode{n}{k}{d}_q$\\
        3 & 3 & $\stabcode{169}{9}{9}_3$\\
        3 & 4 & $\stabcode{1600}{16}{27}_3$\\
        5 & 3 & $\stabcode{961}{9}{25}_5$\\
        5 & 4 & $\stabcode{24336}{16}{125}_5$\\
        7 & 3 & $\stabcode{3249}{9}{49}_7$\\
        7 & 4 & $\stabcode{160000}{16}{343}_7$
    \end{tblr}
    \caption{The first 2 qudit SHYPS codes for each of $q \in \{3, 5, 7\}$ using $r \in \{3,4\}$.}
    \label{tab:sample-shyps-codes}
\end{table*}

\section{Qudit high-dimensional expander codes}\label{sec:qudit-hdx-codes}
In this section, we quditize the \emph{high-dimensional expander} (HDX) codes~\cite{evra2020}. HDX codes are a class of CSS codes defined by a quantum code and a classical linear code. The HDX codes take advantage of concepts from graph theory to construct codes with excellent distance properties. Expander graphs are simultaneously sparse (maintaining the LDPC property) yet highly-connected. This ensures that errors spread out quickly and are efficiently detected by the sparse checks. HDX codes utilize high-dimensional generalizations of these structures to achieve superior distance scaling.

We begin by considering the chain complexes (see~\cref{subsec:homological-algebra-primer} for a primer on homological algebra)
\begin{equation}
    \label{eqn:3-and-2-chain-complexes}
    \begin{aligned}
        X &: 0 \xrightarrow{\partial_3^X} X_2 \xrightarrow{\partial_2^X} X_1 \xrightarrow{\partial_1^X} X_0 \xrightarrow{\partial_0^X} 0\,,\\
        Y &: 0 \xrightarrow{\partial_2^Y} A \xrightarrow{\partial_1^Y} B \xrightarrow{\partial_0^Y} 0\,,
    \end{aligned}
\end{equation}
with their co-complexes given by
\begin{align}
    X^* &: 0 \xrightarrow{\delta_3^X} X_2^* \xrightarrow{\delta_2^X} X_1^* \xrightarrow{\delta_1^X} X_0^* \xrightarrow{\delta_0^X} 0\\
    &= 0 \xrightarrow{\delta_3^X} X_0 \xrightarrow{\delta_2^X} X_1 \xrightarrow{\delta_1^X} X_2 \xrightarrow{\delta_0^X} 0
\end{align}
and
\begin{align}
    Y^* &: 0 \xrightarrow{\delta_2^Y} B \xrightarrow{\delta_1^Y} A \xrightarrow{\delta_0^Y} 0\,.
\end{align}
The product chain complex $X^*\otimes Y^*$ between co-complexes $X^*$ and $Y^*$, is given by
\begin{align}
    0 &\to (X_2^*\otimes A)\oplus(X_1^*\otimes B) \xrightarrow{\delta_2^{X^*\otimes Y^*}} (X_1^*\otimes A) \oplus (X_0^*\otimes B) \xrightarrow{\delta_1^{X^*\otimes Y^*}} (X_0^*\otimes A) \to 0\,,
\end{align}
whereas the co-complex of $X^*\otimes Y^*$, denoted by $\mathcal{X}:=(X^*\otimes Y^*)^*$, is given by
\begin{align}
    0 &\to (X_0^*\otimes A) \xrightarrow{\partial_2^\mathcal{X}} (X_1^*\otimes A) \oplus (X_0^*\otimes B) \xrightarrow{\partial_1^\mathcal{X}} (X_2^*\otimes A)\oplus(X_1^*\otimes B) \to 0\,,
\end{align}
or, equivalently
\begin{align}\label{eqn:HDX-cochain-complex}
    0&\xrightarrow{} (X_2\otimes A) \xrightarrow{\partial_2^\mathcal{X}} (X_1\otimes A) \oplus (X_2\otimes B) \xrightarrow{\partial_1^\mathcal{X}} (X_0\otimes A)\oplus(X_1\otimes B) \xrightarrow{} 0\,,
\end{align}
where its boundary maps are given by
\begin{align}
    \partial_2^\mathcal{X} &= (\partial_2^X\otimes I_A) + (I_{X_2}\otimes\partial_1^Y)
\end{align}
and
\begin{align}
    \partial_1^\mathcal{X}(v) = (\partial_1^X\otimes I_A)(v) - (I_{X_1}\otimes \partial_1^Y)(v)
\end{align}
if $v \in X_1 \otimes A$, or
\begin{align}
    \partial_1^\mathcal{X}(v) = (\partial_2^X\otimes I_B)(v)
\end{align}
if $v \in X_2 \otimes B$. Thus, for $u \oplus v \in (X_1 \otimes A) \oplus (X_2 \otimes B)$, we have 
\begin{align}
    \partial_1^\mathcal{X}(u\oplus v) &= \lp \lp \partial_1^X\otimes I_A \rp(u)\rp \oplus \lp -(I_{X_1} \otimes \partial_1^Y)(u) + (\partial_2^X \otimes I_B)(v) \rp \,.
\end{align}
We can directly verify that $(\partial^\mathcal{X})^2 = 0$:
\begin{align}
    \partial_1^\mathcal{X} \circ \partial_2^\mathcal{X} &= \partial_1^\mathcal{X}\circ(\partial_2^X \otimes I_A) + \partial_1^\mathcal{X}\circ(I_{X_2} \otimes \partial_1^Y)\\
    \begin{split}
        &= (\partial_1^X \otimes I_A)\circ(\partial_2^X \otimes I_A) - (I_{X_1} \otimes \partial_1^Y)\circ(\partial_2^X \otimes I_A)\\
        &\quad + (\partial_2^X \otimes I_B)\circ(I_{X_2} \otimes \partial_1^Y)
    \end{split}\\
    &= -(\partial_2^X \otimes \partial_1^Y) + (\partial_2^X \otimes \partial_1^Y)\\
    &= 0\,,
\end{align}
where for the third equality we used the fact that  $\partial_2^X\circ\partial_1^X=0$. 

We note that the differences in how the boundary maps are defined in this is a qudit generalization of the qubit HDX code given in~\cite[Def. 4.1]{evra2020}. The minus sign in front of the $I_{X_1}\otimes\partial_1^Y$ factor in $\partial_1^\mathcal{X}$ can be treated as a plus sign in qubit codes where the boundary maps correspond to parity check matrices over $\field_2$ (as $-1=1$ in $\field_2$). Such treatment of qubit HDX codes gives $\partial_1^\mathcal{X}\circ\partial_2^\mathcal{X}=0$, indicating that the CSS condition $H_XH_Z^\intercal=0$ is satisfied (since $(\partial_2^X \otimes \partial_1^Y) + (\partial_2^X \otimes \partial_1^Y) = 0$ in $\field_2$). However in qudit HDX codes, the boundary maps correspond to parity check matrices over $\field_q$. Namely, $\partial_1^\mathcal{X} \mapsto H_X$ and $\partial_2^\mathcal{X} \mapsto H_Z^\intercal$ where $H_X:\field_q^{X_1\otimes A \oplus X_2\otimes B}\rightarrow\field_q^{X_0\otimes A \oplus X_1\otimes B}$ and $H_Z:\field_q^{X_1\otimes A \oplus X_2\otimes B}\rightarrow\field_q^{X_2\otimes A}$. Not adding the minus sign may lead to $\partial_1^\mathcal{X}\circ\partial_2^\mathcal{X}\neq0$, since in general $(\partial_2^X \otimes \partial_1^Y) + (\partial_2^X \otimes \partial_1^Y) \neq 0$. The minus sign Thus guarantees that $\partial_1^\mathcal{X}\circ\partial_2^\mathcal{X}=0$, and hence $H_XH_Z^\intercal=0$.

Now, for chain complex $\mathcal{X}$ with boundary map $\partial^\mathcal{X}$ and co-boundary map $\delta^\mathcal{X}$, we define the $k$-\emph{systole} and $k$-\emph{co-systole} of $\mathcal{X}$ as
\begin{align}
    S_k(\mathcal{X}) &= \min\lc \abs{v} : v \in (\ker{\partial_k^\mathcal{X}}) \backslash (\im\partial_{k+1}^\mathcal{X}) \rc\\
    \intertext{and}
    S^k(\mathcal{X}) &= \min\Big\{ |v| : v\in(\ker\delta_k^\mathcal{X})\backslash(\im\delta_{k-1}^\mathcal{X}) \Big\} \,,
\end{align}
respectively. At the same time, the $k\numth$ homology and cohomology group of $\mathcal{X}$ are given by
\begin{align}
    H_k(\mathcal{X}) &= \ker \partial_k^\mathcal{X} / \im \partial_{k+1}^\mathcal{X}\\
    \intertext{and}
    H^k(\mathcal{X}) &= \ker \delta_k^\mathcal{X} / \im \delta_{k-1}^\mathcal{X}\,,
\end{align}
respectively, where $\delta_k^\mathcal{X} = (\partial_{k+1}^\mathcal{X})^\intercal$. Since $\im\delta_{k-1}^\mathcal{X} = \im(\partial_k^\mathcal{X})^\intercal = (\ker \partial_k^\mathcal{X})^\perp$ and $(\im\partial_{k+1}^\mathcal{X})^\perp = \ker (\partial_{k+1}^\mathcal{X})^\intercal = \ker\delta_{k}^\mathcal{X}$, we have
\begin{align}
    H^k(\mathcal{X}) &= (\im\partial_{k+1}^\mathcal{X})^\perp / (\ker \partial_k^\mathcal{X})^\perp = \ker \partial_k^\mathcal{X} / \im\partial_{k+1}^\mathcal{X} = H_k(\mathcal{X}) \,,
\end{align}
that is, the $k\numth$ homology and $k\numth$ cohomology coincide. With this in mind, we have the following:

\begin{proposition}\label{lemma:HDX_complex_homology_systole}
    Consider the product chain complex $\mathcal{X} = (X^* \otimes Y^*)^*$ from~\cref{eqn:HDX-cochain-complex}:
    \begin{align}
        0 &\to (X_2\otimes A) \xrightarrow{\partial_2^\mathcal{X}} (X_1 \otimes A) \oplus (X_2 \otimes B) \xrightarrow{\partial_1^\mathcal{X}} (X_0 \otimes A) \oplus (X_1 \otimes B) \to 0
    \end{align}
    such that $H_2(Y)=\{0\}$. 
    Then, it holds that:
    \begin{enumerate}
        \item $\dim H_1(\mathcal{X}) = (\dim H_1(X)) (\dim H_1(Y))$.
        \item $S_1(\mathcal{X}) = S_1(X) S_1(Y)$.
        \item $S^1(\mathcal{X}) = S^1(X)$.
    \end{enumerate}
\end{proposition}

\begin{proof}
    First, note that the second homology group of $Y$ satisfies $H_2(Y) = 0$. Now, the first homology group of $\mathcal{X}=(X^*\otimes Y^*)^*$ is given by
    \begin{align}
        H_1(\mathcal{X}) = H_1((X^*\otimes Y^*)^*) &\cong H^2(X^*\otimes Y^*) \,, 
    \end{align}
    since $\mathcal{X}$ is the dual chain complex of $X^*\otimes Y^*$.
    By the K\"unneth formula, therefore we have
    \begin{equation}
    \begin{aligned}
        H^2(X^*\otimes Y^*) &= H^2(X^*)\otimes H^0(Y^*) \oplus H^0(X^*)\otimes H^2(Y^*) \oplus H^1(X^*)\otimes H^1(Y^*) \\
        &\cong H_0(X)\otimes H_2(Y) \oplus H_2(X)\otimes H_0(Y) \oplus H_1(X)\otimes H_1(Y) \\
        &= H_1(X)\otimes H_1(Y)
    \end{aligned}
    \end{equation}
    because $H_2(Y)=H_0(Y)=\{0\}$.
    Therefore the dimension of the first homology group of $\mathcal{X}$ is given by
    \begin{equation}
        \dim H_1(\mathcal{X}) = (\dim H_1(X)) (\dim H_1(Y)) \,,
    \end{equation}
    as claimed.

    Now we show that $S_1(\mathcal{X}) = S_1(X)S_1(Y)$.
    First we show that $S_1(\mathcal{X}) \leq S_1(X)S_1(Y)$. Consider $z_X\in\ker\partial_1^X\setminus\im\partial_2^X$ and $z_Y\in\ker\partial_1^Y\setminus\{0\}$ and a vector $x\in\mathcal{X}_1$ given by $z=(z_X\otimes z_Y)\oplus 0$. Furthermore, let $|z_X|=S_1(X)$ and $|z_Y|=S_1(Y)$ so that $|z| = S_1(X)S_1(Y)$.
    Note that $z\in\ker\partial_1^\mathcal{X}$ because $\partial_1^X\otimes\partial_1^Y(z) = 0$.
    Also $z\notin\im\partial_1^\mathcal{X}$ since there is no $z'\in\mathcal{X}_2=X_2\otimes A$ such that $\partial_1^X\otimes I_{A}(z')=z$ because $z_X\notin\im\partial_1^X$. Thus $z\in\ker\partial_1^\mathcal{X}\setminus\im\partial_2^\mathcal{X}$, which implies that
    \begin{equation}
        S_1(\mathcal{X}) \leq |z| = S_1(X)S_1(Y) \,.
    \end{equation}
    Now we show that $S_1(\mathcal{X}) \geq S_1(X)S_1(Y)$. Let $z\in \ker\partial_2^\mathcal{X}\setminus\im\partial_1^\mathcal{X}$ with support in $X_1\otimes A$, so that
    \begin{equation}
        \partial_1^X\otimes I_A(z)=0 \,,\quad 
        I_{X_1}\otimes\partial_1^Y(z)=0 \,.
    \end{equation}
    Consider a basis $\mathcal{U}=\{u_x\}_x$ of $X_1$ and basis $\mathcal{V} = \{v_a\}_a$ of $A$ and projection of $z\in X_1\otimes A$ onto a basis element of $\mathcal{U}$ and $\mathcal{V}$ as
    \begin{equation}
        \pi_x^{X_1}(z) = \sum_a z_{x,a} v_a
        \quad\text{and}\quad
        \pi_a^A(z) = \sum_x z_{x,a} u_x \,,
    \end{equation}
    respectively, where $z_{x,a}\in\field_q$ are coefficients from the expansion $z = \sum_{x,a} z_{x,a} u_x\otimes v_a$. Note that the projections satisfy\begin{equation}\label{eqn:weight_sum_nontrivial_cycle_HDX}
        \partial_1^Y(\pi_x^{X_1}(z)) = 0
        \quad\text{and}\quad
        \partial_1^X(\pi_a^A(z)) = 0
    \end{equation}
    since $z\in \ker\partial_2^\mathcal{X}$.
    On the other hand, $\pi_a^A(z) \notin\im\partial_2^X$ since $z\notin \im\partial_1^\mathcal{X}$.
    Let us write the weight of $z$ as
    \begin{equation}
        |z| = \sum_x |\pi_x^{X_1}(z)| \,.
    \end{equation}
    Since $\pi_a^A(z) \in\ker\partial_1^X\setminus\im\partial_2^X$ we have $|\pi_a^A(z)|\geq S_1(X)$.
    On the other hand, $|\pi_x^{X_1}(z)|\geq S_1(Y)$ since $\pi_x^{X_1}(z)\in\ker\partial_1^Y$.
    Thus there are at least $S_1(Y)$ nonzero summands in~\cref{eqn:weight_sum_nontrivial_cycle_HDX}, and we have
    \begin{equation}
        |z| = \sum_x |\pi_x^{X_1}(z)| \geq S_1(Y) S_1(X) \,,
    \end{equation}
    giving us the desired lower bound when we set $z$ to be the minimum weight element in $\ker\partial_1^\mathcal{X}\setminus\im\partial_2^\mathcal{X}$.

    Now we show that $S^1(\mathcal{X}) = S^1(X)$.
    First we show that $S^1(\mathcal{X}) \leq S^1(X)$.
    Consider a partition $A = A'\oplus A''$ such that $\dim A' = \dim B$ and the restriction of $(\partial_2^Y)^\intercal$ to $A'$, denoted as $(\partial_2^Y)^\intercal|_{A'}$ is full-rank (i.e. $(\partial_2^Y)^\intercal|_{A'}$ is a $\dim B \times \dim A$ non-singular submatrix of $(\partial_2^Y)^\intercal$).
    So, subspace $A''$ is not in $\im(\partial_2^Y)^\intercal$.
    Now consider $z = z'\oplus z'' \in\mathcal{X}^1$ with $z''=0$ and $z'=x\otimes a''$ for $x\in\ker(\partial_2^X)^\intercal\setminus\im(\partial_1^X)^\intercal$ with weight $|x|=S^1(X)$ and a basis element $a''\in A''$ (so that $|a''|=1$).
    Thus, we have $|z| = S^1(X)$.
    Now we need to show that $z\in\ker(\partial_2^\mathcal{X})^\intercal\setminus\im(\partial_1^\mathcal{X})^\intercal$.
    Since $x\in\ker(\partial_2^X)^\intercal$, we have $z\in\ker(\partial_2^\mathcal{X})^\intercal$.
    Also we have $z\notin\im(\partial_1^\mathcal{X})^\intercal$ because $x\notin\im(\partial_1^X)^\intercal$ and $a''\notin(\partial_1^Y)^\intercal$.
    Thus $z\in\ker(\partial_2^\mathcal{X})^\intercal\setminus\im(\partial_1^\mathcal{X})^\intercal$, which implies that
    \begin{equation}
        S^1(\mathcal{X}) \leq |z| = S^1(X) \;.
    \end{equation}

    To show that $S^1(\mathcal{X}) \geq S^1(X)$, consider again partition $A=A'\oplus A''$ and full-rank restriction $(\partial_2^Y)^\intercal|_{A'}$ as above.
    Consider a basis $\mathcal{Z}=\{x_i \otimes a_j''\}_{i,j}$ of $H^1(\mathcal{X}) = H^1(X)\otimes H^1(Y)$ where $\{a_j''\}_j$ is a basis of $A''$ and $\{x_i\}_i$ is a minimal weight representative basis of $H^1(X)$ so that $|a_j''|=1$ and $|x_i|\geq S^1(X)$.
    Thus, $|x_i \otimes a_j''|\geq S^1(X)$. 
    Note that $|\mathcal{Z}| = \dim H^1(X) \dim A'' = \dim H^1(X) \dim H^1(Y)$ since $\dim(\ker(\partial_2^Y)^\intercal/\im(\partial_1^Y)^\intercal = \dim A''$.
    Now let $z = z'\oplus z'' \in\mathcal{X}^1$ with $z''=0$ and $z' = \sum_{i,j} z_{i,j}' x_i\otimes a_j''$ for $z_{i,j}\in\field_q$.
    It holds that $z\in\ker(\partial_2^\mathcal{X})^\intercal$ because $(\partial_2^X)^\intercal x_i = 0$ for all $i$.
    Also, we have $z\notin\im(\partial_1^\mathcal{X})^\intercal$ because $x_i\notin\im(\partial_1^X)$ and $a_j''\notin\im(\partial_1^Y)^\intercal$.
    So $z\in\ker(\partial_2^\mathcal{X})^\intercal\setminus\im(\partial_1^\mathcal{X})^\intercal$.
    It remains to show that $|z|\geq S^1(X)$.
    Since $\{x_i\}_i$ are minimal weight representative basis of $H^1(X)$, it holds that sum between any two elements is lower bounded by $S^1(X)$.
    Also since $\{a_j''\}_j$ is a basis of $A''$, adding any two basis elements increases the weight.
    Thus $|z| = |z'| = |\sum_{i,j} z_{i,j}' X_i\otimes a_j''| \geq S^1(X)$.
\end{proof}

By~\cref{lemma:HDX_complex_homology_systole}, the parameters $\stabcode{n}{k}{d}_q$ (for $d = \min\{d_X,d_Z\}$) of a qudit CSS code defined by $\mathcal{X}$ are given by
\begin{align}
    n &= \dim{\mathcal{X}_1} = \abs{X_1}\abs{A} + \abs{X_2}\abs{B}\,,
\end{align}
for the number of physical qudits,
\begin{align}
    k &= \dim{H_1(\mathcal{X})} = \dim{H_1(X)}\dim{H_1(Y)}
\end{align}
for the number of logical qudits, and
\begin{align}
    d_X &= S^1(\mathcal{X}) = S^1(X) \,, \\
    d_Z &= S_1(\mathcal{X}) = S_1(X)S_1(Y)
\end{align}
for the $X$- and $Z$-distances. The qudit HDX code is then defined by a (family of) chain complex $\mathcal{X}$ where chain complexes $X$ and $Y$ with nice properties are chosen to give the desired parameters.

\subsection{Simplicial complexes and Ramanujan complexes}\label{subsec:ramanujan-complexes}
We now define a $d$-dimensional simplicial complex, namely a length-$(d+1)$ chain complex $X$ with sets $\{X_k\}_{k=0}^{d+1}$, where the set $X_k$ for $k \in \{1, \ldots, d\}$ describe the set of $k$-simplices and the boundary map $\partial^X$ maps a $k$-simplex in $X_k$ to a set of $k-1$ simplices. Note that $X_0$ is the set of points, $X_1$ is the set of edges, $X_2$ is the set of triangles, $X_3$ is the set of tetrahedrons, and so on for higher-dimensional simplices. Thus, the boundary map $\partial^X$ maps a tetrahedron in $X_3$ to the four triangles in $X_2$ on its faces. Similarly, $\partial^X$ maps each triangle in $X_2$ to its three edges in $X_1$. Also, $\partial^X$ maps each edge in $X_1$ to the pair of vertices in $X_0$ that it connects. Formally, we can define a simplicial complex as follows.

\begin{definition}[{Simplicial chain complex}]\label{def:simplicial-complex}
    A length-$(d+1)$ chain complex 
    \begin{align}
        X : 0 \xrightarrow{\partial_{d+1}^X} X_d \xrightarrow{\partial_d^X} \ldots \xrightarrow{\partial_1^X} X_0 \xrightarrow{\partial_0^X} 0
    \end{align}
    is a $d$-dimensional simplicial chain complex if for all $k \in \{1, \ldots, d\}$ and any basis element $v \in X_k$ representing a $k$-simplex, the Hamming weight of its boundary is $\wt{\partial_k^X(v) = k+1}$ (i.e., the boundary is a linear combination of exactly $k+1$ distinct $(k-1)$-simplices with non-zero coefficients).
    The size of a simplicial complex $X^{(j)}$ is denoted by $\abs{X^{(j)}} := \abs{X_0^{(j)}}$, which is the number of vertices in $X^{(j)}$.
\end{definition}

Similar to the qubit HDX code in~\cite{evra2020}, the qudit HDX code is a family of CSS code defined by a sequence of 3-dimensional chain complexes $\mathcal{X}^{(1)}, \mathcal{X}^{(2)},\mathcal{X}^{(3)}, \ldots$ where $\mathcal{X}^{(j)} = ({X^{(j)}}^* \otimes {Y^{(j)}}^*)^*$ as is given in~\cref{eqn:HDX-cochain-complex} and where we take $\Tilde{X}^{(j)}$ as a subcomplex
\begin{align}
    \Tilde{X}^{(j)} : X^{(j)}_3 \xrightarrow{\partial_3^{(j)}} X^{(j)}_2 \xrightarrow{\partial_2^{(j)}} X^{(j)}_1
\end{align}
from a 3-dimensional simplicial complex
\begin{align}
    X^{(j)} : 0
    \xrightarrow{\partial_4^{(j)}} X^{(j)}_3 \xrightarrow{\partial_3^{(j)}} X^{(j)}_2 \xrightarrow{\partial_2^{(j)}} X^{(j)}_1 \xrightarrow{\partial_1^{(j)}} X^{(j)}_0 \xrightarrow{\partial_0^{(j)}} 0
\end{align}
for each $j$. On the other hand, we take $Y^{(1)},Y^{(2)}, \ldots$ to be a sequence of 2-chain complexes corresponding to a good classical LDPC code. Then, by picking a particular sequence of simplicial complexes $X^{(1)}, X^{(2)}, \ldots$, we obtain a family CSS codes $C^{(1)}, C^{(2)}, \ldots$ where CSS code $C^{(j)}$ is defined by the chain complex $\mathcal{X}^{(j)}$ with parameters $\stabcode{n{(j)}}{k{(j)}}{d_X{(j)},d_Z{(j)}}_q$ satisfying $\lim_{j\to\infty} n{(j)} = \infty$ and
\begin{equation}
    \label{eqn:hdx-parameters}
    \begin{aligned}
        k{(j)} &= \Theta(n{(j)})\,,\\
        d_X{(j)} &= \Omega(n{(j)})\,,\\
        d_Z{(j)} &= \Omega((\log n{(j)})^2)\,.
    \end{aligned}
\end{equation}
As in the case for the qubit HDX code~\cite{evra2020}, we take the sequence of simplicial complexes $X^{(1)}, X^{(2)}, \ldots$ to be \emph{Ramanujan complexes}~\cite{lubotzky2013,lubotzky2017} with certain properties such that the asymptotic values of the $\stabcode{n{(j)}}{k{(j)}}{d_X{(j)},d_Z{(j)}}_q$ of the $C^{(j)}$ parameters above are obtained. For the qudit HDX codes we also use the sequence of Ramanujan complexes described in~\cite{evra2020} (which in turn is based on~\cite[Theorem 1.1]{lubotzky2005},\cite[Section 2.3]{lubotzky2013}), which can be shown to satisfy $\lim_{j\to\infty} \abs{X^{(j)}} = \infty$ and have nontrivial first and second cohomology groups:
\begin{align}
    \dim H^1(X^{(j)}) &> 0\,,\\
    \dim H^2(X^{(j)}) &> 0\,,
\end{align}
as well as $k\numth$ systole and co-systole that are lower-bounded as~\cite[Theorem 5.11]{evra2020}
\begin{align}
    S_k(X^{(j)}) \geq c\lp\log_q\abs{X^{(j)}}\rp^k
\end{align}
and~\cite[Theorem 5.7]{evra2020}
\begin{align}
    S^k(X^{(j)}) \geq d \abs{X^{(j)}}\,,
\end{align}
respectively. Thus, by~\cref{lemma:HDX_complex_homology_systole} the parameters $\stabcode{n{(j)}}{k{(j)}}{d_X{(j)},d_Z{(j)}}_q$ of $C^{(j)}$ are
\begin{align}
    n{(j)} &= \abs{X_1^{(j)}}\abs{Y_1^{(j)}} + \abs{X_2^{(j)}}\abs{Y_0^{(j)}}\,,\label{eqn:HDX_param1}\\
    k{(j)} &= \dim H_1(\Tilde{X}^{(j)})\dim H_1(\Tilde{Y}^{(j)})\,,\label{eqn:HDX_param2}\\
    d_X{(j)} &= S^1(\Tilde{X}^{(j)})\,,\label{eqn:HDX_param3}\\
    d_Z{(j)} &= S_1(\Tilde{X}^{(j)}) S_1(Y^{(j)})\,.\label{eqn:HDX_param4}
\end{align}
\sloppy Thus, by substituting in the values of the dimension of the homology groups and the 1-systole and co-systole, we obtain the asymptotic values of the $C^{(j)}$ parameters $\stabcode{n{(j)}}{k{(j)}}{d_X{(j)},d_Z{(j)}}_q$ in~\cref{eqn:hdx-parameters}.

\section{Qudit fiber bundle codes}\label{sec:qudit-fiber-bundle-codes}
The first quantum LDPC codes to achieve a distance greater than $n^{1/2}\polylog(n)$ was introduced a few years ago and is based on fiber bundles~\cite{hastings2020}. For qubits, this code achieves a distance of $\Omega\lp n^{3/5}/\polylog(n) \rp$ and a code dimension of $\Tilde{\Theta}\lp n^{3/5} \rp$ for $n$ physical qubits. Now, the final LDPC codes that we quditize are the \emph{fiber bundle codes} to give \emph{qudit fiber bundle codes}. In the first part of this section, we produce the background and code properties from  the original work by Hastings, Haah, and O'Donnell~\cite{hastings2020}; we refer the reader to this reference for all of the details, where here we only reproduce what we need to find the code parameters of our qudit fiber bundle codes.

We start with a bit of mathematical background needed to understand this class of codes. Familiarity with the concepts introduced in~\cref{subsec:homological-algebra-primer} is helpful for this section. We first define the fundamental object these codes are based on, the \emph{fiber bundle}:

\begin{definition}[Fiber bundle]\label{def:fiber-bundle}
    A \emph{fiber bundle} is a structure that consists of:
    \begin{itemize}
        \item A base space $B$,
        \item A total space $E$,
        \item A fiber $F$, and
        \item A continuous surjection $\pi: E \to B$ that satisfies a local triviality condition.
    \end{itemize}
    The fiber bundle $(E, B, \pi, F)$ is often denoted as
    \begin{align}
        F \to E \xrightarrow{\pi} B\,.
    \end{align}
\end{definition}
A fiber bundle generalizes a product space. Stated differently, a fiber bundle can be thought of as a space that is locally a product space but globally may have a non-product-like topological structure. While a simple product (like a cylinder, which is a product of a circle and an interval) has the same structure everywhere, a fiber bundle allows for a global "twist." A common example of a fiber bundle is a M\"{o}bius strip, where locally it looks like the product of a circle with an interval, but globally there is a twist, where going around once reverses the interval. Fiber bundle codes exploit this by starting with a homological product (like the HGP codes discussed in~\cref{sec:qudit-hgp-codes}) and introducing intentional twists. This modification to the boundary operators significantly enhances the code's distance, which led to the first LDPC codes that surpass the $\sqrt{n}$ distance barrier~\cite{hastings2020}. In the construction of fiber bundle codes, we associate the fiber with a cycle graph and a random base with a random classical LDPC code. The main result that Hastings, Haah, and O'Donnell derive for qubit fiber bundle codes is given as follows:

\begin{fact}[Qubit fiber bundle codes~{\cite[Thm. 1.1]{hastings2020}}]\label{fact:qubit-fiber-bundle-codes}
    There exists a family of quantum codes on $n$ qubits with $d_X = \bigOmega{n^{1/2}\polylog(n)}$ and $d_Z = \bigOmega{n^{3/4}\polylog(n)}$, where all stabilizer generators have weight at most $\polylog(n)$ and all qubits participate in at most $\polylog(n)$ stabilizer generators. The code has $\bigTheta{n^{1/2}}$ logical qubits.
\end{fact}

Note that this code is \emph{not} LDPC, but it can be weight-reduced to LDPC at polylogarithmic cost in the distance and number of physical qubits. For qubit LDPC fiber bundle codes, we have the following corollary:

\begin{corollary}[Qubit LDPC fiber bundle codes~{\cite[Cor. 1.2]{hastings2020}}]\label{cor:qubit-ldpc-fiber-bundle-codes}
    There exists a family of quantum LDPC codes on $n$ qubits having distance $d = \bigOmega{n^{3/5}/\polylog(n)}$ and with $\bigOmega{n^{3/5}/\polylog(n)}$ logical qubits.
\end{corollary}

The fiber bundle code construction relies on homological algebra and chain complexes, which we introduced in~\cref{subsec:homological-algebra-primer}. Here, a chain complex defines a quantum code by picking an integer $\nu > 0$ and associating $\nu$-cells of the complex with qubits and $(\nu-1)$- and $(\nu+1)$-cells with $X$- and $Z$-stabilizer generators, respectively. $Z$ logical operators are then the $\nu\numth$ homology class and $X$ logical operators are the $\nu\numth$ cohomology class. $d_X$ is the lowest possible Hamming weight of a vector that represents a nontrivial $\nu\numth$ cohomology and $d_Z$ is likewise defined for the homology.

In order to derive these code parameters, we need to build up the mathematical background for fiber bundles, which we do here rather than in~\cref{subsec:homological-algebra-primer} because these notions are only relevant to the fiber bundle codes. We follow the presentation in~\cite{hastings2020} and reproduce the relevant definitions, results, etc., indicating where in the original paper these notions are. We start by reviewing the definition of a homological product. Given a base complex $B$ and a fiber complex $F$, we construct the product $E$, called a bundle, by taking tensor products of the component chain vector spaces:
\begin{equation}
    \begin{gathered}
    \xymatrix{
    \cdots \ar[d] &\ar[l] \cdots \ar[d]& B_j \otimes F_k \ar[d]^{\identity \otimes \partial_k}\ar[l]_{\partial_j \otimes \identity}\\
    B_0 \otimes F_1 \ar[d]^{\identity \otimes \partial_1} &\ar[l]_{\partial_1 \otimes \identity} B_1 \otimes F_1 \ar[d]^{\identity \otimes \partial_1}& \vdots \ar[l] \ar[d]\\
    B_0 \otimes F_0 &\ar[l]_{\partial_1 \otimes \identity } B_1 \otimes F_0 & \vdots \ar[l]
    }
    \end{gathered}
\end{equation}
The \emph{chain space} $E_\rho$ of the bundle is then the direct sum
\begin{align}
    E_\rho = \bigoplus_{\mu+\nu=\rho}{E_{\mu,\nu}}\,,
\end{align}
where $E_{\mu,\nu} = B_\mu \otimes F_\nu$, along the diagonal line $\mu+\nu=\rho$ for $\rho \geq 0$ in the diagram. The boundary map on each $E_{\mu,\nu}$ is defined as
\begin{align}
    \partial_\rho^{E}\eval_{(\mu,\nu)} = I \otimes \partial_\nu^{F} + \partial_\mu^{B} \otimes I
\end{align}
for vector spaces over $\field_2$. For general coefficient groups (i.e., for $\field_q$), we have
\begin{align}
    \partial_{(\mu,\nu)}^{E} = (-1)^\mu I \otimes \partial_\nu^{F} + \partial_\mu^{B} \otimes I\,,
\end{align}
that is, we pick up a minus sign as usual.

This formalism was for non-twisted bundles; we now move to twisted bundles. We assume that the fiber admits an automorphism group $G$, which is a collection of permutation actions on a set of $\nu$-cells for each $\nu$ such that the boundary operator commutes with the permutation:
\begin{align}
    g \partial f^\nu = \partial gf^\nu\ \forall g \in G, f^\nu \in F_\nu\,.\label{eqn:fiber-automorphism}
\end{align}

\begin{definition}[Twisted boundary map~{\cite[Def. 2.1]{hastings2020}}]\label{def:twisted-boundary-map}
    Given a fiber automorphism $G$ obeying~\cref{eqn:fiber-automorphism}, a \emph{connection} $\varphi$ of a bundle is an arbitrary assignment of an automorphism group element, a \emph{twist}, for each pair of a base cell and one of its boundary cell:
    \begin{align}
        \{(b,a) : b, a \text{ are cells such that } a \in \partial b\} \xrightarrow{\varphi} G\,.
    \end{align}
    We define a \emph{twisted boundary map} $\partial^{E}$ by $\varphi$:
    \begin{align}
        \partial_{(0,\nu)}^{E}(b^0 \otimes f) &= b^0 \otimes \partial f\,,\\
        \partial_{(1,\nu)}^{E}(b^1 \otimes f) &= b^1 \otimes \partial f + \defsum{a^0 \in \partial b^1}{}{a^0 \otimes \varphi(b^1, a^0)f}\,.
    \end{align}
    Again, for coefficients over $\field_q$, we replace $b^1 \otimes \partial f$ by $-b^1 \otimes \partial f$.
\end{definition}

The twisted boundary map can be generalized to any higher dimensional base complex, but we must have $\partial^{E}\partial^{E} = 0$. For example, if the base is a 2-complex, we may need an extra term:
\begin{align}
    \partial_{(2,\nu)}^{E}(b \otimes f) = b \otimes \partial f + \defsum{e \in \partial b}{}{e \otimes \varphi(b,e)f} + \defsum{\substack{v \in \partial b :\\e \in \partial b}}{}{v \otimes f_{v,e,b}^+}\,,
\end{align}
where $f_{v,e,b}^+$ is a $(\nu+1)$-cell of the fiber. However, we only consider 1-complexes and leave higher dimensional base complexes to future work.

We now consider isomorphisms on (co)homology groups. The first homology and cohomology, $H_1(E)$ and $H^1(E)$, respectively, of the bundle are isomorphic to $H_1(B)$ and  $H^1(B)$, the homology and cohomology of the base, respectively. This implies that $b_1(E) = b_1(B)$.

\begin{definition}[Bundle projection~{\cite[Def. 2.4]{hastings2020}}]\label{def:bundle-projection}
    The \emph{bundle projection} $\Pi_\rho : E_\rho \to B_\rho$ is defined as
    \begin{align}
        b^\rho \otimes f^0 &\mapsto b^\rho\,,\\
        b^{\rho-j} \otimes f^j &\mapsto 0 \text{ for } j > 0
    \end{align}
    for all $\rho$-cells $b^\rho$ and $(\rho-j)$-cells $b^{\rho-j}$ of the base and 0-cells $f^0$ and $j$-cells $f^j$ of the fiber. Furthermore, the bundle projection forms a chain map.
\end{definition}

With this definition in mind, we have the following result:

\begin{fact}[Qubit bundle projection properties~{\cite[Lemma 2.5]{hastings2020}}]\label{fact:qubit-bundle-projection-props}
    The bundle projection induces vector space isomorphisms $\Pi_* : H_1(E) \to H_1(B)$ and $\Pi^* : H^1(B) \to H^1(E)$ if all of the following hold:
    \begin{enumerate}[label=(\roman*),ref=(\roman*)]
        \item $B$ is a 1-complex
        \item The boundary $\partial f^1$ of any fiber 1-cell $f^1$ has even weight
        \item Every fiber 0-chain of even weight is a boundary
        \item $H_0(B) = 0$, that is, the zeroth Betti number vanishes
        \item Every fiber automorphism acts trivially on $H_1(F)$.
    \end{enumerate}
\end{fact}

See~\cite[Sec. 2.4]{hastings2020} for a proof of this lemma. With all of this background, we now discuss the basic features of the error correcting codes we get from fiber bundles:

\begin{definition}[Fiber bundle code~{\cite[Def. 2.11]{hastings2020}}]\label{def:fiber-bundle-code}
    The \emph{fiber bundle code} is a quantum CSS code whose logical operators are associated with homology and cohomology at dimension 1 of the twisted bundle complex $E_2 \to E_1 \to E_0$ built from the circle fiber $F_1 \to F_0$ and a base $B_1 \to B_0$.
\end{definition}

Twists can be arbitrary members of $\integers_{n_{F}}$, where we choose $n_{F} = m_{F} = \ell^2$ for some odd integer $\ell$. Then, we choose a random classical code for the base $B$ using $n_{B}$ bits (1-cells of $B$) and $m_{B}$ parity checks (0-cells). We represent the code using its Tanner graph, which is a bipartite graph $B$ with $m_{B}$ left-vertices and $n_{B}$ right-vertices. Choosing $m_{B} = \frac{3}{4}n_{B}$ and vertices with degree $\Delta = \bigTheta{\log^2(n_{B})}$ gives a minimum distance of $\bigOmega{n_{B}}$ and a set of linearly independent parity checks with high probability.

Thus, the bundle $E$ has $n = n_{B} \cdot m_{F} + m_{B} \cdot n_{F}$ 1-cells corresponding to qubits and the total number of cells is $(n_{B} + m_{B}) \cdot (n_{F} + m_{F})$. Thus, we have $\bigTheta{n_{B}}$ logical qubits and the distances are $d_X = \bigOmega{m_{F}/\log^2{n_{B}}}$ and $d_Z = \bigOmega{n_{B} \cdot m_{F}^{1/2} / \log^2{n_{B}}}$ for $n_{B} \geq m_{F}$. Choosing $n_{B} \sim m_{F}$ gives $n = \bigTheta{n_{B}^2}$, so we have $d_X = \bigOmega{n^{1/2}/\log^2{N}}$, $d_Z = \bigOmega{n^{3/4}/\log^2{n}}$, and $k = \bigTheta{n^{1/2}}$ (i.e., the number of logical qubits). After distance balancing, we have a $\stabcode{n}{\bigTheta{n^{1/2}}}{\bigOmega{n^{3/5}\polylog(n)}}$ fiber bundle code. This was, of course, a brief overview of the qubit fiber bundle codes to set the stage for their quditization. For all of the mathematical details, we refer the reader to~\cite{hastings2020}.

\subsection{Extension to qudits}\label{subsec:qudit-fiber-bundle-codes}
We now extend the above formalism to qudits and recover the number of physical qudits, number of logical qudits, and the distance. We have already done a lot of the heavy-lifting by developing the theory of fiber bundles and the codes we can derive from them. We start with a chain complex over $\field_q$ for some prime power $q = p^s$ for prime $p$ and integer $s \geq 1$:
\begin{align}
    \cdots \xrightarrow{\partial_{j+1}} \mathcal{A}_j \xrightarrow{\partial_j} \mathcal{A}_{j-1} \xrightarrow{\partial_{j-1}} \cdots \xrightarrow{\partial_1} \mathcal{A}_0 \xrightarrow{\partial_0 = 0} 0\,.
\end{align}
Here, each vector space is over $\field_q$ and the boundary maps $\partial_j : \mathcal{A}_j \to \mathcal{A}_{j-1}$ are matrices with elements from $\field_q$. The definition of the homology is the same:
\begin{align}
    H_j(\mathcal{A}) = \ker_{\field_q}{\partial_j} / \text{im}_{\field_q}{\partial_{j+1}}\,,
\end{align}
but now operations are done over $\field_q$, as indicated by the subscripts. The cohomology is extended over $\field_q$ analogously. The $j\numth$ Betti number is
\begin{align}
    b_j(\mathcal{A}) = \dim_{\field_q}{H_j(\mathcal{A})}\,.
\end{align}
The base complex $B$ and the fiber complex $F$ are now taken over $\field_q$ with boundary maps $\partial_j$ that have matrix elements from $\field_q$. The chain space $E_\rho$ of the bundle is then the direct sum of the product spaces of $B$ and $F$, as it was for qubits:
\begin{align}
    E_\rho = \bigoplus_{\rho=\mu+\nu}{E_{\mu,\nu}}\,,
\end{align}
where $E_{\mu,\nu} = B_\mu \otimes F_\nu$. With the chain spaces $E_\rho$ defined, we move on to the boundary maps $\partial^E$, starting with the untwisted boundary maps:
\begin{align}
    \partial_r^E\eval_{(\mu,\nu)} = (-1)^\mu I \otimes \partial_\nu^F + \partial_\mu^B \otimes I\,,
\end{align}
where the $(-1)^\mu$ ensures the relation
\begin{align}
    \partial_\rho^E\partial_{\rho+1}^E = 0\,,
\end{align}
which we can confirm with a quick calculation. Let $b^\mu \in B_\mu$ and $f^\nu \in F_\nu$. We want to show that
\begin{align}
    \partial_\rho^E \partial_{\rho+1}^E(b^\mu \otimes f^\nu) = 0\,,
\end{align}
where $\rho+1 = \mu+\nu$. First, we apply $\partial_{\rho+1}^E$:
\begin{align}
    \partial_{\rho+1}^E(b^\mu \otimes f^\nu) &= \lp (-1)^\mu I \otimes \partial_\nu^F+ \partial_\mu^B \otimes I \rp (b^\mu \otimes f^\nu)\\
    &= (-1)^\mu b^\mu \otimes \partial_\nu^F f^\nu + \partial_\mu^B b^\mu \otimes f^\nu\,.
\end{align}
Now, apply $\partial_\rho^E$:
\begin{align}
    \begin{split}
        \partial_\rho^E \partial_{\rho+1}^E(b^\mu \otimes f^\nu) = &(-1)^\mu b^\mu \otimes \partial_\nu^F f^\nu\\ &+ \partial_\mu^B b^\mu \otimes f^\nu\,.
    \end{split}
\end{align}
The first term gives
\begin{align}
    (-1)^\mu\lb (-1)^\mu b^\mu \otimes \underbrace{\partial_{\nu_1}^F(\partial_\nu^F f^\nu)}_{0} + \partial_\mu^B b^\mu \otimes (\partial_\nu^F f^\nu) \rb
\end{align}
and the second term gives
\begin{align}
    (-1)^{\mu-1}(\partial_\mu^B b^\mu) \otimes (\partial_\nu^F f^\nu) + \underbrace{\partial_{\mu-1}^B(\partial_\mu^B b^\mu)}_{0} \otimes f^\nu\,.
\end{align}
Combining everything, we have
\begin{align}
    \partial_\rho^E \partial_{\rho+1}^E(b^\mu \otimes f^\nu) &= (-1)^{\mu-1}(\partial_\mu^B b^\mu) \otimes (\partial_\nu^F f^\nu) + (-1)^\mu\partial_\mu^B b^\mu \otimes (\partial_q\nu F f^\nu)\\
    &= (-1)^\mu\lp \partial_\mu^B b^\mu \otimes \partial_\nu^F f^\nu- \partial_\mu^B b^\mu \otimes \partial_\nu^F f^\nu \rp\\
    &= 0\,,
\end{align}
which holds for all $\rho$. We can extend this to \emph{twisted bundles}, where we first assume that the fiber admits an automorphism group $G$ such that for each $\nu$,
\begin{align}
    g \partial f = \partial gf\ \forall g \in G\,,\ f \in F_\nu\,,
\end{align}
where arithmetic is done over $\field_q$. Then, a \emph{connection} $\varphi$ of a bundle is an assignment of an element of $G$ for each pair of a base cell and its boundary:
\begin{align}
    \{(b,a) : b,a \text{ are cells such that } a \in \partial b\} \xrightarrow{\varphi} G\,.
\end{align}
Then, a \emph{twisted boundary map} $\partial^E$ acts as
\begin{align}
    \partial^E_{(0,\nu)}(b^0 \otimes f) &= b^0 \otimes \partial f\\
    \intertext{and}
    \partial^E_{(1,\nu)}(b^1 \otimes f) &= -b^1 \otimes \partial f + \defsum{a^0 \in \partial b^1}{}{a^0 \otimes \varphi(b^1, a^0)f}\,,
\end{align}
where we note that the minus sign on the first term on the RHS of the second equation is to handle coefficients from $\field_q$; if our coefficient group were $\field_2$, this sign would not be needed, as we saw above. We can check, as we did above, that $\partial^E_\rho \partial^E_{\rho+1} = 0$ for all $\rho \geq 0$ for bases that are 1-complexes, but we do not show the math here. However, we note that more terms are needed in defining the action of the twisted boundary map for 2-complexes, 3-complexes, etc., so in this work we only focus on 1-complexes.

We will make use of the fact that the first homology $H_1(E)$ and cohomology $H^1(E)$ of the bundle are isomorphic to $H_1(B)$ and $H^1(B)$, respectively, such that the Betti numbers are equal: $b_1(E) = b_1(B)$. This isomorphism is induced by the \emph{bundle projection}, which we defined in~\cref{def:bundle-projection}. The definition for qubits holds for qudits, but we have a modified~\cref{fact:qubit-bundle-projection-props} for qudits:

\begin{proposition}[Qudit bundle projection properties]\label{prop:qudit-bundle-projection-props}
    The bundle projection induces vector space isomorphisms $\Pi_* : H_1(E) \to H_1(B)$ and $\Pi^* : H^1(B) \to H^1(E)$ if all of the following are true:
    \begin{enumerate}[label=(\roman*),ref=(\roman*)]
        \item\label{item:1-complex} $B$ is a 1-complex 
        \item\label{item:sum-zero} The sum of the coefficients of $\partial f^1$ of any fiber 1-cell $f^1$ is $0 \pmod{q}$
        \item\label{item:boundary} Every fiber 0-chain whose coefficients sum to $0 \pmod{q}$ is a boundary
        \item\label{item:betti-number} $H_0(B) = 0$, that is, the zeroth Betti number $b_0(B)$ vanishes
        \item\label{item:trivial} Every fiber automorphism acts trivially on $H_1(F)$
    \end{enumerate}
\end{proposition}

\begin{proof}
    Note that conditions~\cref{item:sum-zero,item:boundary} are the only conditions that differ between qubits and qudits. Nevertheless, we reprove the lemma for qudits, which involves generalizing the series of propositions to $\field_q$.

    \begin{proposition}[Well-defined]\label{prop:well-defined}
        Assume~\cref{item:1-complex} and~\cref{item:sum-zero} are true. Then, the induced map $\Pi_* : H_1(E) \to H_1(B)$ is well-defined.
    \end{proposition}

    \begin{proof}
        Following the proof in~\cite{hastings2020}, we need to show that (1) any closed 1-chain becomes closed and (2) any 1-chain that is a boundary becomes a boundary.

        (1) To show that a closed 1-chain becomes closed, let $h^1 + v^1$ be a closed 1-chain. That means that $\partial(h^1 + v^1) = \partial h^1 + \partial v^1 = 0$, where now arithmetic is done over the field $\field_q$ and $0 = 0 \pmod{q}$, though we exclude the $\pmod{q}$ part for easier readability. For a fiber 1-cell $f^1$,~\cref{item:sum-zero} means that the sum of the coefficients of $\partial f^1$ is 0. If we have $v^1 = \defsum{j}{}{a_j^0 \otimes f_j^1}$, then $\partial v^1 = \defsum{j}{}{a_j^0 \otimes \partial f_j^1}$. Then, applying the projection $\Pi$ to $\partial v^1$ simply multiplies the base complex elements $a_j^0$ by the sum of the coefficients of $\partial f^1$, which is 0. Thus, we have $\Pi \partial v^1 = 0$. Therefore,
        \begin{align}
            \partial \Pi h^1 &= \Pi \partial h^1= \Pi(-\partial v^1)= -\Pi \partial v^1= 0\,.
        \end{align}
        To prove (2), we need to show that the boundary becomes a boundary, where it suffices to check $(1,1)$-cells:
        \begin{align}
            \partial(b^1 \otimes f^1) = -b^1 \otimes \partial f^1 + \defsum{a^0 \in \partial b^1}{}{a^0 \otimes \varphi(b^1,a^0)f^1}\,.
        \end{align}
        We then just apply the projection $\Pi$ to this:
        \begin{align}
            \Pi \partial(b^1 \otimes f^1) = \Pi\lp-b^1 \otimes \partial f^1 + \defsum{a^0 \in \partial b^1}{}{a^0 \otimes \varphi(b^1,a^0)f^1}\rp
        \end{align}
        The second term (the sum) vanishes by the definition of $\Pi$ in~\cref{def:bundle-projection} and the first term maps to $b^1$ times the sum of the coefficients of $\partial f^1$, which, recall, is 0, so $\Pi \partial(b^1 \otimes f^1) = 0$.
    \end{proof}

    \begin{proposition}[Onto]\label{prop:onto}
        Assume that~\cref{item:1-complex,item:sum-zero,item:boundary} are true. Then, $\Pi_* : H_1(E) \to H_1(B)$ is onto.
    \end{proposition}

    \begin{proof}
        We again generalize the proof in~\cite{hastings2020}. Given a base cycle $b^1$, we choose a fiber 0-cell $f^0$ and consider $b^1 \otimes f^0$. Its boundary is
        \begin{align}
            \partial(b^1 \otimes f^0) = \defsum{a^0 \in \partial b^1}{}{a^0 \otimes \varphi(b^1,a^0)f^0}\,.
        \end{align}
        Since $b^1$ is a cycle, $\partial b^1 = 0$ and so each $a_j^0 \in \partial b^1$ is also 0. That handles the base part. Now, we address the fiber part. Denote the fiber-chain over $a^0_j$ as $f'_j = \varphi(b^1,a^0)f^0$. We need to show that $f'_j$ is a boundary, which means its coefficients sum to 0:
        \begin{align}
            \Pi\partial(b^1 \otimes f^0) &= \defsum{a^0 \in \partial b^1}{}{\Pi(a^0 \otimes \varphi(b^1,a^0)f^0)}\\
            &= \defsum{j}{}{\abs{f'_j}a_j^0}= 0\,.
        \end{align}
        Thus, we must have $\abs{f'_j} = 0$. Also, $f'_j = \partial s'_j$ for a fiber 1-chain $s'_j$ and $b^1 \otimes f^0 + \defsum{j}{}{a_j^0 \otimes s'_j}$ is closed and projects to $b^1$.
    \end{proof}

    \begin{proposition}[One-to-one]\label{prop:one-to-one}
        Assume~\cref{item:1-complex,item:sum-zero,item:boundary,item:betti-number,item:trivial} are true. Then, $\Pi_* : H_1(E) \to H_1(B)$ is one-to-one.
    \end{proposition}

    \begin{proof}
        The proof in~\cite{hastings2020} extends trivially to working over the field $\field_q$.
    \end{proof}

    Thus, the series of proofs of the propositions concludes the proof of~\cref{prop:qudit-bundle-projection-props}.
\end{proof}

We now discuss the choice of circle bundle over classical codes. The fiber can remain a cycle graph (i.e., a circle), which is a 1-complex with $m_F$ 0-cells and $n_F$ 1-cells with $m_F = n_F > 1$. The only difference between the case for qubits and that of qudits is that the qudit cells and their connections are defined over $\field_q$. For qubits, the circle admits an automorphism group that is the dihedral group of order $2n_F$ and~\cite{hastings2020} uses only the rotational symmetry. Over $\field_q$, the automorphism group generalizes to a linear transformation on the fiber vector spaces over $\field_q$ while preserving the cycle structure, but possibly changing the coefficients. This circle construction still obeys the conditions of~\cref{prop:qudit-bundle-projection-props}, as long as the changes mentioned in the proof are taken into account (e.g., the sum of the coefficients of $\partial f^1$ equals 0). We keep with the convention in~\cite{hastings2020} and choose $n_F = m_F = \ell^2$ for some integer $\ell$.

For the base complex, we now choose a random classical LDPC code over $\field_q$, rather than a binary code, which can still be represented by a Tanner graph. The Tanner graph is a bipartite graph with $m_B$ left-vertices and $n_B$ right-vertices, where we choose $m_B = \frac{3}{4}n_B$ and a vertex degree of $\Delta = \bigTheta{\log^2{n_B}}$ as in~\cite{hastings2020}. With high probability, a random classical non-binary LDPC code will have its parity-checks be linearly independent, which is equivalent to $H_0(B) = 0$. With all of this formalism established, we can give the qudit fiber bundle code parameters:

\begin{lemma}[Qudit fiber bundle code parameters]\label{lemma:qudit-fiber-bundle-code-parameters}
    For a prime power $q = p^s$ for prime $p$ and $s \geq 1$, a \emph{qudit fiber bundle code} is a quantum LDPC code with parameters $\stabcode{n}{k}{d}_q$, where
    \begin{itemize}
        \item $n = n_B \cdot m_F + n_F \cdot m_B$
        \item $k = \bigTheta{n_B}$
        \item $d = \min\{d_X,d_Z\}$,
    \end{itemize}
    where $d_X = \bigOmega{m_F/\polylog(n_B)}$ and $d_Z = \bigOmega{n_B \cdot m_F^{1/2}/\polylog(n_B)}$ and where $m_B$, $m_F$, $n_B$, and $n_F$ are defined above. Choosing $n_B \sim m_B$ gives $n = \bigTheta{n_B^2}$, which simplifies the code dimension to $k = \bigTheta{n^{1/2}}$ and the distances to $d_X = \bigOmega{n^{1/2}/\polylog(n)}$ and $d_Z = \bigOmega{n^{3/4}/\polylog(n)}$.
\end{lemma}

\begin{proof}
    The number of physical qudits $n$ and the number of logical qudits $k$ follow directly from the analysis in~\cite{hastings2020}. The scaling of the distance is the same as for qubits (extending to the field $\field_q$ does not affect the scaling), but the specific constants will be different.
\end{proof}

Since we only care about the scaling of the distance with the number of physical qudits, we do not do a detailed analysis of the constants in this work, which would require a significant reworking of the proofs in~\cite{hastings2020} and would add several pages to this work. An easy alternative is to calculate the distance numerically, as was done for several of the other code constructions we presented above.

\section{Discussion}\label{sec:discussion}
In this paper, we introduced a framework for generalizing quantum LDPC codes from qubits to qudits. We gave the formalism for several promising LDPC codes, namely bivariate bicycle codes, hypergraph product codes, subsystem hypergraph product simplex codes, high-dimensional expander codes, and fiber bundle codes. Furthermore, our numerical code searches and decoding simulations introduce several novel qudit error correcting codes. For the decoding of our qudit codes, we relied on mixed-integer programming, which, while accurate, is relatively slow. \revA{As such, we limited our studies to the code capacity, while future work could tackle the more challenging problem of decoding measurement or circuit-level errors. } The code is available on GitHub in Ref.~\cite{haug2025qudit}. For future work, one could find faster decoders tailored to qudits, such as one based on belief propagation~\cite{roffe2020,panteleev2021}, neural networks~\cite{torlai2017neural,fosel2018reinforcement}, or \textsc{MaxSAT}~\cite{berent2023,noormandipour2024}. We also leave to future work the details of logical gates and operators in the qudit regime, which can be done, for example, using code automorphisms~\cite{sayginel2024}.

We hope this work serves as a reference for researchers working on qudit error correction. The set of codes considered in this work is by no means exhaustive and there are many quantum error correcting codes that remain to be generalized to qudits. We also leave to future work more numerical analysis and searching for codes of the types that we introduced here, especially the high-dimensional expander and fiber bundle codes, for which we only gave theoretical results.

\begin{acknowledgments}
    The authors thank Xiaozhen Fu for helpful discussions. D.J.S.~acknowledges funding support from a Graduate Research Fellowship from the Joint Quantum Institute (JQI) at the University of Maryland, College Park. D.J.S.~also acknowledges support by ARL (W911NF-24-2-0107). A.T.~is supported by a Centre for Quantum Technologies (CQT) PhD scholarship at the National University of Singapore, the Google PhD Fellowship, and the CQT Young Researcher Career Development Grant. K.B.~is supported by Q.InC Strategic Research and Translational Thrust.
\end{acknowledgments}

\bibliographystyle{plainnat}
\bibliography{references}
\newpage
\appendix

\section{Primer on qudit error correction}\label{app-sec:qudit-ec-primer}
In this Appendix, we recall some useful definitions from the quantum error correction literature~\cite{gottesman2024}. All definitions below are for qubits until we specify that we are working with qudits. First, recall the definition of the Pauli group for qubits:

\begin{definition}[Qubit Pauli group]\label{def:pauli-group}
    The $n$-qubit \emph{Pauli group} $\mathcal{P}_n$ is the group of size $4^{n+1}$ whose elements are the tensor products of the identity matrix $I$ and the Pauli matrices $X$, $Y$, and $Z$, as well as an overall phase $\pm 1, \pm i$.
\end{definition}

We now define what a \emph{stabilizer} is:

\begin{definition}[Stabilizer~{\cite[Def. 3.3]{gottesman2024}}]\label{def:stabilizer}
    Given a quantum error correcting code $T \subseteq \mathcal{H}$ for some $n$-qubit Hilbert space $\mathcal{H}$, the \emph{stabilizer} $\mathsf{S}(T)$ is defined as
    \begin{align}
        \mathsf{S}(T) = \{M \in \mathcal{P}_n : M\ket{\psi} = \ket{\psi}\ \forall \ket{\psi} \in T\}\,.
    \end{align}
\end{definition}

Then, a \emph{stabilizer code} is defined as

\begin{definition}[Stabilizer code~{\cite[Def. 3.5]{gottesman2024}}]\label{def:stabilizer-code}
    If $T$ is a quantum error correcting code such that $T = \mathcal{T}(\mathsf{S}(T))$, where
    \begin{align}
        \mathcal{T}(\mathsf{S}) = \{\ket{\psi} : M\ket{\psi} = \ket{\psi}\ \forall M \in \mathsf{S}\},
    \end{align}
    then $T$ is a \emph{stabilizer code}.
\end{definition}

That is, if we are given the stabilizers, then to find the code space (and the code words), all we have to do is find the subspace of the total Hilbert space that corresponds to the +1 eigenspace of the stabilizers. With these definitions in mind, we can define a \emph{CSS code}, named after Calderbank, Steane, and Shor~\cite{calderbank1996good,steane1996}:

\begin{definition}[CSS code]\label{def:css-code}
    A \emph{CSS code} is a type of stabilizer code that is constructed from two classical, binary, linear codes $C_X$ and $C_Z$ of length $n$, where $C_X^\perp \subseteq C_Z$. The generators for a CSS codes are composed of either a tensor product of all $X$ and $I$ or all $Z$ and $I$.
\end{definition}

In this work, we focus exclusively on \emph{quantum low-density parity-check} (LDPC) codes, which we define as:

\begin{definition}[Quantum low-density parity-check codes]\label{def:qldpc}
    A \emph{quantum low-density parity-check} code is a quantum error correcting code whose stabilizers all have low weight (i.e., less than some constant $w \in \naturals$), where, the \emph{weight} of an operator is the number of non-identity operators in its tensor product representation.
\end{definition}

\subsection{Qudit error correction}\label{app-subsec:qudit-error-correction}
We now review the important notions of the theory of error correction for \emph{qudits}, the $q$-dimensional generalization of qubits where we now swap physical qubits, usually taken to be two-level systems, for physical qudits, $q$-level systems. One benefit of qudit codes is that they can correct more errors and send quantum data at a higher rate than their qubit counterparts. In this work, we take $q$ to be a prime power, that is, $q = p^m$ for prime $p$ and $m \in \naturals$. As such, the qudits that we work with inherit the mathematical structure afforded by \emph{finite} or \emph{Galois} fields, so we refer to these qudits as \emph{Galois qudits}. We start by defining the Galois qudit Pauli group:

\begin{definition}[{Galois qudit Pauli group~\cite[Section 8.1.2]{gottesman2024},\cite{ashikhmin2001,ketkar2006,grassl2003}}]
    For prime $p$ and $m \in \naturals$, consider the Galois field $\field_q$, where $q = p^m$, and the $q$-dimensional Hilbert space $\mathcal{H}$ with orthonormal basis $\ket{\alpha} : \alpha \in \field_q\}$. The $X$ and $Z$ Galois-qudit Pauli operators are defined by
    \begin{align}
        X^\alpha\ket{\beta} = \ket{\beta + \alpha} \quad\text{and}\quad
        Z^\alpha\ket{\beta} = \omega^{\Tr[\alpha\beta]}\ket{\beta}
    \end{align}
    for all $\alpha, \beta \in \field_q$ and $\omega = e^{2i\pi/p}$ being the $p\numth$ root of unity. Here, the trace function $\Tr: \field_q \to \field_p$ on $\alpha \in \field_q$ is given by
    \begin{align}
        \Tr[\alpha] = \sum_{j=0}^{m-1} \alpha^{p^j} \,.
    \end{align}
    For $\alpha,\beta \in \field_q$, the commutation relation between $X$ and $Z$ is given by
    \begin{align}
        Z^\alpha X^\beta = \omega^{\Tr[\alpha\beta]} X^\beta Z^\alpha \,.
    \end{align}
    For a positive integer $n$ and odd prime-power $q$, the $n$-Galois-qudit Pauli group $\mathcal{P}_n(q)$ is then given by    
    \begin{align}
        \mathcal{P}_n(q) = \left\{\omega^\gamma \bigotimes_{j=1}^n Z^\alpha X^\beta : \gamma \in \field_p, \alpha, \beta \in \field_q\right\}\,.   
    \end{align}
    For even prime-power $q$, the $n$-Galois-qudit Pauli group $\Tilde{\mathcal{P}}_n(q)$ is given by $\Tilde{\mathcal{P}}_n(q) = \{\eta P : \eta \in \{1,i\}, P \in \mathcal{P}_n(q)\}$, that is, its elements are defined in the same way as for odd prime-power $q$ but with an additional $\eta \in \{1,i\}$ factor.
\end{definition}

With this definition, we can now define qudit stabilizers:

\begin{definition}[Qudit stabilizer~{\cite[Def. 8.6]{gottesman2024}}]\label{def:qudit-stabilizer}
    Let $q = p^m$ be a prime power for prime $p$, $\mathcal{P} = \mathcal{P}_n(q)$ be the $n$-qudit Pauli group of dimension $q$, and $Q$ be a subspace of the Hilbert space $\mathcal{H}_q^{\otimes n}$. The $\field_q$ \emph{stabilizer of $Q$} is the set
    \begin{align}
        \mathsf{S}(Q) = \{M \in \mathcal{P} : \ket{\psi} \text{ is an eigenvector of $M$ with eigenvalue +1 } \forall \ket{\psi} \in Q\}.
    \end{align}
    If $\mathsf{S}$ is a subgroup of $\mathcal{P}$, then we say that $\mathsf{S}$ is an $\field_q$ \emph{stabilizer} if it is Abelian and if $e^{i\phi}I \notin \mathsf{S}$ for any phase $\phi \neq 0$. The \emph{code space} of an $\field_q$ stabilizer $\mathsf{S}$ is the subspace
    \begin{align}
        \mathcal{T}(\mathsf{S}) = \{\ket{\psi} : M\ket{\psi} = \ket{\psi}\ \forall M \in \mathsf{S}\}.
    \end{align}
    Finally, the code $Q$ is an $\field_q$ \emph{stabilizer code} if and only if $Q = \mathcal{T}(\mathsf{S}(Q))$. The \emph{normalizer} $\mathsf{N}(\mathsf{S})$ of the stabilizer $\mathsf{S}$ is
    \begin{align}
        \mathsf{N}(\mathsf{S}) = \{N \in \mathcal{P} : NM = MN\ \forall M \in \mathsf{S}\}.
    \end{align}
\end{definition}

With these terms defined, we have the following:

\begin{fact}[$\field_p$ stabilizer code properties]\label{fact:gf(p)-stabilizer-code-properties}
    Let $p$ be prime and let $\mathsf{S}$ be a $\field_p$ stabilizer for the code $\mathcal{T}(\mathsf{S})$, which has $n$ physical qudits. If $\abs{\mathsf{S}} = p^r$ (i.e., $\mathsf{S}$ has $r$ generators), then $\dim\mathcal{T}(\mathsf{S}) = p^{n-r}$ such that $\mathcal{T}(\mathsf{S})$ encodes $k = n - r$ logical qudits. The set of undetectable errors for $\mathsf{S}$ is $\hat{\mathsf{N}}(\mathsf{S}) \backslash \hat{\mathsf{S}}$ and the distance of $\mathsf{S}$ is $\min\{\text{wt}(E) : E \in \hat{\mathsf{N}}(\mathsf{S}) \backslash \hat{\mathsf{S}}\}$.
\end{fact}

In analogy with qubit stabilizer codes, we denote such a $q$-dimensional qudit stabilizer code with $n$ physical qudits, $k$ logical qudits, and distance $d$ as $\stabcode{n}{k}{d}_q$, where the only difference between this notation and that for qubits is the addition of the subscript $q$. We can also get qudit CSS codes. For prime dimension, the stabilizer code takes the form
\begin{align}
    \lp
    \begin{array}{c|c}
        0 & H_1\\
        H_2 & 0
    \end{array}
    \rp,
\end{align}
where $H_1$ and $H_2$ are the parity check operators for linear classical codes $C_1$ and $C_2$, respectively. We denote these classical codes $C_1$ and $C_2$ as $[n_1, k_1, d_1]$ and $[n_2, k_2, d_2]$ codes, respectively, and we define the \emph{dual code} of a classical linear code as

\begin{definition}[Dual code]\label{def:dual-code}
    If $C$ is a classical linear code, then the \emph{dual code} $C^\perp$ is
    \begin{align}
        C^\perp = \{y \in \field^n_q : x \cdot y = 0\ \forall x \in C\}.
    \end{align}
\end{definition}

Many of the same properties for qubit CSS codes hold for qudit CSS codes. For example, if $x \in C_2^\perp$ and $z \in C_1^\perp$, then $x \cdot z = 0$, so we must have that $C_1^\perp \subseteq C_2$. The formulas to find the number of encoded qudits and the distance are also the same as those for qubits, namely
\begin{align}
    k &= k_1 + k_2 - n\\
    d &\geq \min\{d_1, d_2\}.
\end{align}
One difference between the qudit and qubit cases are that we have to do addition modulo $q$ rather than modulo 2 for the commutation relations. For qudit CSS code construction, we first need the definition of a true $\field_q$ stabilizer:

\begin{definition}[True $\field_q$ stabilizer]\label{def:true-fq-stabilizer}
    Let $P \in \mathcal{P}_n(q)$ have $\field_q$ symplectic representation $(x_P | z_P)$ and let $\mathsf{S}$ be a $\field_q$ stabilizer with symplectic representation $\hat{\mathsf{S}}$. Then, $\mathsf{S}$ is a \emph{true $\field_q$ stabilizer} if $(x_P | z_P) \in \hat{\mathsf{S}}$ implies $(\gamma x_P | \gamma z_P) \in \hat{\mathsf{S}}$ as well for all $\gamma \in \field_q$.
\end{definition}

Then, we have the following:

\begin{fact}[Qudit CSS construction~{\cite[Thm. 8.7]{gottesman2024}}]\label{fact:qudit-css}
    Let $C_1$ and $C_2$ be two classical linear codes over $\field_q$ with parameters $[n,k_1,d_1]_q$ and $[n,k_2,d_2]_q$ and satisfying $C_1^\perp \subseteq C_2$. Then, there exists a true $\field_q$ stabilizer code with stabilizer given as the smallest group containing all $\bigotimes_j{Z^{\gamma_j}}$ and $\bigotimes_j{X^{\eta_j}}$, where $\gamma$ runs over elements of $C_1^\perp$ and $\eta$ runs over elements of $C_2^\perp$ and with parameters $\stabcode{n}{k_1+k_2-n}{d}_q$ with $d \geq \min\{d_1,d_2\}$.
\end{fact}

We refer the reader to~\cite[Thm. 8.7]{gottesman2024} for a proof of this statement.

\end{document}